\documentclass{article}
\usepackage{mathrsfs}
\usepackage{verbatim}

\usepackage{geometry}
\usepackage{amsfonts}
\usepackage{graphicx,url}
\usepackage{xcolor}
\usepackage{amsmath, amsthm,amssymb}
\usepackage{vmargin}
\usepackage{cjhebrew}
\usepackage{subfigure}
\usepackage[toc,page]{appendix}

\newtheorem{theorem}{Theorem}[section]
\newtheorem{proposition}[theorem]{Proposition}
\newtheorem{lemma}[theorem]{Lemma}
\newtheorem{corollary}[theorem]{Corollary}
\theoremstyle{definition}
\newtheorem{definition}[theorem]{Definition}

\newtheorem{remark}[theorem]{Remark}

\newcommand{\ind}{1\hspace{-.27em}\mbox{\rm l}}

\newcommand{\E}{\mathbb{E}}

\newcommand{\R}{\mathbb{R}}

\allowdisplaybreaks
\numberwithin{equation}{section}

\setmarginsrb{10mm}{10mm}{10mm}{10mm}{0mm}{10mm}{0mm}{10mm}

\allowdisplaybreaks

\begin{document}

\title{On the distribution of the critical values
of random spherical harmonics}
\author{Valentina Cammarota \\
Department of Mathematics, University of Rome Tor Vergata \and Domenico
Marinucci \\
Department of Mathematics, University of Rome Tor Vergata \and Igor Wigman \\
Department of Mathematics, King's College London}
\maketitle

\abstract{ We study the limiting distribution of critical points and
extrema of random spherical harmonics, in the high energy limit.
In particular, we first derive the density
functions of extrema and saddles; we then provide analytic
expressions for the variances and we show that the empirical
measures in the high-energy limits converge weakly to their
expected values. Our arguments require a careful investigation of
the validity of the Kac-Rice formula in nonstandard circumstances,
entailing degeneracies of covariance matrices for first and second
derivatives of the processes being analyzed.}

\begin{itemize}
\item \textbf{Keywords and Phrases: }Spherical Harmonics, Critical Points,
Kac-Rice Formula, Legendre Polynomials, Hilb's Asymptotics

\item \textbf{AMS Classification:} 60G60, 60D05, 60B10, 33C55, 42C10
\end{itemize}

\section{Introduction}

\subsection{Statement of the main results}

The purpose of this paper is to investigate the asymptotic
distribution of critical values for Gaussian spherical harmonics, in
the high energy (Laplace eigenvalue) limit.
Below we will find explicit expressions for the density of extreme values and saddles; more
importantly, we will also find functional form for
their variances; by means of the above we will study the
convergence of the empirical distributions of extrema to their
limiting expressions. Some motivating
applications are discussed below; first let us introduce our
models and results more formally.

It is well-known that the eigenvalues $\lambda$ of the Laplace equation
\begin{equation*}
\Delta _{{\cal S}^2}f+\lambda f=0
\end{equation*}
on the two-dimensional sphere ${\cal S}^2$, are of the form
$\lambda=\lambda_{\ell}=\ell(\ell+1)$ for some integer $\ell\ge 1$. For any given eigenvalue $\lambda_{\ell}$, the corresponding eigenspace is the
$(2 \ell+1)$-dimensional space ${\cal L}_{\ell}$ of spherical harmonics of degree $\ell$; we can choose an arbitrary $L^{2}$-orthonormal
basis $\left\{ Y_{\mathbb{\ell }m}(.)\right\} _{m=-\ell, \dots, \ell }$, and consider random eigenfunctions of the form
\begin{equation*}
f_{\ell }(x)=\frac{1}{\sqrt{2\ell +1}}\sum_{m=-\ell }^{\ell }a_{\ell
m}Y_{\ell m}(x),
\end{equation*}
where the coefficients $\left\{ a_{\mathbb{\ell }m}\right\}$ are
independent, standard Gaussian variables; this is invariant w.r.t. the choice of $\{Y_{\ell m}\}$. The random fields $\{f_{\ell }(x), \; x\in {\cal S}^2\}$ are isotropic, meaning that
the probability laws of $f_{\ell }(\cdot )$ and $f_{\ell }^{g}(\cdot
):=f_{\ell }(g\cdot )$ are the same for any rotation $g\in SO(3)$. Also, $f_{\ell }$ are centred Gaussian, and from the addition theorem for spherical harmonics (see \cite{AAR} theorem 9.6.3) the
covariance function is given by,
\begin{equation*}
\mathbb{E}[f_{\ell }(x)f_{\ell }(y)]=P_{\ell }(\cos d(x,y)),
\end{equation*}%
where $P_{\ell }$ are the usual Legendre polynomials,
$\cos d(x,y)=\cos \theta _{x}\cos \theta _{y}+\sin \theta _{x}\sin
\theta _{y}\cos (\varphi _{x}-\varphi _{y})$ is the spherical geodesic
distance between $x$ and $y$, $\theta \in \lbrack 0,\pi ]$, $\varphi \in
\lbrack 0,2\pi )$ are standard spherical coordinates and $(\theta
_{x},\varphi _{x})$, $(\theta _{y},\varphi _{y})$ are the spherical
coordinates of $x$ and $y$ respectively.

Let $I\subseteq \mathbb{R}$ be any interval in the real line; we are
interested in the number of critical points, extrema and saddles of $f_{\ell
}$ with value in $I$:
\begin{equation*}
\mathcal{N}^{c}(f_\ell;I )=\mathcal{N}_{I}^{c}(f_{\ell })=\#\{x\in {\cal S}^2
:f_{\ell }(x)\in I,\nabla f_{\ell }(x)=0\},
\end{equation*}%
\begin{equation*}
\mathcal{N}^{e}(f_\ell;I)=\mathcal{N}_{I}^{e}(f_{\ell })=\#\{x\in {\cal S}^2
:f_{\ell }(x)\in I,\nabla f_{\ell }(x)=0,\text{det}(\nabla ^{2}f_{\ell
}(x))>0\},
\end{equation*}%
\begin{equation*}
\mathcal{N}^{s}(f_\ell;I)=\mathcal{N}_{I}^{s}(f_{\ell })=\#\{x\in {\cal S}^2
:f_{\ell }(x)\in I,\nabla f_{\ell }(x)=0,\text{det}(\nabla ^{2}f_{\ell
}(x))<0\}.
\end{equation*}%
We use $a=c,e,s$ to denote critical points, extrema and saddles respectively;
it is obvious that for all $I$ we have a.s. $$\mathcal{N}_{I}^{c}(f_{\ell })=\mathcal{N}
_{I}^{e}(f_{\ell })+\mathcal{N}_{I}^{s}(f_{\ell }).$$ For a nice domain ${\cal D} \subseteq {\cal S}^2$ we introduce
\begin{equation*}
\mathcal{N}^c (f_\ell; {\cal D},I )=\#\{x\in {\cal D}
:f_\ell (x) \in I,\nabla f_{\ell }(x)=0\}.
\end{equation*}
Our first theorem
gives the asymptotic behaviour for the expected number of critical points
of $f_{\ell }$ with values lying in $I$. Let us introduce the density
functions
\begin{equation*}
\pi _{1}^{c}(t)=\frac{\sqrt{3}}{\sqrt{8\pi }}(2e^{-t^{2}}+t^{2}-1)e^{-\frac{
t^{2}}{2}},
\end{equation*}
\begin{equation*}
\pi _{1}^{e}(t)=\frac{\sqrt{3}}{\sqrt{2\pi }}(e^{-t^{2}}+t^{2}-1)e^{-\frac{
t^{2}}{2}},
\end{equation*}
\begin{equation*}
\pi _{1}^{s}(t)=\frac{\sqrt{3}}{\sqrt{2\pi }}e^{-\frac{3}{2}t^{2}}.
\end{equation*}
We have the following:

\begin{proposition}
\label{expectation copy(1)} For every interval $I\subseteq \R$ we have as $\ell \rightarrow \infty$
\begin{equation*}
\mathbb{E}[\mathcal{N}_{I}^{c}(f_{\ell })] =\frac{2}{\sqrt{3}} \ell^2 \int_{I}\pi _{1}^{c}(t)dt+O(1),
\end{equation*}
and
\begin{equation*}
\mathbb{E}[\mathcal{N}_{I}^{a}(f_{\ell })] =\frac{\ell^2}{\sqrt{3}}  \int_{I}\pi _{1}^{a}(t)dt+O(1),
\end{equation*}
for $a=e,s$. The constant in the $O(\cdot)$ term is universal.
\end{proposition}

The results above were confirmed with great accuracy by numerical simulations \cite{fantaye}
to be published; the total number of critical points (i.e. the special case $I={\mathbb{R}}$) was addressed in
\cite{Nicolaescu}.
Indeed it is immediate from Proposition \ref{expectation copy(1)} that we have\footnote{Our leading constant for
$\mathbb{E}[\mathcal{N}_{\mathbb{R}}^{c}(f_{\ell })]$, supported by numerics
~\cite{fantaye}, differs from ~\cite{Nicolaescu}
by a factor of $\sqrt{8}$.}
\begin{equation*}
\mathbb{E}[\mathcal{N}_{\mathbb{R}}^{c}(f_{\ell })]=\frac{2 }{\sqrt{3}
} \ell^{2} +O(1),\hspace{0.5cm} \mathbb{E}[\mathcal{N}_{\mathbb{R}}^{e}(f_{\ell })]=
\frac{\ell ^{2}}{\sqrt{3}}+O(1), \hspace{0.5cm} \mathbb{E}[\mathcal{N}_{\mathbb{R}%
}^{s}(f_{\ell })]=\frac{\ell ^{2}}{\sqrt{3}}+O(1).
\end{equation*}%
The density functions for critical points, extrema and saddles are
plotted in Figure \ref{fig1} and Figure \ref{fig11}. The
distribution of saddle points is Gaussian with zero mean, whereas
the extrema are bimodal; since for $f_{\ell}$ all the maxima (resp., minima) are necessarily
positive, this also holds in the limit; the unique peak of their density is located
approximately at $\pm 1.685$\ldots.

\begin{figure}[h]
\centering {\includegraphics[width=7cm, height=5cm]{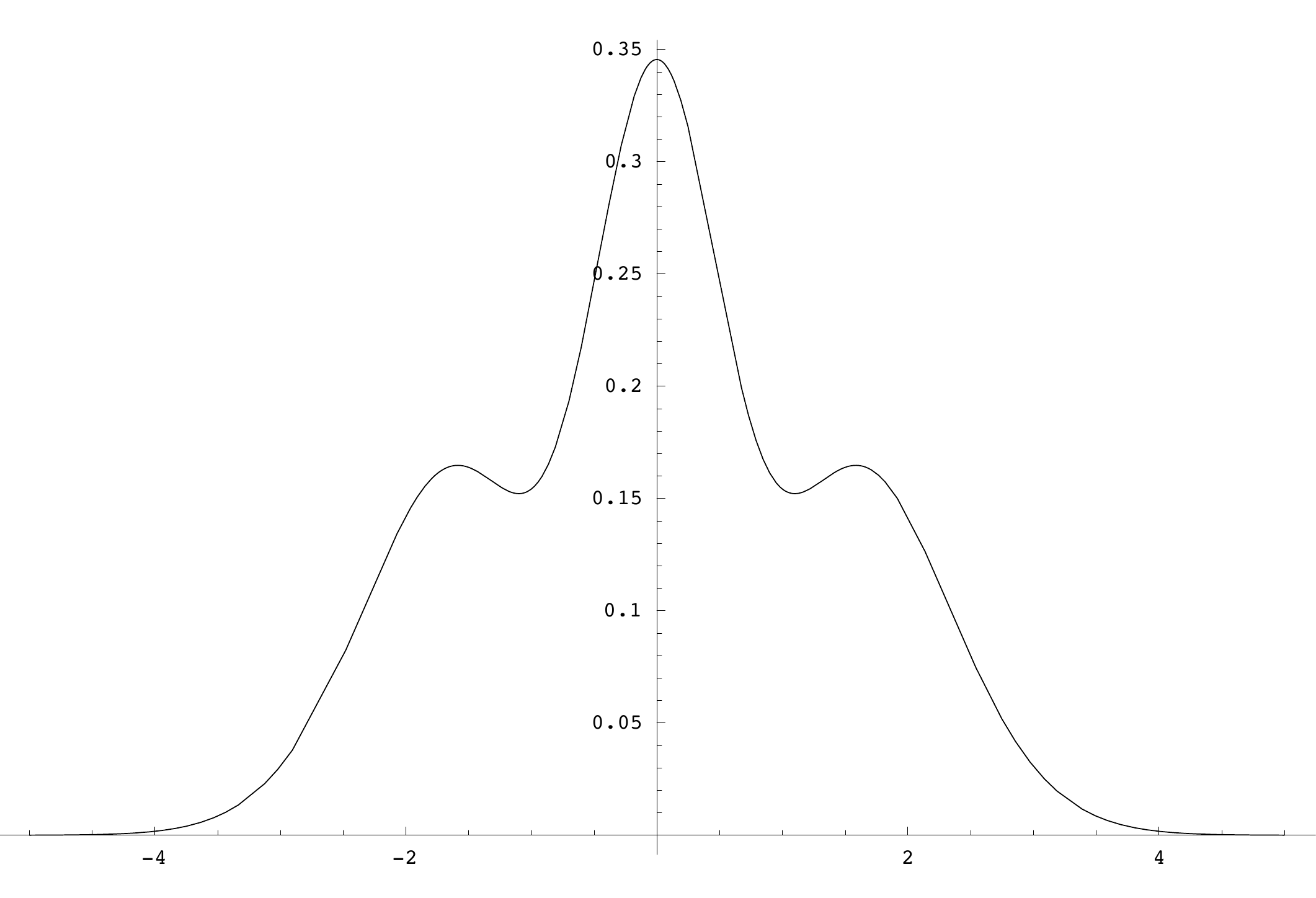}}
\caption{Limiting probability density for critical points: $\frac{\protect%
\sqrt 3}{\protect\sqrt{ 8 \protect\pi}} (2 e^{-t^2}+t^2-1) e^{- \frac{t^2}{2}%
}$. }
\label{fig1}
\end{figure}

\begin{figure}[h]
\centering
\subfigure[$\frac{\sqrt 3}{\sqrt{2 \pi}} (e^{-t^2}+t^2-1)e^{-
\frac{t^2}{2}}$]  {\includegraphics[width=7cm, height=4.5cm]{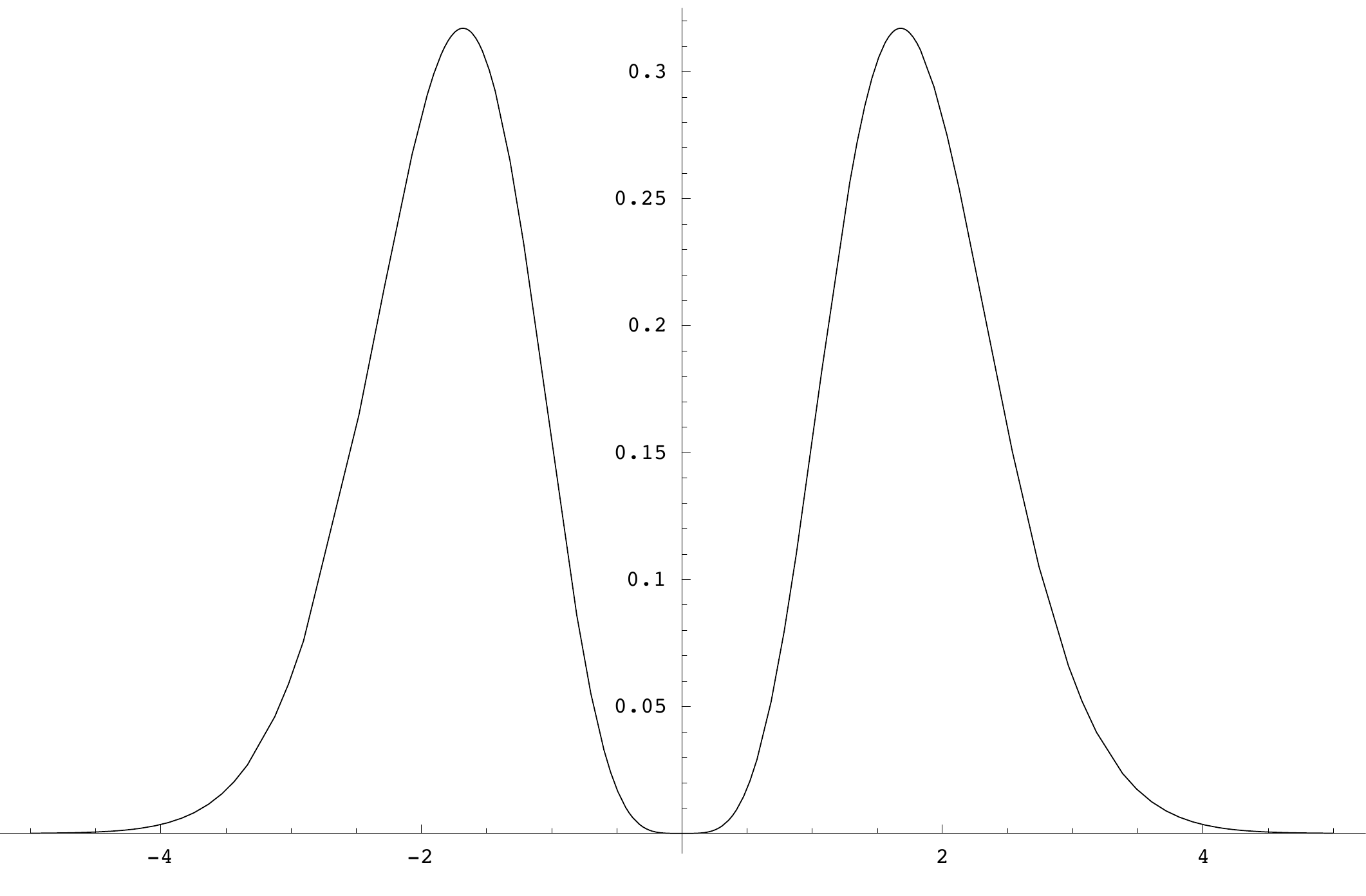} } %
\subfigure[$\frac{\sqrt 3}{\sqrt{2 \pi}} e^{- \frac{3 }{2} t^2} $]
{\includegraphics[width=7cm, height=4.5cm]{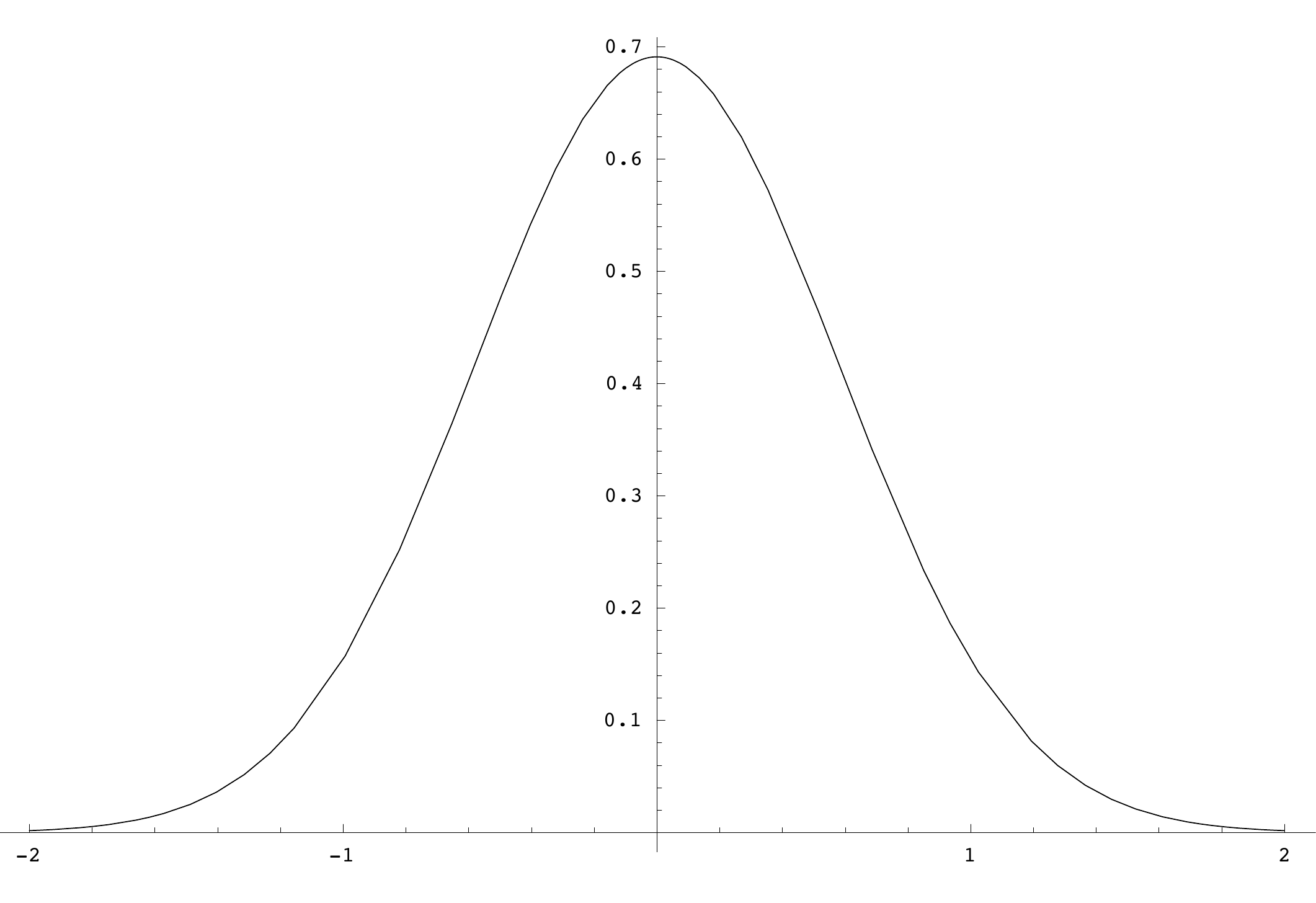} }
\caption{Limiting probability density for extrema (a) and saddles (b).}
\label{fig11}
\end{figure}

The question of asymptotic fluctuations of critical values around the expected number
is more challenging. Here we write
\begin{equation*}
p_{1}^{c}(t)= \frac{4}{\sqrt{3}} \pi _{1}^{c}(t), \hspace{0.5cm} p_{1}^{e}(t)=
\frac{2}{\sqrt{3}}\pi _{1}^{e}(t), \hspace{0.5cm} p_{1}^{s}(t)=\frac{2}{
\sqrt{3}}\pi _{1}^{s}(t),
\end{equation*}
and introduce the functions
\begin{equation*}
p_{2}^{c}(t)=\frac{\sqrt{2}}{\sqrt{\pi }}\left[-4+t^{2}+t^{4}+e^{-t^{2}}2(4+3t^{2})\right]e^{-
\frac{t^{2}}{2}},
\end{equation*}
\begin{equation*}
p_{2}^{e}(t)=\frac{\sqrt{2}}{\sqrt{\pi }}\left[-4+t^{2}+t^{4}+e^{-t^{2}}(4+3t^{2})\right]e^{-
\frac{t^{2}}{2}},
\end{equation*}
\begin{equation*}
p_{2}^{s}(t)=\frac{\sqrt{2}}{\sqrt{\pi }}(4+3t^{2})e^{-\frac{3}{2}t^{2}},
\end{equation*}
and
\begin{equation*}
p_3^a(t)=\frac{5}{4} p^a_1(t)-\frac{1}{4} p_2^a(t)
\end{equation*}
for $a=c,e,s$. Explicitly,
\begin{equation*}
p_{3}^{c}(t)=\frac{1}{ \sqrt{8 \pi }}e^{-\frac{3}{2} t^{2}}
\left[2-6t^{2}-e^{t^{2}}(1-4t^{2}+t^{4})\right],
\end{equation*}%
\begin{equation*}
p_{3}^{e}(t)=\frac{1}{ \sqrt{8 \pi }}e^{-\frac{3}{2} t^{2}
}\left[1-3t^{2}-e^{t^{2}}(1-4t^{2}+t^{4})\right],
\end{equation*}%
\begin{equation*}
p_{3}^{s}(t)=\frac{1}{ \sqrt{8 \pi }}(1-3t^{2})e^{-\frac{3}{2}%
t^{2}}.
\end{equation*}
Finally, for $a=c,e,s$ we denote
\begin{align*}
\nu^a(I)=\left[ \int_I p_3^a(t) d t \right]^2.
\end{align*}
Our principal result concerns the asymptotic behaviour of the variance.

\begin{theorem}
\label{th_variance copy(1)} For every interval $I\subseteq \R$ as $\ell \rightarrow \infty $
\begin{equation*}
{\text{Var}}(\mathcal{N}_{I}^{a}(f_{\ell }))=\ell^3\nu^a(I)+O(\ell ^{5/2}),
\end{equation*}
$a=c,e,s$, where the constant in the $O(\cdot)$ term is universal.
\end{theorem}

\noindent Note that for simplicity all our results are formulated for
intervals $I$, however they can be easily extended to more
general cases, for instance Borel subsets of $\mathbb{R}$. The
plots for the kernel of these variances are given in Figure
\ref{fig2}.

\begin{figure}[h]
\centering
\subfigure[$\frac{1}{ \sqrt{8  \pi}} e^{- \frac{3 }{2} t^2} (2-6 t^2-e^{t^2}
(1-4 t^2+t^4))$ ]  {\includegraphics[width=5cm,
height=3.5cm]{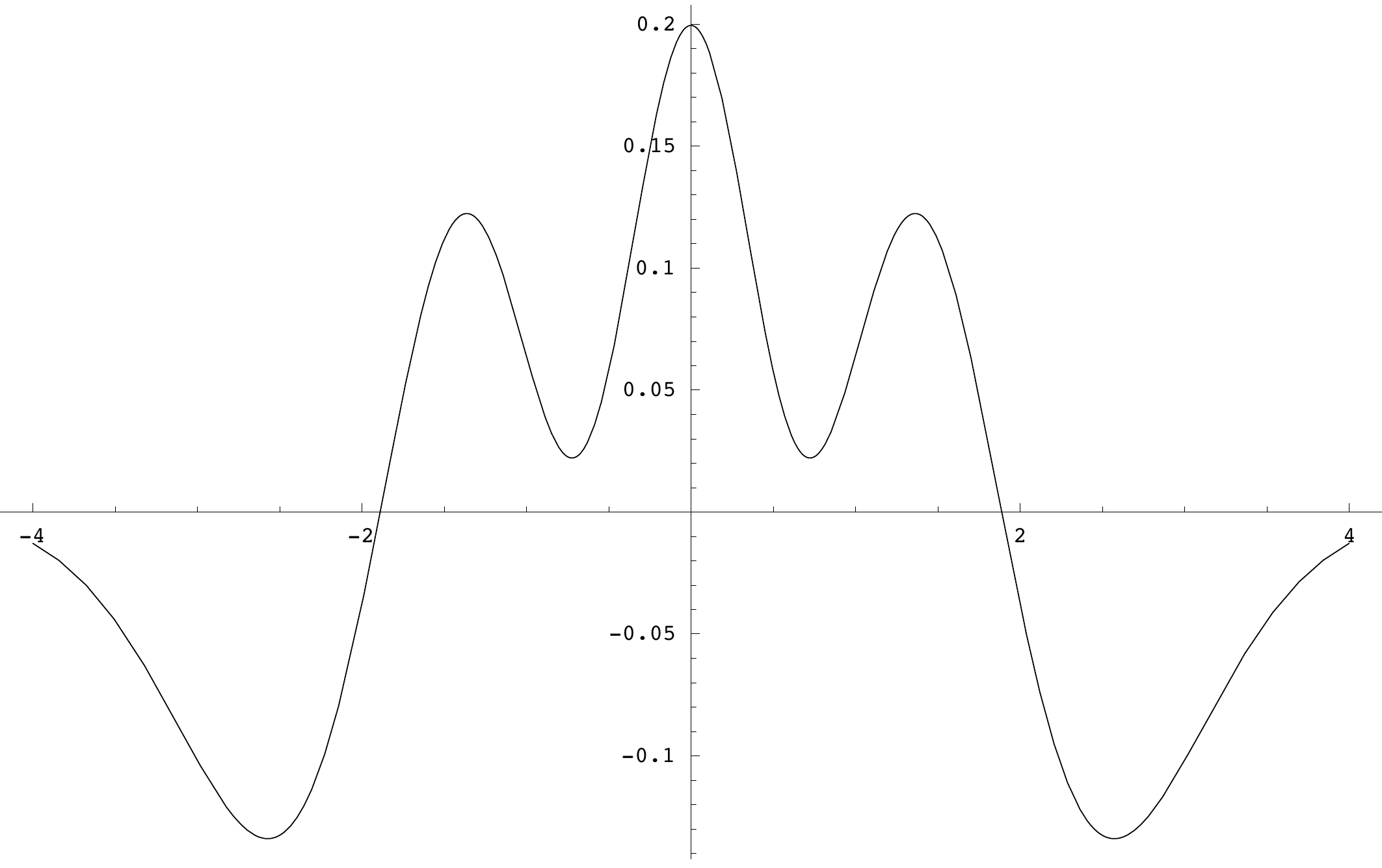} }
\subfigure[$ \frac{1}{\sqrt{8  \pi}}  e^{-
\frac{3}{2}  t^2} (1-3 t^2-e^{t^2} (1-4 t^2+t^4))$]  {\includegraphics[width=5cm, height=3.5cm]{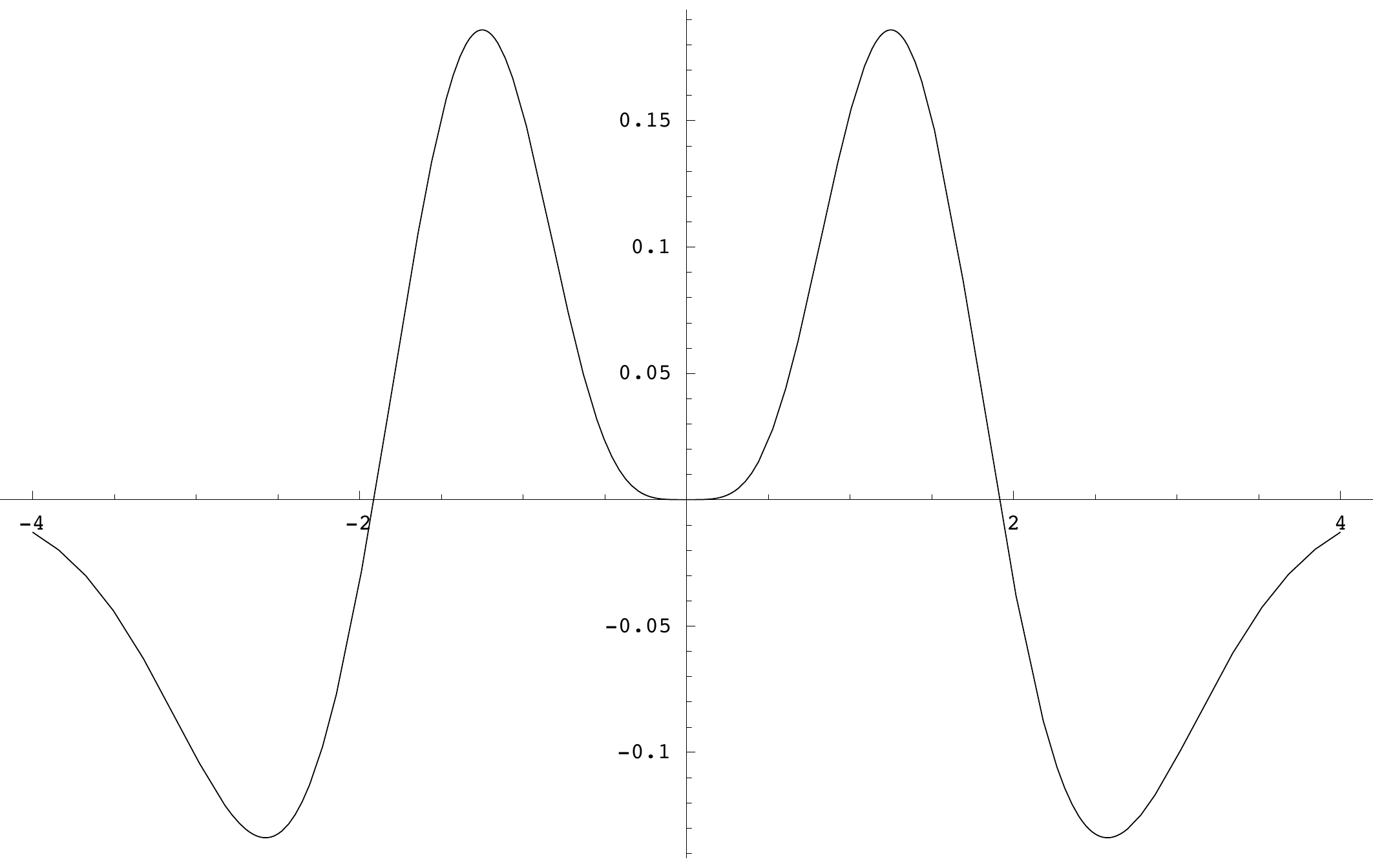} }
\subfigure[$
\frac{1}{\sqrt{8  \pi}}  e^{- \frac{3}{2}  t^2} (1-3 t^2)$]  {\includegraphics[width=5cm, height=3.5cm]{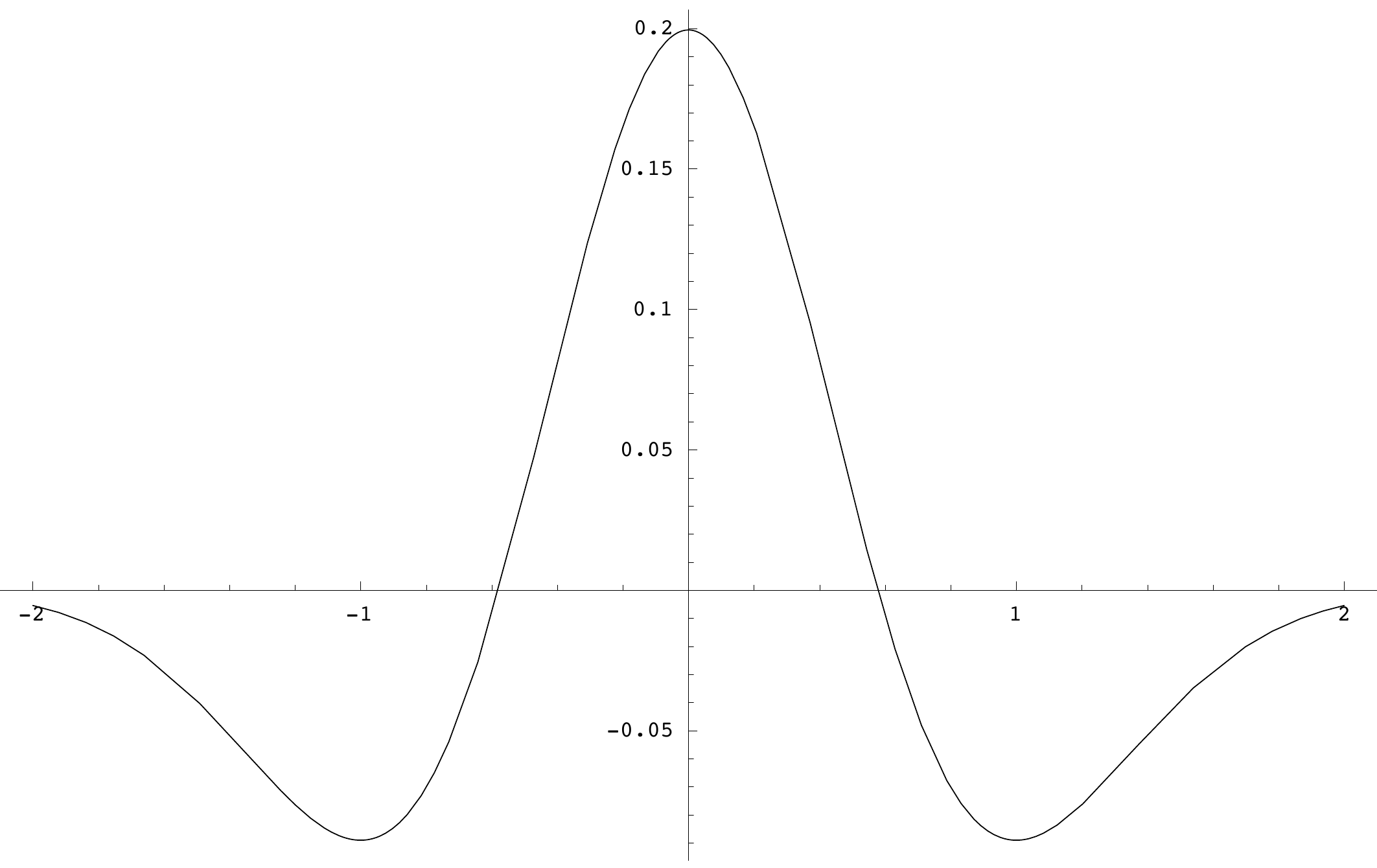} }
\caption{$ p_{3}^a$ for critical (a), extrema (b) and saddles (c).}
\label{fig2}
\end{figure}

\begin{remark}
It is straightforward to evaluate the leading
terms for critical points, extrema and saddles for any given interval
$[a,b]$, as an explicit function of $a$ and $b$. We have

\begin{align*}
\nu^c([a,b])&=\left[ \frac{
-a e^{-\frac{3}{2} a^{2} }(2+(a^{2}-1)e^{a^{2}})+b e^{-\frac{3}{2}%
b^{2}}(2+(b^{2}-1)e^{b^{2}})}{ \sqrt{ 8\pi }}\right] ^{2},   \\
\nu^e([a,b])& =\left[ \frac{%
-a e^{-\frac{3}{2} a^{2}}(1+(a^{2}-1)e^{a^{2}})+be^{-\frac{3}{2}%
b^{2}}(1+(b^{2}-1) e^{b^{2}})}{ \sqrt{ 8 \pi }}\right] ^{2},   \\
\nu^s([a,b])& =\left[ \frac{%
-a e^{-\frac{3}{2}a^{2}}+b e^{-\frac{3}{2}b^{2}}}{ \sqrt{8 \pi }}\right]^{2}.
\end{align*}
More discussion on the behaviour of the leading constant $\nu^a(I)$ for symmetric intervals
$I$ around the origin is reported in the next subsection.
\end{remark}

\begin{remark}
In Figure \ref{fig3} we illustrate the behaviour of  the variances for the excursion sets $I=[u,\infty)$, i.e., we plot
\begin{align*}
\nu^c([u,\infty))&=\frac{1}{8 \pi} e^{- 3 u^2} u^2 (2+ e^{u^2}
(u^2-1))^2, \\
\nu^e([u,\infty))&= \frac{1}{8 \pi}  e^{- 3
u^2} u^2 (1+ e^{u^2} (u^2-1))^2,\\
\nu^s([u,\infty))&=\frac{1}{8 \pi}  e^{- 3 u^2} u^2.
\end{align*}
\end{remark}

\begin{figure}[h]
\centering
\subfigure[$\frac{1}{8 \pi} e^{- 3 u^2} u^2 (2+ e^{u^2}
(u^2-1))^2$ ]  {\includegraphics[width=5cm,
height=3.5cm]{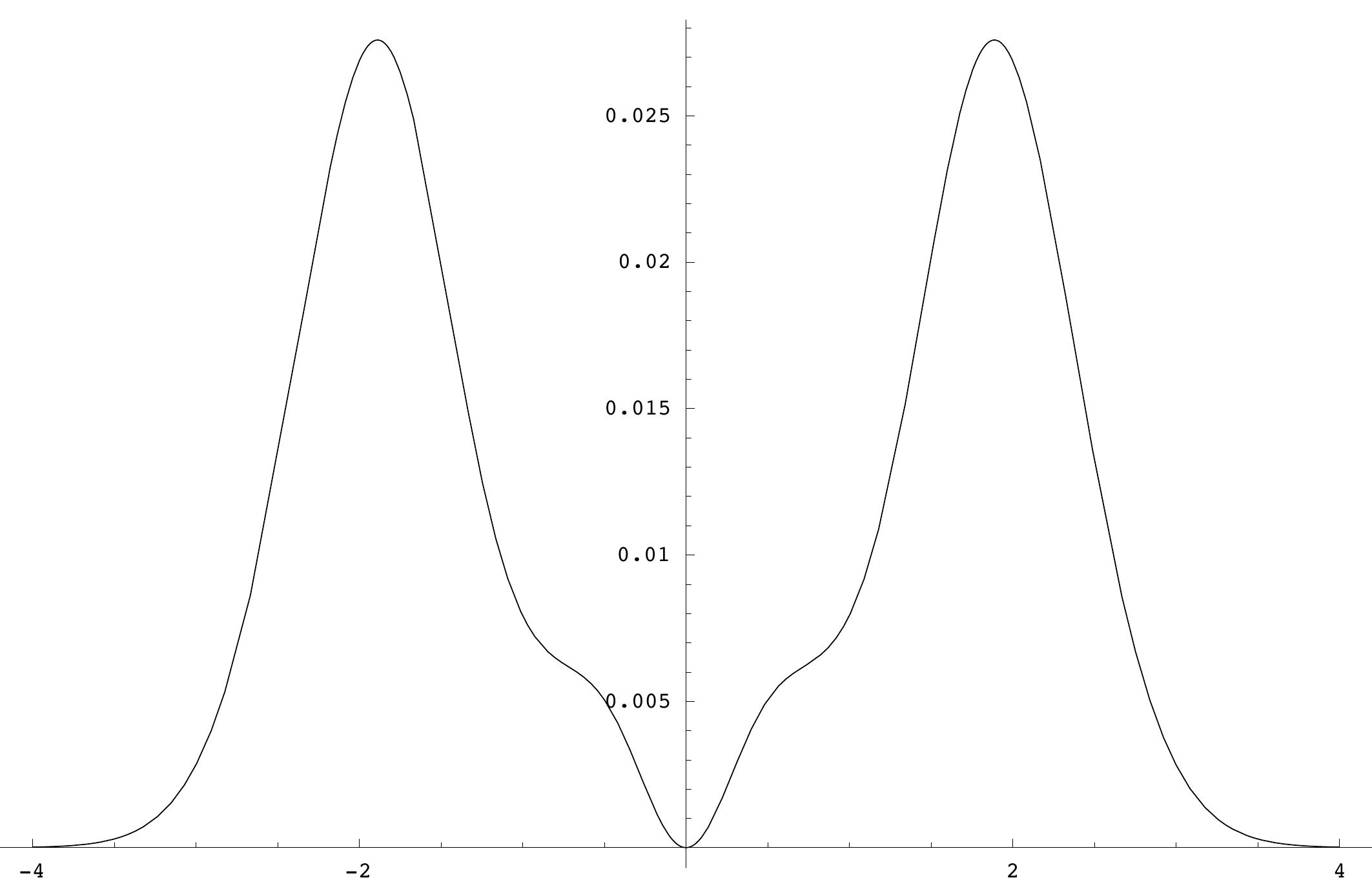} }
\subfigure[$ \frac{1}{8 \pi}  e^{- 3
u^2} u^2 (1+ e^{u^2} (u^2-1))^2 $]  {\includegraphics[width=5cm,
height=3.5cm]{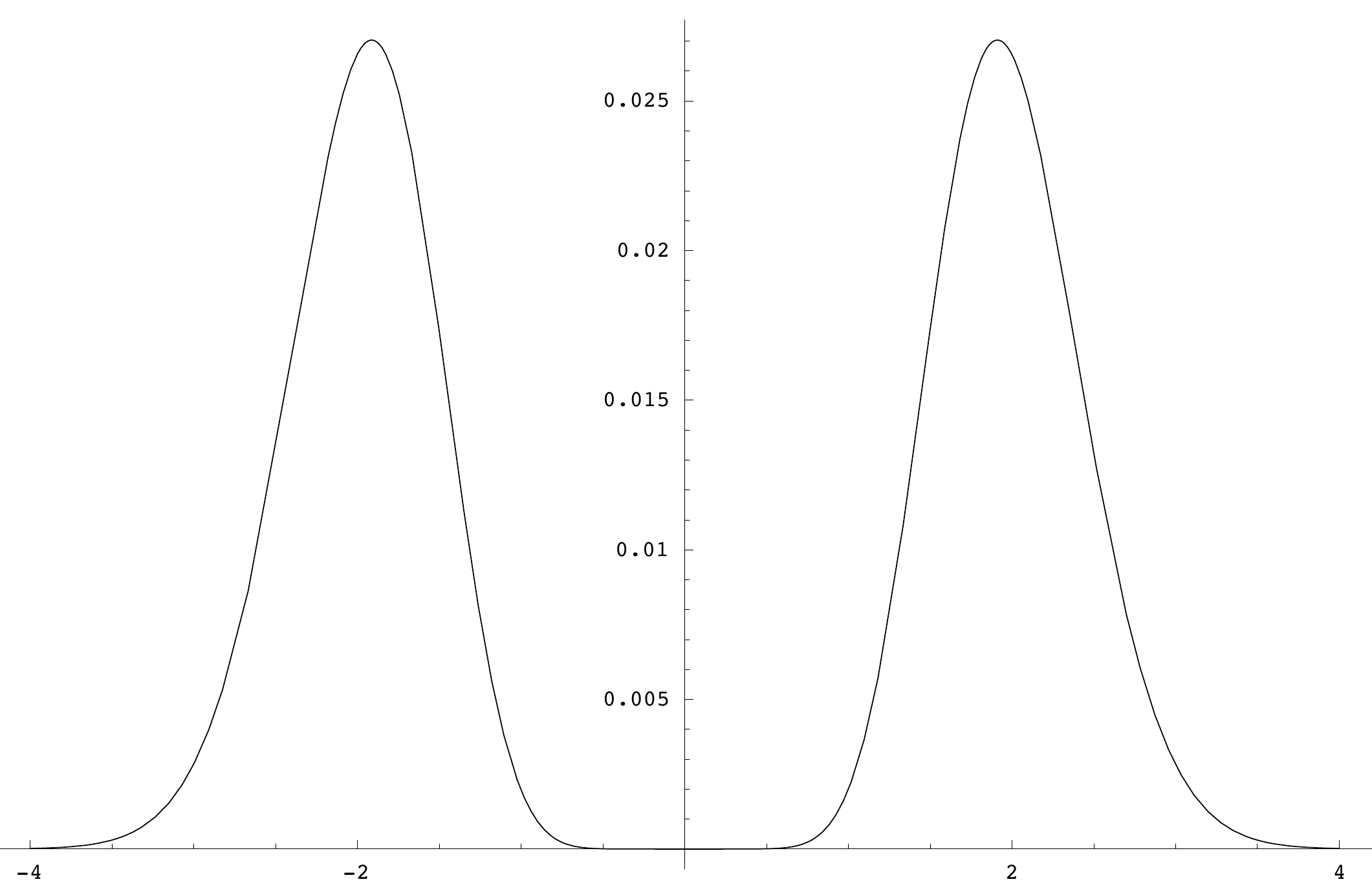} }
\subfigure[$
\frac{1}{8 \pi}  e^{- 3 u^2} u^2 $]  {\includegraphics[width=5cm,
height=3.5cm]{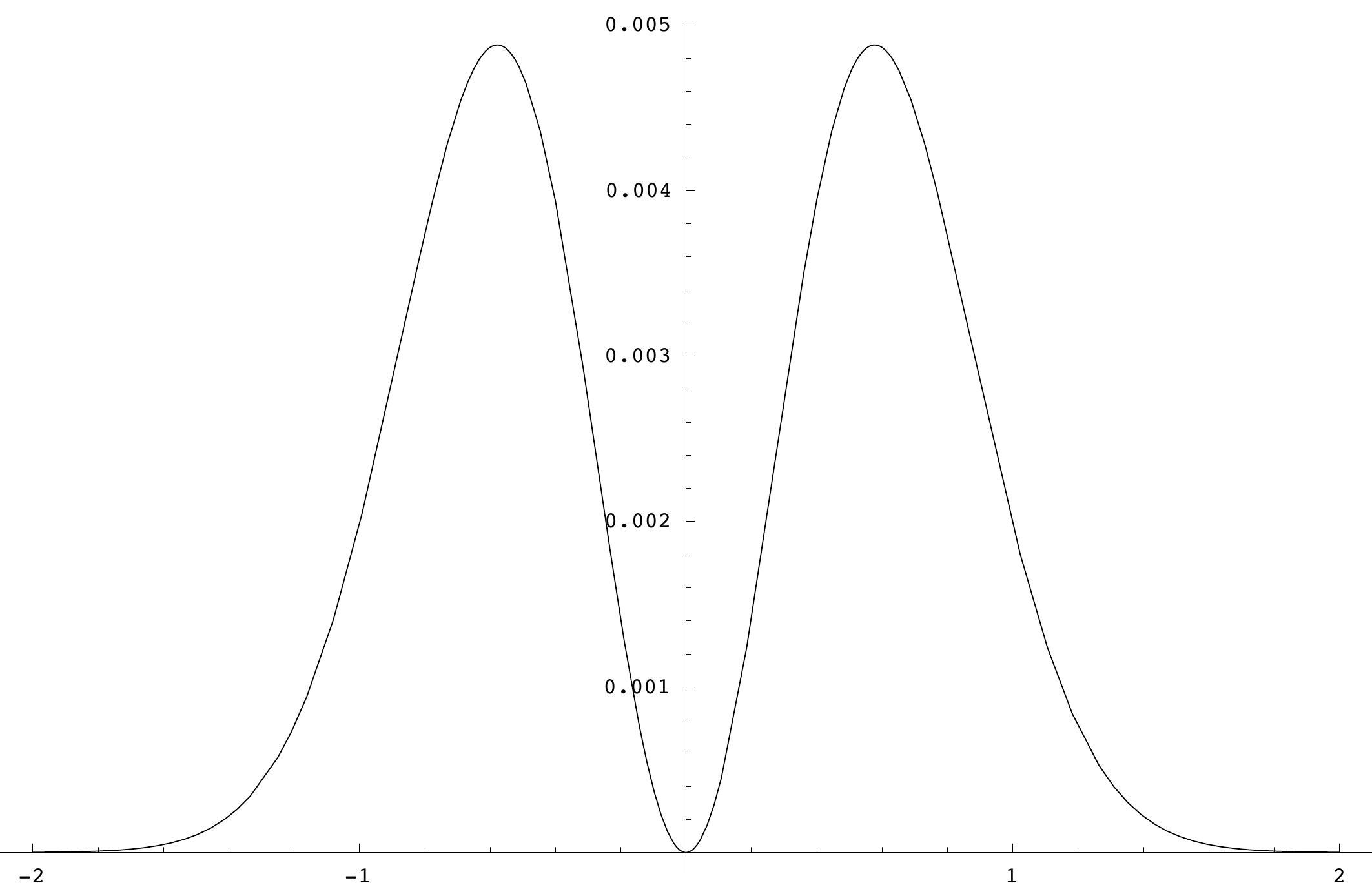} }
\caption{$\nu^a([u,\infty))$
for critical (a), extrema (b) and saddles (c).}
\label{fig3}
\end{figure}

In our view, the asymptotic law we proved for the variances are of
independent interest; they also imply the convergence of empirical
measures of critical points and extrema to their theoretical
limit. More precisely, let
\begin{equation*}
F_{\ell }(z)=\frac{\mathcal{N}^{c}(f_{\ell };(-\infty ,z))}{\mathbb{E}[%
\mathcal{N}^{c}(f_{\ell };\mathbb{R})]},\hspace{2cm}F_{\ell }^{\ast }(z)=%
\frac{\mathcal{N}^{c}(f_{\ell };(-\infty ,z))}{\mathcal{N}^{c}(
f_{\ell };\mathbb{R})},
\end{equation*}%
be the empirical distribution function of critical points for $f_{\ell }$
under deterministic and random normalizations,
respectively. Now define the
distribution functions $F_{\infty }$ as
\begin{equation*}
F_{\infty }(z)=\lim_{\ell \rightarrow \infty }\mathbb{E}[F_{\ell
}(z)]=\int_{-\infty }^{z}\pi _{1}^{c}(t)dt=\int_{-\infty }^{z}\frac{\sqrt{3}%
}{\sqrt{8\pi }}(2e^{-t^{2}}+t^{2}-1)e^{-\frac{t^{2}}{2}}dt.
\end{equation*}

\noindent Our next result concerns the uniform convergence of the empirical
distribution function to $F_{\infty }(z)$.

\begin{corollary} \label{emprcl}
\label{the_th} For all $\varepsilon >0$, as $\ell \rightarrow \infty $, we
have
\begin{equation*}
\mathbb{P}\{\sup_{z}|F_{\ell }^{\ast }(z)-F_{\infty }(z)|\geq \varepsilon \}%
\text{ },\text{ }\mathbb{P}\{\sup_{z}|F_{\ell }(z)-F_{\infty }(z)|\geq
\varepsilon \}\rightarrow 0.
\end{equation*}
\end{corollary}

In practice, loosely speaking, the latter result shows that for each random
realization of a high degree spherical harmonic the same empirical
density of critical values will be observed, up to asymptotically
negligible fluctuations.

\subsection{On Berry cancellation}

An interesting phenomenon occurs when we consider the extrema variance
with values falling into $I$, an infinitesimally small neighbourhood of the origin,
or for a fixed interval $I$ with vanishing leading constant $\nu^a(I)$
(more details are given below). In related
circumstances, it is known ~\cite{wigman dartmouth}
that the nodal length variance for random
eigenfunctions on the torus and on the sphere is of lower order than for
other level curves; on ${\cal S}^2$ the nodal length variance is proportional to $\log \ell$ \cite{Wig},
whereas for generic level curves the variance is
proportional to $\ell$. This behaviour was discovered by Michael Berry in \cite{Berry 1977}, and thereupon
is referred to as \emph{Berry's cancellation phenomenon}.

From Theorem \ref{th_variance copy(1)}, it is easy to obtain, by a simple evaluation of the integral, that the variance of the number
of extrema for a generic interval $I=[-\varepsilon / 2+
x_0,\varepsilon / 2 +x_0]$, is asymptotic to
$$\lim_{\varepsilon \to 0} \lim_{\ell \to \infty} \frac{\text{Var}(\mathcal{N}^{e}(f_\ell ; [- \varepsilon / 2+x_0,\varepsilon / 2+x_0])}{\varepsilon^2 \ell^3}= [p_3^e(x_0)]^2 ,$$
 where $[p_3^e(x_0)]^2>0$, almost everywhere, see Figure \ref{fig2}.  In contrast, from Theorem \ref{th_variance copy(1)}, we may also deduce
the behaviour of the critical points variance in a vanishing interval $I=[-\varepsilon ,\varepsilon ]$
around the origin:

\begin{corollary} As $\varepsilon \to 0$
\begin{equation} \label{thiseq}
\lim_{\varepsilon \rightarrow 0} \frac{\nu^e([-\varepsilon,\varepsilon])}{\varepsilon^{10}}
= \frac{1}{8 \pi }.
\end{equation}
\end{corollary}

\begin{proof} The statement follows immediately by evaluating
$$\lim_{\varepsilon \to 0} \frac{1}{\varepsilon^{10}} \left[ \int_{-\varepsilon}^{\varepsilon} \frac{1}{ \sqrt{8 \pi}} e^{-\frac 3 2 t^2} \left[1-3t^2-e^{t^2} (1-4 t^2+t^4)\right] d t \right]^2  =\lim_{\varepsilon \to 0} \frac{e^{- 3 \varepsilon^2}  \left[1+e^{ \varepsilon^2} (-1+ \varepsilon^2)\right]^2}{2 \pi \varepsilon^8}=\frac{1}{8 \pi}.$$
\end{proof}

\noindent Figure \ref{5} illustrates the behaviour of these functions for symmetric intervals around the origin.

\begin{figure}[h]
\centering
\subfigure[$\frac{1}{2 \pi} e^{- 3 \varepsilon^2} \varepsilon^2 (2+
e^{\varepsilon^2} (\varepsilon^2-1))^2$ ]  {\includegraphics[width=5cm,
height=3.5cm]{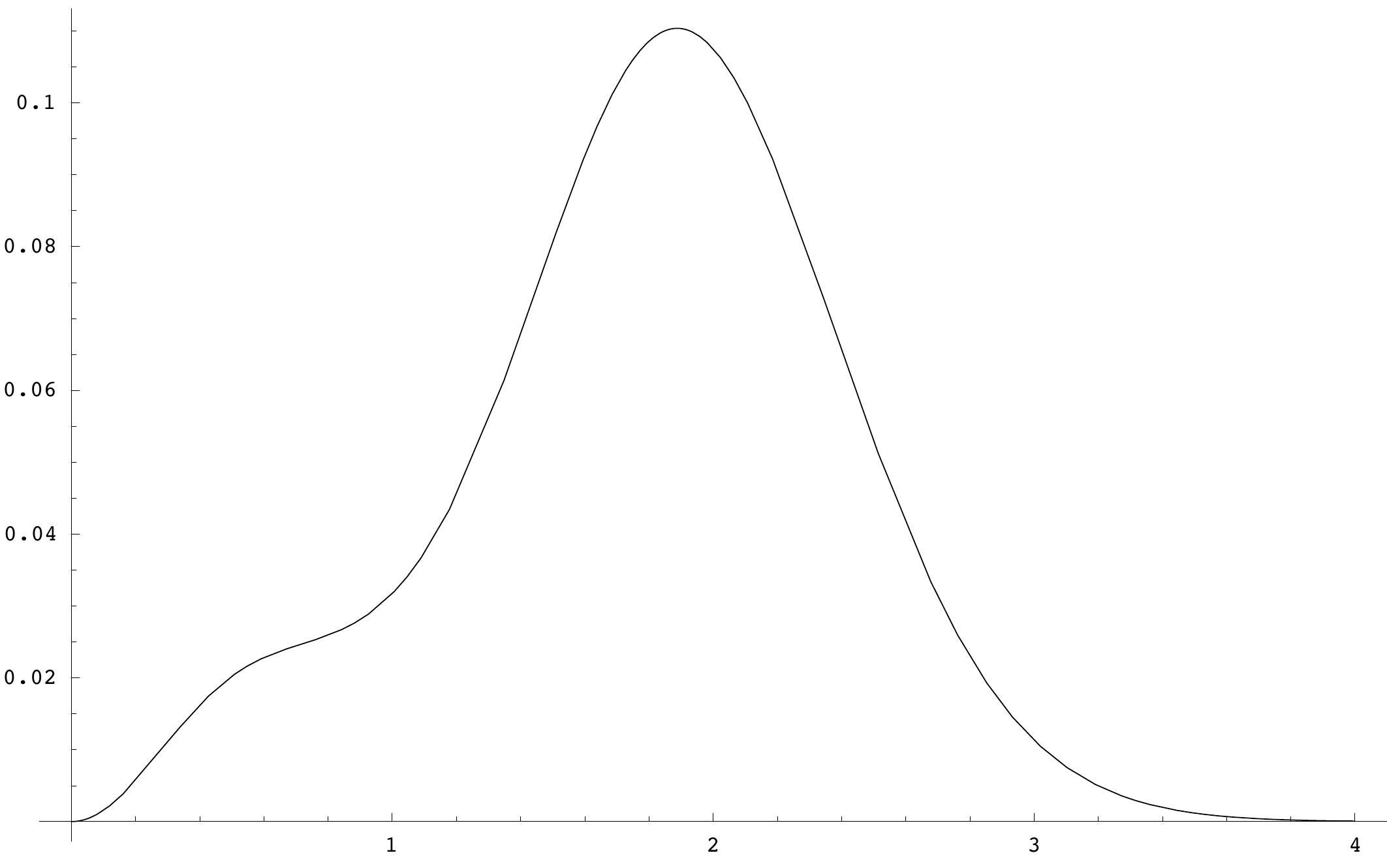} }
\subfigure[$ \frac{1}{2 \pi}  e^{- 3
\varepsilon^2} \varepsilon^2 (1+ e^{\varepsilon^2} (\varepsilon^2-1))^2 $]  {\includegraphics[width=5cm, height=3.5cm]{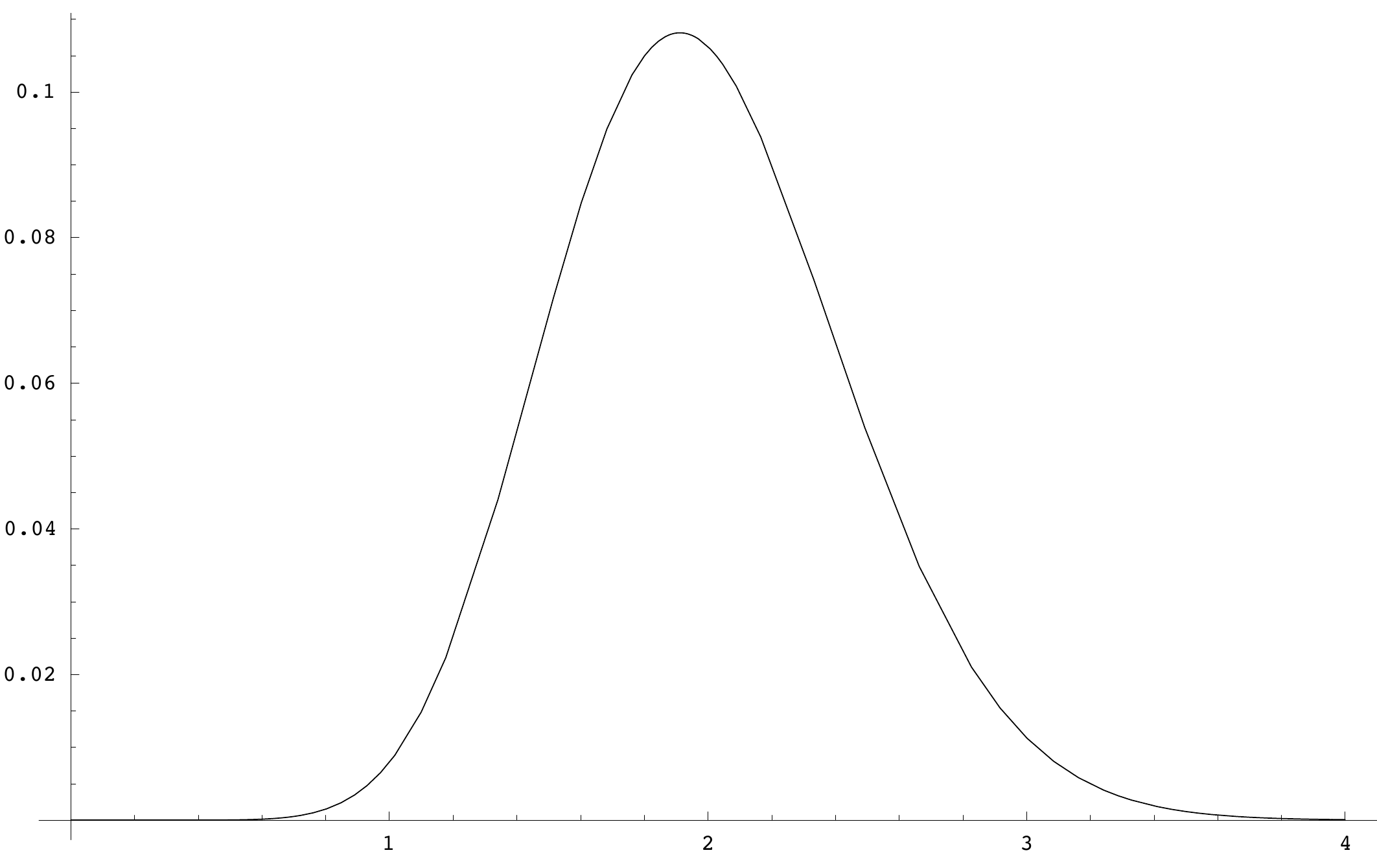} }
\subfigure[$
\frac{1}{2 \pi}  e^{- 3 \varepsilon^2} \varepsilon^2 $]  {\includegraphics[width=5cm, height=3.5cm]{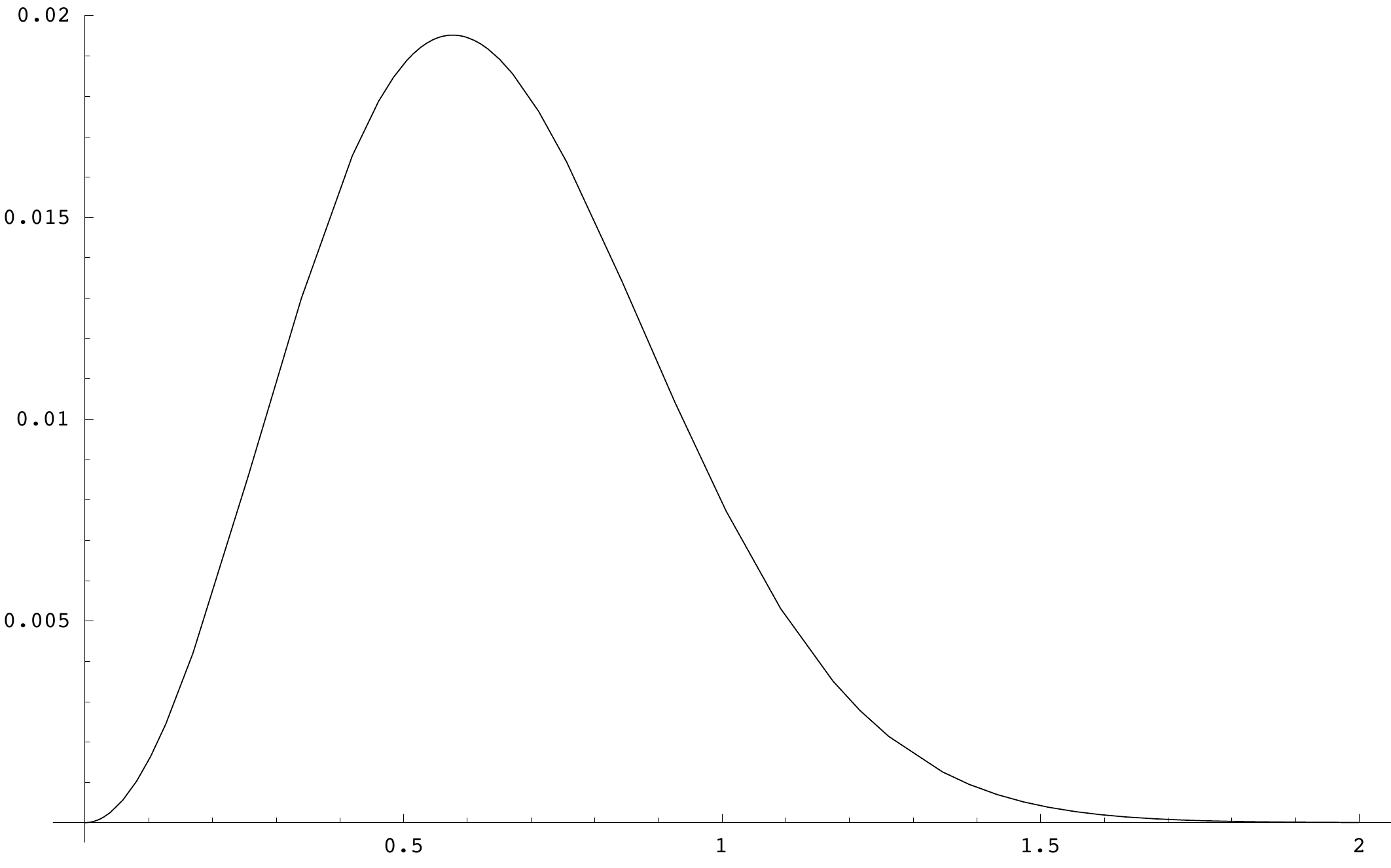} }
\caption{$\frac{\text{Var}[\mathcal{N}^a(f_{\ell };[-\varepsilon ,
\varepsilon ])]}{\ell ^{3}}$ for critical (a), extrema (b) and
saddles (c).} \label{5}
\end{figure}

We note also that for some specifically chosen (but fixed)
intervals $I\subseteq \R$ with $\nu^a(I)=0$, such as, for example
$I=\R$ (for the latter case the unrestricted total number of
critical points, extrema or saddles is counted), the order of
magnitude of the variance is lower than $\ell^{3}$. In this case
Theorem \ref{th_variance copy(1)} reads:
\begin{equation*}
\text{Var}(\mathcal{N}_{\mathbb{R}}^{c}(f_{\ell })),
\text{Var}(\mathcal{N}_{\mathbb{R}}^{e}(f_{\ell })) ,
\text{Var}(\mathcal{N}_{\mathbb{R}}^{s}(f_{\ell }))
=O(\ell^{\frac{5}{2}}).
\end{equation*}
as $\ell \rightarrow\infty$.
It seems though that by our present methods we may sharpen the latter bound to the next term of order $O(\ell^{2}\log{\ell})$, and,
unless further cancellation occurs, it may be the true asymptotic behaviour of each of the three quantities above.
However, in a recent numerical simulation by D. Belyaev \cite{Belyaev private} the observed fluctuations were far too small
for the latter to hold.

\subsection{Overview of the proof}

Our proof below is technically demanding,  and we present here its main
conceptual steps for the variance result (Theorem \ref{th_variance copy(1)}).
Our argument is based on a suitably modified version of the Kac-Rice formula
for the number of zeroes of the gradient of $f_{\ell }.$ The first
technical difficulty is related to the fact that the $6$-dimensional vector $
(f_{\ell }(x),\nabla f_{\ell }(x),\nabla ^{2}f_{\ell }(x))$ is always
degenerate, as the level field $f_{\ell }$ is a linear combination of
gradient and second order derivatives. However, this issue is relatively easily
mended by reducing the dimension of the problem to take this
degeneracy into account.

A much trickier issue arises when considering the two-point correlation
function needed for the evaluation of the variance. Here we have to cope
with the $10$-dimensional Gaussian random vectors of the form
\begin{equation*}
(\nabla f_{\ell }(x), \nabla f_{\ell }(y),\nabla ^{2}f_{\ell }(x),\nabla
^{2}f_{\ell }(y)), \;\; (x,y)\in {\cal S}^2,
\end{equation*}%
imposing suitable conditions to ensure that $(f_{\ell }(x),f_{\ell }(y))\in
I$. {\em A priori} there is no certainty that this random
vector is nondegenerate, a condition that guarantees the applicability of the standard
Kac-Rice. Our basic idea is to split the range ${\cal S}^2 $ of the integration of the Kac-Rice integral
into two parts: the short range regime $d(x,y)<C/\ell $, $d(x,y)=\arccos (\left\langle x,y\right\rangle )$ denoting the usual
spherical distance and $C$ a sufficiently big positive constant, and the long range regime $d(x,y)>C/\ell$. In the short range regime the
Kac-Rice formula holds only approximately, and we can prove by a
partitioning argument inspired from \cite{rudnickwigman} that the corresponding
contribution is of order $O(\ell^{2})$. The proof of the latter requires a
precise Taylor analysis of the behaviour of Legendre polynomials and their
derivatives around the origin, and related analytic functions.

The main term comes from the long range regime.
Here the asymptotic analysis is based on the
properties of multivariate conditional Gaussian variables, and an
asymptotic study of the tail decay of the Legendre
polynomials and their derivatives. In this regime, Kac-Rice formula holds exactly and
we shall exploit the fact that a Gaussian expectation is an analytic function with respect to
the parameters of the corresponding covariance matrix outside its singularities. It is then possible to compute the Taylor expansion of these expected values around the origin with respect to the vanishing entries of the covariance matrix; a small finite number of these
(depending on the interval $I$) make an asymptotically significant contribution to the variance, whereas the rest are negligible.

\subsection{Background and motivation}

\subsubsection{Cosmology and CMB}

Our main motivation for this paper is given by cosmological and astrophysical applications.
Indeed, it is well-known that random spherical
harmonics are the Fourier components of square integrable
isotropic fields on the sphere, i.e., for every centred Gaussian
spherical random field $f(x)$ the following spectral
representation holds \cite{MaPeCUP}:
\[
f(x)=\sum_{\ell=1}^{\infty}f_{\ell}(x)=\sum_{\ell=1}^{\infty}\sum_{m=-\ell}^{\ell}\sqrt{C_{\ell}}\, a_{\ell m} \, Y_{\ell m}(x),
\]
where equality holds in the $L^2$ sense and the sequence
$\{C_{\ell}\}_{\ell=1,2,...}$ denotes the so-called angular power
spectrum, which fully characterizes the dependence structure of
${f(x)}$. The analysis of spherical random fields is now at the
heart of observational cosmology, for instance for experiments
handling Cosmic Microwave Background radiation data, see e.g.,
\cite{planck} and \cite{cobe}. In summary, we can represent CMB
observations as a realization of the isotropic, Gaussian random
function, which we denote by $\hat{f}(x)$; realizations of the
random spherical harmonic components $\hat{f}_{\ell}(x)$ are then
obtained by standard Fourier analysis, i.e.
\begin{equation}
\hat{f}_{\ell}(x):=\sum_{m=-\ell}^{\ell}\hat{a}_{\ell m} Y_{\ell
m}(x), \;\;\;  \hat{a}_{\ell
m}:=\int_{{\cal S}^2 }\hat{f}(x)\bar{Y}_{\ell m}(x)dx,
\label{inverseFourier}
\end{equation}
the bar denoting as usual complex conjugation. It is to be noted
that in many experimental circumstances the realizations of these
random fields are observed only on subsets of the sphere, and this
can make the inverse Fourier transform in \eqref{inverseFourier}
unfeasible: however, very recently some more sophisticated
statistical techniques have indeed led to the reconstruction of full sky
data maps, see \cite{starck2014}, and in this setting the empirical
derivation of $\hat{f}_{\ell}$ has become possible. A natural
question is whether these observed CMB maps are indeed consistent
with the starting assumptions of Gaussianity and isotropy;
departures from these assumptions could signal either spurious
features introduced by the algorithms to produce the maps, or
physically motivated deviations from standard cosmological models.
Examples of the former are, for instance, astrophysical components
which have not been properly removed from CMB maps, such as
so-called point-sources (galaxies and other astrophysical
objects unrelated to CMB).

Our results can be exploited in this setting by means of the
implementation of a number of Gaussianity and isotropy tests. For
instance, it is possible to compare the actual number of maxima
above a given threshold $u$ for an observed component
$\hat{f}_{\ell}$ with its expected value and standard deviation
which we reported in the previous subsection; i.e., for any given threshold value $u$, we may construct statistics such as
\begin{align*}
Z_{\ell}(u)=\frac{{\cal N}_{\ell}(f_{\ell}; [u,\infty))-\mathbb{E}[{\cal N}_{\ell}(f_{\ell}; [u,\infty))]}{\sqrt{\text{Var}({\cal N}_{\ell}(f_{\ell}; [u,\infty)))}}.
\end{align*}
By the results on expected values and variances provided in this paper, the previous statistic can be computed explicitly for any value of $u$. It is natural to expect that convergence to a standard Gaussian limit will hold in the high-energy regime, under the null assumption that $\{f_{\ell}\}$ is a pure Gaussian field; on the contrary, nonGaussian features such as the previously mentioned point sources will show up as a higher number of observed maxima than predicted under Gaussian circumstances; therefore high values of $Z_{\ell}$ will  signal the presence of spurious components. Extensions to cover joint tests on multiple threshold $u_1,\dots, u_p$ are straightforward. Of course, the actual implementation of these procedures on real data will require further work, which we delay to future research (see \cite{fantaye}).

\subsubsection{Nodal domains of Laplace eigenfunctions}

The {\em nodal components} of $f_{\ell}$ are the connected components of the {\em nodal line}
$f_{\ell}^{-1}(0)$, and the nodal domains of $f_{\ell}$ are the connected components of its complement
${\cal S}^2\setminus f_{\ell}^{-1}(0)$. It was asserted that the nodal structure of $f_{\ell}$
(or Laplace eigenfunctions, random or deterministic, on generic surfaces)
could be modelled ~\cite{Bogomolny-Schmit} by a bond percolation-like model which could
be explained as follows.

Let $\mathcal{L}^{+}$ and $\mathcal{L}^{-}$ be the (random) sets of
the local minima and maxima of $f_{\ell}$ respectively.
Under the percolation model $\mathcal{L}^{+}$ and $\mathcal{L}^{-}$ are thought of as
mutually dual square grids with $\approx \ell^{2}= \ell \times \ell$ points (`sites'), each
representing a maximum or minimum respectively.
Each pair of adjacent (w.r.t. the grid) sites are connected by an `open' bond
in $\mathcal{L}^{+}$ with probability ${1}/{2}$ independent of other bonds,
whence the dual bond in $\mathcal{L}^{-}$ is `closed' and vice
versa. One can then study some aspects of the percolation
process described, such as the number of clusters of $\mathcal{L}^{\pm}$
(representing the number of the nodal domains
of $f_{\ell}$), their area distribution etc.
It is important that there are only few
low-lying extrema (see Proposition \ref{expectation copy(1)} and
Figure \ref{fig11} (a)), corresponding to nodal domains
unstable under small perturbations of $f_{\ell}$.

Recent numerical studies
revealed small but significant deviation from the percolation model
(e.g. ~\cite{Belyaev Kereta}); this deviation may be attributed ~\cite{Belyaev private} to the unsubstantiated
rigidity assumption on the sites positions along $\mathcal{L}^{\pm}$, and it was suggested
~\cite{Belyaev private} that the rigidity of the sets $\mathcal{L}^{\pm}$ should be relaxed.
Theorem ~\ref{th_variance copy(1)} then may be used to determine the measure of flexibility or rigidity
expected from the sets $\mathcal{L}^{\pm}$ to satisfy in a more sophisticated percolation-like model for
the nodal structure of Laplace eigenfunctions.

\subsubsection{Persistence barcodes}

Our results may also find natural applications in the rapidly
growing areas of applied algebraic topology and topological data
analysis, and in particular for the characterization of the
stochastic properties of \emph{persistence barcodes} and
\emph{persistence diagrams} (see e.g. \cite{carlsson} or
\cite{adlerstflour}) for excursion sets of random spherical
harmonics. Write
\[
A_u(f_{\ell})=\left \{x \in \mathcal{S}^2:f_{\ell}(x)>u \right \}
\]
for the excursion sets of $f_{\ell}$, and let us recall that a
\emph{barcode} for $A_u(f_{\ell})$ is a pair of graphs, each
corresponding to one of the two homology groups for the
corresponding excursion sets, $H_k( A_u(f_{\ell}))$ where $k=0,1$.
Loosely speaking, $H_0(A_u(f_{\ell}))$ is generated by the elements
that represent the connected components of the excursion sets, and
$H_1(A_u(f_{\ell}))$ is generated by elements that represent
1-dimensional "loops". Each of the two graphs in this barcode is a
collection of bars; a bar in the graph representing $H_0$,
starting at threshold $u_{start}$ and ending at threshold
$u_{end}$, corresponds to a generator of $H_0(A_u)$ that
"appeared" at level $u_{start}$ and "disappeared" at level
$u_{end}$; if two connected components of $A_u(f_{\ell})$ merge,
then only one of the two corresponding bars remains. An analogous meaning can be
given to the bars in the second graph, see
\cite{carlsson},\cite{adlerstflour} for more details and
discussion. Hence the number of bars in graph $k$ at any level $u$
equals the \emph{Betti number} $\beta_k(A_u(f_{\ell}))$ for the
excursion region corresponding to this threshold.

A \emph{persistence diagram} for $H_k$, $k=0,1$ is a set of pairs
$(u_{end}(k),u_{start}(k))$ corresponding to the starting and
ending points of these bars. In \cite{adlerstflour}, p. 107--108
it is explained that the starting points of the $H_0$ bars
correspond to the heights of local maxima, whereas the ending
points of the $H_1$ bars correspond to the heights of the local
minima. Hence our results in this paper establish the density of
$u_{end}(1)$ and $u_{start}(0)$ in the case of random spherical
harmonics; the shape of our curves can be compared to the
simulated results reported in \cite{adlerstflour}, figure 6.2.2,
which represent persistence diagrams of excursion sets from a
Gaussian isotropic random field on the unit square.

\subsection{Plan of the paper}

The plan of this paper is as follows: in section \ref{due} we
establish the asymptotic density of critical points, extrema and
saddles; in section \ref{tre} we discuss the approximate Kac-Rice
formula instrumental for establishing our results; section
\ref{asymptoticsection} discusses the derivation of the two-point
correlation function while section \ref{cinque} is devoted to the
proofs for the expressions of the variances reported in the
introduction. Finally, section \ref{sei} provides the convergence
results for the empirical measures of critical points and extrema.
A number of auxiliary results of more technical nature
facilitating the computations of covariance matrices and
asymptotics for Legendre polynomials are collected in the
appendix.

\subsection{Acknowledgements}

The research leading to these results has received funding from the
European Research Council under the European Union's Seventh
Framework Programme (FP7/2007-2013) / ERC grant agreements
n$^{\text{o}}$ 277742 (D.M.) and n$^{\text{o}}$ 335141 (I.W.). I.W. was also
partially supported by the EPSRC grant under the First Grant scheme
(EP/J004529/1). We are grateful to J.-M. Azais and Z. Rudnick for
some useful discussion and suggestions; the usual disclaimers apply.

\section{Asymptotic density of critical values} \label{due}

\subsection{On the Kac-Rice formula for the expected number of critical values}

\label{sec:Kac-Rice expected}

Let $\mathcal{E}\subseteq \R^{n}$ be a nice Euclidian domain, and
$g:\mathcal{E}\rightarrow \R^{n}$ a centred Gaussian random field,
a.s. smooth. The set $g^{-1}(0)$ a.s. consists of finitely many
zeros of $g$. One defines the zero density (also referred to as
``first intensity") $K_{1}=K_{1;g}:\mathcal{E}\rightarrow\R$ of
$g$ as
\begin{equation*}
K_{1}(x) = \phi_{g(x)}(\mathbf{0})\cdot \E[ |\det J_{g}(x) |  \big| g(x)=\mathbf{0}  ],
\end{equation*}
where $\phi_{g(x)}$ is the (Gaussian) probability density of $g(x)\in \R^{n}$ and $J_{g}(x)$ is the
Jacobian matrix of $g$ at $x$.
Under the assumption that for all $x\in \mathcal{E}$ the distribution of $(g(x),J_{g}(x)) $
is non-degenerate in $\R^{n}\times \R^{n(n+1)/2}$ (i.e. is not concentrated in a proper subspace
of the latter), the expected number of zeros of $g$ on $\mathcal{E}$ is given by ~\cite{adlertaylor}, Theorem 11.2.1
\begin{equation*}
\E[\#g^{-1}(\mathbf{0})] = \int\limits_{\mathcal{E}} K_{1}(x)dx.
\end{equation*}

To apply the latter formula in our case we will work with spherical coordinates on $\mathcal{S}^{2}$ and
use an explicit orthonormal frame (see section \ref{sec:Kac-Rice coord}); counting the critical points of
$f=f_{\ell}$ is equivalent to counting the zeros of the map $[0,\pi]^{2}\rightarrow\R^{2}$ given by $x\mapsto \nabla f(x) =(f_{1}(x),f_{2}(x))$,
where $(f_{1},f_{2})$ are the partial derivatives of $f$. Here we have
\begin{equation*}
K_{1}(x)=K_{1;l}(x) = \phi_{\nabla f(x)}(0,0)\cdot \E[ |\det H_{f}(x) |  \big| \nabla f(x)=0  ],
\end{equation*}
where $H_{f}(x)$ is the Hessian matrix of $f$ at $x$. An explicit computation of the covariance matrix of
$(\nabla f(x),H(x))\in \R^{5}$ shows that for $\ell$ sufficiently big the distribution
of the latter is non-degenerate so that ~\cite{adlertaylor}, Theorem 11.2.1 yields that the expected total number
of critical points of $f$ is
\begin{equation*}
\E[\mathcal{N}_{\R}^{c}(f)] = \int\limits_{\mathcal{S}^{2}} K_{1}(x)dx;
\end{equation*}
the isotropy of $f_{\ell}$ implies that $K_{1}(x)\equiv K_{1}$ depends on $\ell$ only, independent of $x$.

For counting the number of critical points with corresponding value lying in $I\subseteq \R$, we
need to modify $K_{1}(x)$ so that this time we define
\begin{equation}
\label{eq:K1 xI def}
K_{1}(x;I) = \phi_{\nabla f (x)}(0,0)\cdot \E[ |\det H(x) |\cdot  1\hspace{-0.27em}\mbox{\rm l}_{I}(f(x)) \big| \nabla f(x)=0],
\end{equation}
where $1\hspace{-0.27em}\mbox{\rm l}_{I}$ is the characteristic function of $I$ on $\R$. In this case ~\cite{adlertaylor}, Theorem 11.2.1 yields
\begin{equation*}
\E[\mathcal{N}_{I}^{c}(f)] = \int\limits_{\mathcal{S}^{2}} K_{1;I}(x)dx,
\end{equation*}
again, under the non-degeneracy assumption on $(\nabla f(x),H(x))$. Note that in our case there is
a linear dependency between the value $f(x)$, involved in the definition \eqref{eq:K1 xI def} of $K_{1}(x;I)$,
and the Hessian $H(x)$ (see \eqref{linear_dp} below); nevertheless
the non-degeneracy of $(\nabla f(x),H(x))$ is sufficient for an application of ~\cite{adlertaylor}, Theorem 11.2.1;
the linear dependency \eqref{linear_dp} allows us to reduce the dimension of the Gaussian distribution involved
in the evaluation of $K_{1}(x,I)$ from $6$ to $5$, a considerable technical simplification. It is easy to
adapt the same approach to separate the critical points into extrema and saddles (see section \ref{sec:Kac-Rice coord},
towards the end).

\subsection{Application of Kac-Rice in coordinates}

\label{sec:Kac-Rice coord}

In this section we formulate the (precise) Kac-Rice formula to derive the expected value of the number of critical points,
saddles and extrema with values in a given interval $I\subset \mathbb{R}$. To
this aim, let us first introduce some notation.

\noindent Given $x\in {\cal S}^2$, consider
a local orthogonal frame $\{e_{1}^{x},e_{2}^{x}\}$ defined in some neighbourhood of $x$, such that, for any
regular function $h:{\cal S}^2\rightarrow \mathbb{R}$, we have $%
e_{1}^{x}e_{2}^{x}h=e_{2}^{x}e_{1}^{x}h$. Via an isometry, for every $x\in
{\cal S}^2$, it is possible to obtain a (local) identification
\begin{equation*}
T_{x}({\cal S}^2)\cong \mathbb{R}^{2},
\end{equation*}%
so that we do not have to work with probability densities defined on the
tangent planes $T_{x}({\cal S}^2)$ which depend on the point $x \in {\cal S}^2$; in particular,  we shall work with the
orthogonal frame
\begin{equation*}
\Big\{e_{1}^{x}=\frac{\partial }{\partial \theta _{x}},\;e_{2}^{x}=\frac{%
\partial }{\partial \varphi _{x}}\Big\}.
\end{equation*}%

\noindent Since the $f_{\ell }$ are eigenfunctions of the
spherical Laplacian, we have that the value of the spherical harmonic at
every fixed point $x\in {\cal S}^2$ is a linear
combination of its first and second order derivatives at $x$. If
the point $x\in {\cal S}^2$ is also a critical point for $ f_{\ell
}$ it follows that the value of the spherical harmonic at $x$ is a
linear combination of its second order derivatives, i.e.,
\begin{equation} \label{linear_dp}
e_{1}^{x}e_{1}^{x}f_{\ell }(x)+e_{2}^{x}e_{2}^{x}f_{\ell }(x)=-\ell (\ell
+1)f_{\ell }(x).
\end{equation}
For $x\in {\cal S}^2$ we define the random vectors:
\begin{equation*}
Z_{\ell ;x}=(\nabla f_{\ell }(x),\nabla ^{2}f_{\ell }(x)),
\end{equation*}%
where
\begin{equation*}
\nabla f_{\ell }(x)=(e_{1}^{x}f_{\ell }(x),e_{2}^{x}f_{\ell }(x)),
\end{equation*}%
and $\nabla ^{2}f_{\ell }$ is defined as
\begin{equation*}
\nabla ^{2}f_{\ell }(x)=(e_{1}^{x}e_{1}^{x}f_{\ell
}(x),e_{1}^{x}e_{2}^{x}f_{\ell }(x),e_{2}^{x}e_{2}^{x}f_{\ell }(x)).
\end{equation*}
We denote by
\begin{equation*}
D_{\ell; x}(\xi _{x,1},\xi _{x,2},\zeta _{x,1},\zeta _{x,2},\zeta _{x,3}),
\end{equation*}%
the probability density functions of $Z_{\ell, x}$; the vectors
$Z_{\ell ,x}$ are centered Gaussian in $\mathbb{R}^{5}$. By the
isotropic property of $f_{\ell}$ it is possible and indeed convenient
to perform our  computations along a specific geodesic; we
constrain ourselves to the equatorial line
$\theta_x=\frac{\pi}{2}$. With this choice the $5\times 5$
covariance matrix $\sigma _{\ell }$ of $Z_{\ell ;x}$ is (see the
computations in Appendix \ref{cov_matx})
\begin{equation*}
\sigma _{\ell }=\left(
\begin{array}{cc}
a_{\ell } & b_{\ell } \\
b_{\ell }^{t} & c_{\ell }%
\end{array}%
\right),
\end{equation*}%
where
\begin{equation*}
a_{\ell }=\left(
\begin{array}{cc}
\frac{\lambda _{\ell }}{2} & 0 \\
0 & \frac{\lambda _{\ell }}{2}%
\end{array}%
\right), \hspace{1cm}
b_{\ell }=\left(
\begin{array}{ccc}
0 & 0 & 0 \\
0 & 0 & 0%
\end{array}%
\right),
\end{equation*}
and
\begin{align*}
c_{\ell }=\left(
\begin{array}{ccc}
\frac{\lambda _{\ell }}{8}[3\lambda _{\ell }-2] & 0 & \frac{\lambda _{\ell }%
}{8}[\lambda _{\ell }+2] \\
0 & \frac{\lambda _{\ell }}{8}[\lambda _{\ell }-2] & 0 \\
\frac{\lambda _{\ell }}{8}[\lambda _{\ell }+2] & 0 & \frac{\lambda _{\ell }}{%
8}[3\lambda _{\ell }-2]%
\end{array}%
\right)
=\frac{\lambda _{\ell }^{2}}{8}\left(
\begin{array}{ccc}
3-\frac{2}{\lambda _{\ell }} & 0 & 1+\frac{2}{\lambda _{\ell }} \\
0 & 1-\frac{2}{\lambda _{\ell }} & 0 \\
1+\frac{2}{\lambda _{\ell }} & 0 & 3-\frac{2}{\lambda _{\ell }}%
\end{array}%
\right),
\end{align*}
where $\lambda_\ell=\ell(\ell+1)$. From the isotropy
the following result follows at once: \newline

\begin{lemma}[Kac-Rice formula] \label{K-R_Exp}
The expected value of $\mathcal{N}_{I}^{c}(f_{\ell })$ is given by
\begin{equation*}
\mathbb{E}[\mathcal{N}_{I}^{c}(f_{\ell })]=4\pi K_{1,\ell }(I),
\end{equation*}
where
\begin{equation*}
K_{1,\ell }(I)=\int_{\mathbb{R}^{3}}|\zeta _{x,1}\zeta _{x,3}-\zeta
_{x,2}^{2}|1\hspace{-0.27em}\mbox{\rm l}_{\big\{\frac{\zeta _{x,1}+\zeta
_{x,3}}{\ell (\ell +1)}\in I\big\}}D_{\ell ;x}(0,0,\zeta _{x,1},\zeta
_{x,2},\zeta _{x,3})\;d\zeta _{x,1}\;d\zeta _{x,2}\;d\zeta _{x,3}.
\end{equation*}
\end{lemma}

\begin{proof} First, from \eqref{linear_dp}, we have
\begin{align*}
\mathcal{N}_{I}^{c}(f_{\ell })&=\#\{x\in {\cal S}^2: f_{\ell }(x)\in I, \nabla f_{\ell
}(x)=0\}\\
&=\#\Big\{x\in {\cal S}^2: -\frac{e_1^x e_1^xf_{\ell }(x)+e_2^x e_2^x f_{\ell }(x) }{\ell (\ell+1)} \in I, \nabla f_{\ell}(x)=0\Big\}.
\end{align*}
We can now apply Theorem 11.2.1 in \cite{adlertaylor}, and
get:
\begin{align*}
\mathbb{E}[\mathcal{N}_{I}^{c}(f_{\ell })]&= \int_{{\cal S}^2} dx  \int_{\mathbb{R}^3} |\zeta_{x,1} \zeta_{x,3} -\zeta_{x,2}^2| \ind_{\big\{-\frac{\zeta_{x,1} + \zeta_{x,3} }{\ell (\ell+1)} \in I \big\}} D_{\ell;x} (0 ,0,\zeta_{x,1},\zeta_{x,2}, \zeta_{x,3})\; d \zeta_{x,1}\; d \zeta_{x,2} \; d \zeta_{x,3}\\
&= \int_{{\cal S}^2} dx  \int_{\mathbb{R}^3} |\zeta_{x,1} \zeta_{x,3} -\zeta_{x,2}^2| \ind_{\big\{\frac{\zeta_{x,1} + \zeta_{x,3} }{\ell (\ell+1)} \in I \big\}} D_{\ell;x} (0 ,0,\zeta_{x,1},\zeta_{x,2}, \zeta_{x,3})\; d \zeta_{x,1}\; d \zeta_{x,2} \; d \zeta_{x,3}.
\end{align*}
By the isotropic property the density
\begin{align*}
K_{1,\ell}(I)=\int_{\mathbb{R}^3} |\zeta_{x,1} \zeta_{x,3}
-\zeta_{x,2}^2| \ind_{\big\{\frac{\zeta_{x,1} + \zeta_{x,3} }{\ell
(\ell+1)} \in I \big\}} D_{\ell;x} (0 ,0,\zeta_{x,1},\zeta_{x,2},
\zeta_{x,3})d\zeta _{x,1}\;d\zeta _{x,2}\;d\zeta _{x,3}
\end{align*}
does not depend on $x$, and the result of the present lemma follows.
\end{proof}

\begin{remark} For the critical points and the saddles we have the analogous result
\begin{equation*}
\mathbb{E}[\mathcal{N}_{I}^{a}(f_{\ell })]=4\pi K^a_{1,\ell }(I),
\end{equation*}
for $a=e,s$, where, for example,
\begin{align*}
K_{1,\ell }^{e}(I) =\int_{\mathbb{R}^{3}}|\zeta _{x,1}\zeta _{x,3}-\zeta
_{x,2}^{2}|1\hspace{-0.27em}\mbox{\rm l}_{\big\{\frac{\zeta _{x,1}+\zeta
_{x,3}}{\ell (\ell +1)}\in I\big\}}1\hspace{-0.27em}\mbox{\rm l}_{\big\{%
\zeta _{x,1}\zeta _{x,3}-\zeta _{x,2}^{2}>0\big\}}  D_{\ell ;x}(0,0,\zeta _{x,1},\zeta _{x,2},\zeta _{x,3})\;d\zeta
_{x,1}\;d\zeta _{x,2}\;d\zeta _{x,3}.
\end{align*}

\end{remark}

\subsection{Asymptotic density of critical points}

We will now exploit the Kac-Rice formula and the degeneracy discussed above for
spherical harmonics to prove our first result on the expected
number of critical points and extrema of $f_{\ell }$ with values lying in an
interval $I\subseteq \mathbb{R}$.

\begin{lemma} \label{expectation}
For $\ell \rightarrow \infty$, we have
\begin{equation*}
\mathbb{E}[\mathcal{N}_{I}^{c}(f_{\ell })]=\frac{\ell ^{2}}{2}%
\int_{I}p_{1}^{c}(t)dt+O(1),
\end{equation*}%
where
\begin{equation*}
p_{1}^{c}(t)=\frac{1}{(2\pi )^{3/2}}\int_{\mathbb{R}^{2}}|z_{1}t\sqrt{8}%
-z_{1}^{2}-z_{2}^{2}|\exp \left\{-\frac{3}{2}t^{2}\right\} \exp \left\{-\frac{1}{2}%
(z_{1}^{2}+z_{2}^{2}-\sqrt{8}tz_{1})\right\} dz_{1}dz_{2}.
\end{equation*}
\end{lemma}

\begin{proof} From Lemma \ref{K-R_Exp}, we have
\begin{align} \label{xx}
\mathbb{E}[\mathcal{N}^c_I(f_\ell)]&=4 \pi K_{1,\ell}(I).
\end{align}
Since the first and the second order derivatives of $f_{\ell}(x)$ are independent at every fixed point $x \in {\cal S}^2$, we can write
\begin{align*}
K_{1,\ell}(I)&=\int_{\mathbb{R}^3} |\zeta_{1} \zeta_{3} -\zeta_{2}^2| \ind_{\big\{\frac{\zeta_{1} + \zeta_{3} }{\lambda_\ell} \in I \big\}} D_{\ell} (0 ,0,\zeta_{1},\zeta_{2}, \zeta_{3}) \; d \zeta_{1}\; d \zeta_{2} \; d \zeta_{3}
\end{align*}
with
$$D_{\ell} (0 ,0,\zeta_{1},\zeta_{2}, \zeta_{3})=D_{1,\ell} (0 ,0)  D_{2,\ell} (\zeta_{1},\zeta_{2}, \zeta_{3}),$$
where $D_{1,\ell}$ and $D_{2,\ell}$ are the marginal densities of the random vectors $\nabla f_\ell$ and $\nabla^2 f_\ell$ respectively. Now, observing the matrices $a_\ell$ and $c_\ell$, it follows immediately that
$$D_{1,\ell}(0,0)=\frac{1}{2 \pi} \frac{2}{\lambda_\ell},$$
and
$$\frac{\sqrt 8}{\lambda_\ell} \nabla^2 f_\ell=(\tilde{Z}_1,\tilde{Z}_2,\tilde{Z}_3) \sim N(0,\tilde c_\ell),$$
with
\begin{align} \label{xxx}
\tilde{c}_\ell= \frac{8}{\lambda_\ell^2}  c_\ell= \left(
\begin{array}{ccc}
3 -\frac{2}{\lambda_\ell } & 0 & 1 +\frac{2}{\lambda_\ell} \\
0 & 1 -\frac{2}{\lambda_\ell} & 0 \\
1 +\frac{2}{\lambda_\ell} & 0 & 3 -\frac{2}{\lambda_\ell}%
\end{array}%
\right).
\end{align}
It then follows that
\begin{align*}
K_{1,\ell}(I)&=  \frac{1}{ \pi \lambda_\ell} \int_{\mathbb{R}^3} |\zeta_{1} \zeta_{3} -\zeta_{2}^2| \ind_{\big\{\frac{\zeta_{1} + \zeta_{3} }{\lambda_\ell} \in I \big\}}  D_{2,\ell} (\zeta_{1},\zeta_{2}, \zeta_{3})\; d \zeta_{1}\; d \zeta_{2} \; d \zeta_{3} \nonumber \\
&=  \frac{1}{ \pi \lambda_\ell} \frac{\lambda_\ell^2}{8} \int_{\mathbb{R}^3} |\tilde{z}_{1} \tilde{z}_{3} -\tilde{z}_{2}^2| \ind_{\big\{\frac{\tilde{z}_{1} + \tilde{z}_{3} }{\sqrt{8}} \in I \big\}} \frac{1}{(2 \pi)^{3/2} \sqrt{\det \tilde{c}_\ell}}
\exp \left\{-\frac 1 2 ( \tilde{z}_1, \tilde{z}_2 ,\tilde{z}_3)
\tilde{c}_\ell^{-1} \left(
\begin{array}{c}
\tilde{z}_1 \\
\tilde{z}_2 \\
\tilde{z}_3
\end{array}
\right) \right\} d \tilde{z}_1 d\tilde{z}_2  d\tilde{z}_3.
\end{align*}
After the change of variables
\begin{align*} 
\left(
\begin{array}{c}
\tilde{z}_1  \\
\tilde{z}_2 \\
\tilde{z}_3
\end{array}
\right)=\left(
\begin{array}{ccc}
1&0&0\\
0&1&0 \\
-1&0&\sqrt 8
\end{array}
\right)  \left(
\begin{array}{c}
z_1  \\
z_2 \\
t
\end{array}
\right),
\end{align*}
the latter expression is
\begin{align}
K_{1,\ell}(I)&=  \frac{1}{ \pi \lambda_\ell} \frac{\lambda_\ell^2}{8} \int_{\mathbb{R}^3} |{z}_{1} (\sqrt{8} t-z_1) -{z}_{2}^2| \ind_{\{t \in I\}} \nonumber\\
&\;\;\;\; \times  \frac{1}{(2 \pi)^{3/2} \sqrt{\det \tilde{c}_\ell}}
\exp \left\{-\frac 1 2 ( z_1, z_2 ,\sqrt{8} t -z_1)
\tilde{c}_\ell^{-1} \left(
\begin{array}{c}
z_1 \\
z_2 \\
\sqrt{8} t -z_1
\end{array}
\right) \right\} \sqrt 8 d {z}_1 d{z}_2  d t \nonumber \\
&=  \frac{1}{ \pi } \frac{\lambda_\ell}{ 8} \int_I  d t  \int_{\mathbb{R}^2} |{z}_{1} (\sqrt{8} t -z_1) -{z}_{2}^2| \nonumber \\
&\;\;\;\; \times  \frac{1}{(2 \pi)^{3/2} \sqrt{\det \tilde{c}_\ell}}
\exp \left\{-\frac 1 2 ( z_1, z_2 ,\sqrt{8} t -z_1)
\tilde{c}_\ell^{-1} \left(
\begin{array}{c}
z_1 \\
z_2 \\
\sqrt{8} t -z_1
\end{array}
\right) \right\}  \sqrt{8} d {z}_1 d{z}_2,   \label{xxxx}
\end{align}
where
\begin{align*}
\frac{1}{\sqrt{\det \tilde{c}_\ell}}= \frac{\lambda_\ell}{\sqrt 8 (\lambda_\ell-2)},
\end{align*}
and
$$
\exp \left\{-\frac 1 2 ( z_1, z_2 ,\sqrt{8} t -z_1)
\tilde{c}_\ell^{-1} \left(
\begin{array}{c}
z_1 \\
z_2 \\
\sqrt{8} t -z_1
\end{array}
\right) \right\}   =\exp\left\{- \frac{3 \lambda_\ell-2}{2(\lambda_\ell-2)} t^2 \right\}
\exp\left\{-\frac 1 2 \frac{\lambda_\ell}{\lambda_\ell-2} (z_1^2+z_2^2-\sqrt 8 z_1 t ) \right\}
$$
so that
\begin{align*}
K_{1,\ell}(I)&=  \frac{1}{ \pi } \frac{\lambda_\ell}{8}  \frac{\lambda_\ell}{\lambda_\ell-2} \int_I  dt \int_{\mathbb{R}^2} |{z}_{1} (\sqrt{8} t -z_1) -{z}_{2}^2| \\
&  \;\;\; \times \frac{1}{(2 \pi)^{3/2} }
\exp\left\{- \frac{3 \lambda_\ell-2}{2(\lambda_\ell-2)} t^2 \right\}
\exp\left\{-\frac 1 2 \frac{\lambda_\ell}{\lambda_\ell-2} (z_1^2+z_2^2-\sqrt 8 z_1 t) \right\}  d {z}_1 d{z}_2.
\end{align*}
Now let us write
\begin{align}
\label{eq:EN=lambda*int g1l}
\mathbb{E} [\mathcal{N}^c_I(f_\ell)] = \frac{\lambda_\ell}{ 2 } \int_I
g^c_{1,\ell}(t) d t,
\end{align}
where
\begin{align} \label{g^c}
g^c_{1,\ell}(t)&=\frac{1}{(2 \pi)^{3/2}} \frac{\lambda_\ell}{
\lambda_\ell-2} \int_{\mathbb{R}^2} | z_1 t \sqrt 8 -z_1^2-z_2^2| \exp\left\{ - \frac{3 \lambda_\ell-2}{2(\lambda_\ell-2)} t^2
\right\} \exp\left\{-\frac 1 2 \frac{\lambda_\ell}{\lambda_\ell-2} (z_1^2+
z_2^2-\sqrt 8 t z_1)\right\}d z_1 d z_2.
\end{align}
Now consider the expansions
$$\frac{\lambda_\ell}{\lambda_\ell-2}=1+O(\ell^{-2}), \hspace{1cm} \exp\left\{ - \frac{3 \lambda_\ell-2}{2(\lambda_\ell-2)} t^2\right\} =
\exp\left\{ - \frac{3}{2} t^2\right\} +t^2 \exp\left\{ - \frac{3}{2} t^2\right\} O(\ell^{-2}),
$$
and
$$\exp\left\{-\frac 1 2 \frac{\lambda_\ell}{\lambda_\ell-2} (z_1^2+
z_2^2-\sqrt 8 t z_1)\right\} = \exp\left\{-\frac{1}{2} (z_1^2+
z_2^2-\sqrt 8 t z_1)\right\}+(z_1^2+
z_2^2-\sqrt 8 t z_1) \exp\left\{-\frac{1}{2} (z_1^2+
z_2^2-\sqrt 8 t z_1)\right\} O(\ell^{-2});$$
we can observe that, for $n, n', k, k' \in \mathbb{N}$,
$$ \exp\left\{ - \frac{3}{2} t^2\right\} \int_{\mathbb{R}} |z_2|^n z_2^{n'} t^k \exp\left\{-\frac{z_2^2}{2} \right\} d z_2$$
and
$$ \exp\left\{ - \frac{3}{2} t^2\right\} \int_{\mathbb{R}} |z_1|^n z_1^{n'} |t|^k t^{k'} \exp\left\{-\frac{1}{2} (z_1^2-\sqrt 8 t z_1)\right\} d z_1
$$
are bounded by terms of the form
$\text{const} \times t^k e^{-\frac{3}{2} t^2}$ and $\text{const} \times|t|^{k+k'} e^{-\frac{t^2}{2} } $ respectively; hence we have
\begin{align}
\label{eq:int g1l=int p1+O(l-2)}
\int_{I} g^c_{1,\ell}(t) d t&=\int_{ I} p^c_1(t) d t+O(\ell^{-2}),
\end{align}
where
\begin{align*}
p^c_1(t)&=\frac{1}{(2 \pi)^{3/2}}  \int_{\mathbb{R}^2} |z_1 t \sqrt 8 -z
_1^2-z_2^2| \exp\left\{ - \frac{3 }{2} t^2
\right\} \exp\left\{-\frac 1 2  (z_1^2+%
z_2^2-\sqrt 8 t z_1)\right\}d z_1 d z_2.
\end{align*}
We finally obtain the statement of the present lemma by substituting \eqref{eq:int g1l=int p1+O(l-2)} into \eqref{eq:EN=lambda*int g1l}.

\end{proof}

\begin{remark}
Introducing the corresponding conditions on the Hessian and following the lines of the previous proof we get the analogous result for extrema and saddles, i.e.,
\begin{align*}
\mathbb{E}[{\cal N}_I^a(f_{\ell})]=\frac{\ell^2}{2} \int_{I} p_1^a(t) d t+O(1),
\end{align*}
where, for $a=e,s$, we have
\begin{align*}
 p_1^e(t)= \frac{1}{(2 \pi)^{3/2}}  \int_{\mathbb{R}^2} |z_1 t \sqrt 8 -z
_1^2-z_2^2| \ind_{\{ z_1 t \sqrt 8 -z
_1^2-z_2^2>0 \}}\exp\left\{ - \frac{3 }{2} t^2
\right\} \exp\left\{-\frac 1 2  (z_1^2+%
z_2^2-\sqrt 8 t z_1)\right\}d z_1 d z_2,
\end{align*}
and
\begin{align*}
 p_1^s(t)= \frac{1}{(2 \pi)^{3/2}}  \int_{\mathbb{R}^2} |z_1 t \sqrt 8 -z
_1^2-z_2^2| \ind_{\{ z_1 t \sqrt 8 -z
_1^2-z_2^2 < 0 \}}\exp\left\{ - \frac{3 }{2} t^2
\right\} \exp\left\{-\frac 1 2  (z_1^2+%
z_2^2-\sqrt 8 t z_1)\right\}d z_1 d z_2.
\end{align*}
\end{remark}

We can now prove Proposition \ref{expectation copy(1)} by deriving
explicit expressions for $p_{1}^{c}$, $p_{1}^{e}$ and $p_{1}^{s}$
(see also \cite{azais} for alternative techniques in a related
setting). For this purpose, let $Y=(Y_{1},Y_{2},Y_{3})$ be a
centered jointly Gaussian random vector with covariance matrix
\begin{equation*}
\tilde{c}_{\infty }=\left(
\begin{array}{ccc}
3 & 0 & 1 \\
0 & 1 & 0 \\
1 & 0 & 3%
\end{array}%
\right) .
\end{equation*}%
Denote by $\phi _{Y_{1}+Y_{3}}$ the probability density of $Y_{1}+Y_{3}$.
The proof of Proposition \ref{expectation copy(1)} is given below.
\begin{proof}[Proof of Proposition  \ref{expectation copy(1)}]
From Lemma \ref{expectation}, and in particular from \eqref{xxxx}, we observe that
\begin{align*}
p^c_{1}(t)&=\sqrt{8}  \cdot  \mathbb{E} \left[ \left|
Y_{1}Y_{3}-Y_{2}^{2}\right| \bigg\vert Y_{1}+Y_{3}= \sqrt{8} t \right]
\cdot \phi _{Y_{1}+Y_{3}}(\sqrt 8 t).
\end{align*}
Similarly we write
\begin{align*}
p^e_1(t)&=\sqrt 8 \cdot \mathbb{E} \left[  \left|
Y_{1}Y_{3}-Y_{2}^{2} \right| \ind_{\{Y_{1}Y_{3}-Y_{2}^{2}>0\}} \bigg| Y_{1}+Y_{3}= \sqrt 8 t \right] \cdot \phi _{Y_{1}+Y_{3}}( \sqrt 8 t),
\\ p^s_1(t)&=\sqrt 8 \cdot \mathbb{E} \left[  \left|
Y_{1}Y_{3}-Y_{2}^{2} \right| \ind_{\{Y_{1}Y_{3}-Y_{2}^{2}<0\}} \bigg| Y_{1}+Y_{3}= \sqrt 8 t \right] \cdot \phi _{Y_{1}+Y_{3}}( \sqrt 8 t).
\end{align*}
Now consider the transformation $W_{1}=Y_{1}$, $W_{2}=Y_{2}$ and $W_{3}=Y_{1}+Y_{3}$, i.e. the vector $W=(W_1,W_2,W_3)$ is given by
\[
W=\left(
\begin{array}{ccc}
1 & 0 & 0 \\
0 & 1 & 0 \\
1 & 0 & 1
\end{array}
\right) Y;
\]
the covariance matrix $\Sigma _{W}$ of $W$ is
\[
\Sigma _{W}=\left(
\begin{array}{ccc}
1 & 0 & 0 \\
0 & 1 & 0 \\
1 & 0 & 1
\end{array}
\right) \left(
\begin{array}{ccc}
3 & 0 & 1 \\
0 & 1 & 0 \\
1 & 0 & 3
\end{array}
\right) \left(
\begin{array}{ccc}
1 & 0 & 1 \\
0 & 1 & 0 \\
0 & 0 & 1
\end{array}
\right) =\left(
\begin{array}{ccc}
3 & 0 & 4 \\
0 & 1 & 0 \\
4 & 0 & 8
\end{array}
\right).
\]
Under the obvious notation we write
\begin{align*}
\Sigma _{(W_1,W_2)} &=\left(
\begin{array}{cc}
3 & 0 \\
0 & 1
\end{array}
\right), \hspace{0.5cm} \Sigma _{W_3}=8,
\end{align*}
so that the conditional distribution of $(W_1,W_2)| W_3=\sqrt 8 t$ is Gaussian with covariance matrix
\begin{align*}
\Sigma _{(W_1,W_2)|W_3} =\left(
\begin{array}{cc}
3 & 0 \\
0 & 1
\end{array}
\right) -\left(
\begin{array}{c}
4 \\
0
\end{array}
\right) \frac{1}{8}\left(
\begin{array}{cc}
4 & 0
\end{array}
\right) =\left(
\begin{array}{cc}
1 & 0 \\
0 & 1
\end{array}
\right),
\end{align*}
and expectation
\[
\mathbb{E}  [(W_1,W_2)| W_3=\sqrt 8 t] =\left(
\begin{array}{c}
4 \\
0
\end{array}
\right) \frac{1}{8}\sqrt{8}t=\left(
\begin{array}{c}
\sqrt{2}t \\
0
\end{array}
\right).
\]
Hence we have
\begin{align*}
\mathbb{E}\left[ \left. \left\vert Y_{1}Y_{3}-Y_{2}^{2}\right\vert
\right\vert Y_{1}+Y_{3}=\sqrt{8}t\right] &=\mathbb{E}\left[ \left. \left\vert
W_{1}(W_{3}-W_{1})-W_{2}^{2}\right\vert \right\vert W_{3}=\sqrt{8}t\right] \\
&=\mathbb{E}\left[ \left. \left\vert \sqrt{8}tW_{1}-W_{1}^{2}-W_{2}^{2}%
\right\vert \right\vert W_{3}=\sqrt{8}t\right]  \\
&=\mathbb{E}\left[ \left\vert \sqrt{8}t(Z_{1}+\sqrt{2}t)-(Z_{1}+\sqrt{2}%
t)^{2}-Z_{2}^{2}\right\vert \right]  \\
&=\mathbb{E}\left[ \left\vert -Z_{1}^{2}-Z_{2}^{2}+2t^{2}\right\vert \right],
\end{align*}
where $Z_{1},Z_{2}$ denote standard independent Gaussian variables. We can
then implement a further change of variable
$\zeta =Z_{1}^{2}+Z_{2}^{2}$ with the probability density function
$f_{\zeta }(z)=\frac{1}{2} e^{-\frac{z}{
2}}$.
Hence
\begin{align*}
\mathbb{E}\left[ \left\vert -Z_{1}^{2}-Z_{2}^{2}+2t^{2}\right\vert \right] &=
\mathbb{E}\left[ \left\vert -\zeta +2t^{2}\right\vert \right] =\frac{1}{2} \int_{0}^{2t^{2}}(2t^{2}-z)  e^{-\frac{z}{2}}dz+\frac{1}{2} \int_{2t^{2}}^{
\infty }(z-2t^{2})   e^{-\frac{z}{2}} dz,
\end{align*}
where
\begin{align*}
\frac{1}{2} \int_{0}^{2t^{2}}(2t^{2}-z)   e^{-\frac{z}{2}} dz
&=2(t^{2}-1)+2e^{-t^{2}}.
\end{align*}
Likewise, with the change of variable $y=\frac{z-2 t^2}{2}$,
\begin{align*}
\frac{1}{2} \int_{2t^{2}}^{\infty }(z-2t^{2}) e^{-\frac{z}{2}} dz
&=2e^{-t^{2}}\int_{0}^{\infty }y e^{-y}dy=2e^{-t^{2}}.
\end{align*}
So we have
\begin{align*}
\left.  \mathbb{E}\left[ \left\vert Y_{1}Y_{3}-Y_{2}^{2}\right\vert \right\vert Y_{1}+Y_{3}= \sqrt 8 t \right]
=2 (2 e^{-t^2}+t^2-1),
\end{align*}
and
\begin{align*}
p^c_{1}(t)=\sqrt 8 \left.  \mathbb{E}\left[ \left\vert
Y_{1}Y_{3}-Y_{2}^{2}\right\vert \right\vert Y_{1}+Y_{3}= \sqrt 8 t \right]
\phi _{Y_{1}+Y_{3}}(\sqrt 8 t)
&= \frac{\sqrt{2}}{ \sqrt{  \pi}} (2 e^{-t^2}+t^2-1) e^{-\frac{t^2}{2}},
\end{align*}
in fact, since $Y_1+Y_3$ is a centered Gaussian with variance $8$, we have
$$\phi_{Y_1+Y_3}(\sqrt 8 t)=\frac{1}{\sqrt 8 \sqrt{2 \pi}}e^{-\frac{(\sqrt 8 t)^2}{2 \cdot 8}}=\frac{1}{4 \sqrt \pi} e^{-\frac{t^2}{2}}.$$
Similarly, for the extrema we obtain
\begin{align*}
p^e_{1}(t)
&= \sqrt 8 \mathbb{E}\left[ \left\vert -Z_{1}^{2}-Z_{2}^{2}+2t^{2}\right\vert  \ind_{ \{ -Z_{1}^{2}-Z_{2}^{2}+2t^{2} >0\} }   \right] \phi _{Y_{1}+Y_{3}}(\sqrt 8 t)
=\frac{\sqrt 2}{ \sqrt \pi} (e^{-t^2}+t^2-1) e^{-\frac{t^2}{2}}.
\end{align*}
\end{proof}

\begin{remark}
From the expressions for $p_{1}^{c}$ and $p_{1}^{e}$ we immediately obtain
an expression for $p_{1}^{s}$:
\begin{equation*}
p_{1}^{s}(t)=p_{1}^{c}(t)-p_{1}^{e}(t)=\frac{\sqrt 2 }{\sqrt{\pi }}e^{-\frac{3}{2}%
t^{2}}.
\end{equation*}
\end{remark}

\begin{remark}
As mentioned in the introduction, the distribution that we found
cannot be viewed as a special case of the general result which has
recently been established on the sphere by \cite{chengschwartzman}
for real-valued, $C^2$ Gaussian random fields $\{f(t), t \in T
\}$, $T \subseteq \mathbb{R}^N$. This is because condition C3' on
page 15 of \cite{chengschwartzman} is not satisfied for random
spherical harmonics. Indeed,
following their notation let us write, for $i,j,k,l=1,\dots, N$,
\begin{align*}
f_i(t)=\frac{\partial f(t)}{ \partial t_i}, \hspace{2cm} f_{i,j}(t)=\frac{\partial^2 f(t)}{ \partial t_i t_j},
\end{align*}
and define $C', C''$ such that
\begin{align*}
\mathbb{E}[f_i(t) f_j(t)]=C' \delta_{ij}, \hspace{2cm} \mathbb{E}[f_{ij}(t) f_{kl}(t)]=C''( \delta_{ik} \delta_{jl}+\delta_{il} \delta_{jk})+(C''+C') \delta_{ij} \delta_{kl}.
\end{align*}
Then Condition  C3'  in \cite{chengschwartzman} states that
$C''+C'-(C')^2 \ge 0$; on the other hand in our case of spherical
harmonics we have
\begin{align*}
C''=\frac{\lambda_{\ell}^2}{8}-\frac{\lambda_{\ell}}{4}, \hspace{1cm} C'=\frac{\lambda_{\ell}}{2},
\end{align*}
so the quantity $C''+C'-(C')^2$ is in this case equal to
$(-\ell^4-2 \ell^3-\ell^2+2 \ell)/8$ which is negative for $\ell >
1$. Hence the limiting distribution in \cite{chengschwartzman}
Theorem 3.10, which depends on the square root of $C''+C'-(C')^2$,
is not applicable in our setting.
\end{remark}

\section{Approximate Kac-Rice for variance computation} \label{tre}

\subsection{On the Kac-Rice formula for computing $2$nd (factorial) moment}

In the setting of section \ref{sec:Kac-Rice expected}, $\mathcal{E}\subseteq \R^{n}$ a nice Euclidian domains,
and $g:\mathcal{E}\rightarrow \R^{n}$ a centred Gaussian
random field, a.s. smooth, define the $2$-point correlation function of critical points
(also referred to as ``2nd intensity") $K_{2}=K_{2;g}:\mathcal{E}^{2}\rightarrow\R$
\begin{equation*}
K_{2}(x,y) = \phi_{(g(x),g(y))}(\mathbf{0},\mathbf{0})\cdot
\E[ |\det J_{g}(x) | \cdot |\det J_{g}(y) | \big| g(x)=g(y)=\mathbf{0}  ].
\end{equation*}
By the virtue of ~\cite{adlertaylor}, Theorem 11.2.1, the $2$nd
factorial moment of $g^{-1}(0)$ is given by
\begin{equation*}
\E[\# g^{-1}(0)\cdot (\# g^{-1}(0)-1)] =
\int\limits_{\mathcal{E}^{2}}K_{2}(x,y)dxdy,
\end{equation*}
provided that the Gaussian distribution of $(g(x),g(y),J_{g}(x),J_{g}(y))\in \R^{2n}\times\R^{n(n+1)}$
is non-degenerate for all $(x,y)\in\mathcal{E}^{2}$. Moreover, for
$\mathcal{D}_{1},\mathcal{D}_{2}\subseteq\mathcal{E}$ two nice {\em disjoint}
domains, we have
\begin{equation}
\label{eq:Kac-Rice 2nd fact disjoint}
\E[\# g^{-1}(0)\cap \mathcal{D}_{1}\cdot (\# g^{-1}(0)\cap \mathcal{D}_{2})] =
\iint\limits_{\mathcal{D}_{1}\times \mathcal{D}_{2}}K_{2}(x,y)dxdy,
\end{equation}
under the same non-degeneracy assumption for all $(x,y)\in \mathcal{D}_{1}\times \mathcal{D}_{2}$.

For the critical points of $f=f_{\ell}$ we have
\begin{equation}
\label{eq:K2 glob gen}
K_{2}(x,y)=K_{2;l}(x,y) = \phi_{(\nabla f(x),\nabla f(y))}(\mathbf{0},\mathbf{0})\cdot
\E[ |\det H_{f}(x) | \cdot |\det H_{f}(y) | \big| \nabla f(x)=\nabla f(y)=\mathbf{0}  ];
\end{equation}
by the isotropy $K_{2}(x,y)=K_{2}(d(x,y))$ depends only on the (spherical) distance
between $x$ and $y$. Here ~\cite{adlertaylor}, Theorem 11.2.1 would yield
\begin{equation}
\label{eq:Kac-Rice 2nd mom gen}
\E[\mathcal{N}_{\R}^{c}(f)\cdot (\mathcal{N}_{\R}^{c}(f)-1)] =
\iint\limits_{\mathcal{S}^{2}\times \mathcal{S}^{2}}K_{2}(x,y)dxdy,
\end{equation}
provided that for all $x,y\in \mathcal{S}^{2}$, the Gaussian distribution of
$\left(\nabla f(x), \nabla f(y), H_{f}(x),H_{f}(y)\right)\in \R^{10}$ is non-degenerate.

Unfortunately, we were not able to validate the non-degeneracy assumption due to the technical difficulty
of dealing with complicated $10\times 10$ matrices depending on both $x$ and $y$ (and $\ell$).
Instead, we will prove that the (precise) Kac-Rice formula \eqref{eq:Kac-Rice 2nd mom gen} holds up to an
admissible error, i.e. an {\em approximate Kac-Rice} (formula \eqref{afkao} below),
an approach inspired from \cite{rudnickwigman};
our argument is based on a partitioning of the integration domain of \eqref{eq:Kac-Rice 2nd mom gen}
and applying \eqref{eq:Kac-Rice 2nd fact disjoint} on the valid slices, bounding the contribution of the rest.
It is easy to adapt the definition of the $2$-point correlation in \eqref{eq:K2 glob gen} in order
to count critical points with values lying in $I$, or separate the critical points into extrema and saddles
(cf. \eqref{k2elegant} below).

\subsection{Statement of the principal formula}

In this section we shall formulate the approximate Kac-Rice formula which is
instrumental for our main result. First we need to introduce some more
notation; define the function
\begin{align} \label{bachata}
L_{2,\ell}(\phi ;t_{1},t_{2})&=\frac{1}{2}\sin ^{4}\phi [P_{\ell }^{\prime \prime
}(\cos \phi )]^2 v_{1}(t_{1},t_{2})
-\frac{32}{\ell^2} \sin ^{6}\phi [ P_{\ell }^{\prime \prime \prime }(\cos \phi
)]^{2} v_{2}(t_{1},t_{2})+\frac{64}{\ell^4} \sin ^{8}\phi [
P_{\ell }^{\prime \prime \prime \prime }(\cos \phi )]^{2}
v_{3}(t_{1},t_{2}), 
\end{align}
where
\begin{align*}
v_{1}(t_{1},t_{2})= p_1^c(t_1) p_1^c(t_2),
\end{align*}
\begin{align*}
v_{2}(t_{1},t_{2})& = \frac{1}{8^2} \Big[ -3 p_1^c(t_1) p_1^c(t_2) +\frac 1 2 p_2^c(t_1) p_1^c(t_2) + \frac 1 2 p_1^c(t_1) p_2^c(t_2) \Big],
\end{align*}%
and%
\begin{align*}
 v_{3}(t_{1},t_{2})&= \frac{1}{8^2} \Big[  \frac 3 8 p_1^c(t_1) - \frac 1 8 p_2^c(t_1) \Big] \Big[  \frac 3 8 p_1^c(t_2) - \frac 1 8 p_2^c(t_2) \Big] .
\end{align*}%
We are now in a position to formulate the \textit{Approximate Kac-Rice
formula:}

\begin{proposition} \label{P-afkao}
For any sufficiently big constant $C>0$, the variance of the critical points number $\mathcal{N}_{I}^{c}(f_{\ell })$ satisfies
\begin{equation}
\text{Var} \left( \mathcal{N}_{I}^{c}(f_{\ell })\right) =\int_{C/\ell}^{\pi
/2} \iint_{I \times I} L_{2,\ell}(\phi ;t_{1},t_{2}) d t_1 dt_2  \sin \phi d\phi +O(\ell ^{5/2}). \label{afkao}
\end{equation}
\end{proposition}

\noindent The rest of the present section is dedicated to proving formula \eqref{afkao}.

\subsection{Two-point correlation function} \label{3.22222}

Here we formulate some auxiliary results instrumental for our
main argument below; our aim is to
write an approximate formula for the variance as an integral of the two-point
correlation function $K_{2,\ell}$ defined by
\begin{equation} \label{k2elegant}
K_{2,\ell}(x,y;t_{1},t_{2})=
\mathbb{E}\left[  \left\vert \nabla ^{2}f_{\ell }(x)\right\vert
\cdot \left\vert \nabla ^{2}f_{\ell }(y)\right\vert \Big|  \nabla
f_{\ell }(x)=\nabla f_{\ell }(y)={\bf 0},f_{\ell }(x)=t_{1},f_{\ell
}(y)=t_{2}\right] \cdot \varphi _{x,y, \ell}(t_{1},t_{2}, \mathbf{0},\mathbf{0}),
\end{equation}
where $\varphi _{x,y, \ell}(t_{1},t_{2},\mathbf{0},\mathbf{0})$ denotes
the density of the 6-dimensional vector
\begin{equation*}
\left( f_{\ell }(x),f_{\ell }(y),\nabla f_{\ell }(x),\nabla
f_{\ell }(y)\right)
\end{equation*}
in $f_{\ell }(x)=t_{1},f_{\ell }(y)=t_{2},\nabla f_{\ell }(x)=\nabla f_{\ell}(y)={\bf 0}$. Note that, by the isotropy, the function $K_{2,\ell}$ depends on the points $x$, $y$ only
via their geodesic distance $\phi=d(x,y)$; by abuse of notation we write
\begin{equation*}
K_{2,\ell}(\phi; t_{1},t_{2})=K_{2,\ell}(x,y;t_{1},t_{2}).
\end{equation*}%
Also, we note that $K_{2,\ell}(\phi ;t_{1},t_{2})$ is everywhere nonnegative.
We shall need several results:

\begin{lemma} \label{e}
There exists a constant $C>0$ sufficiently big, such that for every
nice domains ${\cal D}_{1},{\cal D}_{2}\subseteq {\cal S}^2$ with distance $d({\cal D}_{1},{\cal D}_{2})>C/\ell$, we have
\begin{equation*}
{\text Cov} \left( \mathcal{N}^{c}(f_{\ell };{\cal D}_{1},I),\mathcal{N}^{c}(f_{\ell
};{\cal D}_{2},I)\right) =\iint_{{\cal D}_{1}\times {\cal D}_{2}}\iint_{I\times
I}K_{2,\ell}(x , y ;t_{1},t_{2})dt_{1}dt_{2}dx dy.
\end{equation*}
\end{lemma}

\begin{proposition}[Long-range asymptotics of the 2-point correlation function] \label{key-c}
There exists a constant $C>0$, such that for $d(x,y)>C/\ell$, one has:
\begin{equation*}
16 \pi^2K_{2,\ell}(x,y;t_{1},t_{2})=L_{2,\ell}(x,y;t_{1},t_{2})+\frac{\ell ^{4}}{4} p_{1}^{c}(t_{1}) p_{1}^{c}(t_{2})+E_{2,\ell}(x,y;t_{1},t_{2}),
\end{equation*}%
where $L_{2,\ell}$ is as in \eqref{bachata} and the error term $E_{2,\ell}$ is such that
\begin{equation}
\iint_{d(x,y)>C/\ell}\iint_{\mathbb{R} \times \mathbb{R}}\left\vert
E_{2,\ell}(x,y;t_{1},t_{2})\right\vert dt_{1}dt_{2}dxdy=O(\ell ^{5/2}).
\label{integralerror}
\end{equation}
\end{proposition}

\begin{lemma} \label{d}
For any constant $C>0$, we have
\begin{equation*}
\int_{\mathbb{R}^{2}}\left\vert L_{2,\ell}(x,y;t_{1},t_{2})\right\vert
dt_{1}dt_{2}=O(\ell ^{4}),
\end{equation*}%
uniformly for $\ell \geq 1$, $d(x,y)>C/\ell$.
\end{lemma}

\begin{proposition} \label{a}
There exist a constant $c>0$ such that for every
nice domain $\mathcal{D\subseteq S}^{2}$ contained in some spherical cap of
radius $c/\ell$, one has
\begin{equation*}
\mathbb{E}\left[ \mathcal{N}^{c}(f_{\ell };\mathcal{D},I)\left( \mathcal{N}
^{c}(f_{\ell };\mathcal{D},I)-1\right) \right] =\iint_{\mathcal{D}\times
\mathcal{D}}\iint_{I\times
I}K_{2,\ell}(x,y;t_{1},t_{2})dt_{1}dt_{2}dx dy
\end{equation*}
\end{proposition}

\begin{lemma} \label{b}
There exists a constant $c>0$ such that, for $d(x,y)<c/\ell$, one has
\begin{equation*}
\int_{\mathbb{R}^2
}K_{2,\ell}(x,y;t_{1},t_{2})dt_{1}dt_{2} = O(\ell^4),
\end{equation*}%
where the constant involved in the $O$-notation is universal.
\end{lemma}

\noindent The proofs of all the results given in section \ref{3.22222} are deferred
to section \ref{asymptoticsection}.

\subsection{Proof of Proposition \ref{P-afkao}}

\subsubsection{Partition of the sphere into Voronoi cells}

We introduce the following notation for the spherical caps on ${\cal S}^2$:
\begin{align*}
{\cal B}(a,\varepsilon)=\{x \subseteq {\cal S}^2: d(a,x) \le \varepsilon\}.
\end{align*}
For any $\varepsilon>0$, we say that
$\Xi_\varepsilon=\{\xi_{1,\varepsilon}, \dots,
\xi_{N,\varepsilon}\}$ is a maximal $\varepsilon$-net, if
$\xi_{1,\varepsilon}, \dots, \xi_{N,\varepsilon}$ are in ${\cal
S}^2$, $\forall i \ne j$ we have $d(\xi_{i,\varepsilon},
\xi_{j,\varepsilon})>\varepsilon$ and
$$\forall x \in {\cal S}^2,\;\;  d(x,\Xi_\varepsilon) \le \varepsilon,\;\; \bigcup_{\xi_{i,\varepsilon} \in \Xi_\varepsilon} {\cal B}(\xi_{i,\varepsilon},\varepsilon)={\cal S}^2,$$
$$\forall i \ne j, \;\; {\cal B}(\xi_{i,\varepsilon},\varepsilon/2) \cap {\cal B}(\xi_{j,\varepsilon},\varepsilon/2)=\emptyset.$$
Heuristically, an $\varepsilon$-net is a grid of point at a
distance at least $\varepsilon$ from each other, and such that any
additional point should be within distance $\varepsilon$ from a point
on the grid, see \cite{MaPeCUP}. The number $N$ of points in a
$\varepsilon$-net on the sphere is necessarily commensurable to ${1}/{\varepsilon^{2}}$;
more precisely we have the following:
$$\frac{4}{\varepsilon^2} \le N \le \frac{4}{\varepsilon^2} \pi^2,$$
see \cite{BKMP}, Lemma 5. Given an $\varepsilon$-net it is natural to partition the sphere into its Voronoi cells, defined below,
each associated to a single point on the net; they are disjoint save to boundary overlaps.

\begin{definition}
Let $\Xi_\varepsilon$ be a maximal $\varepsilon$-net. For all $\xi_{i,\varepsilon} \in \Xi_\varepsilon$, the associated family of Voronoi cells is defined by
$${\cal V}(\xi_{i,\varepsilon},\varepsilon)=\{x \in {\cal S}^2: \forall j \ne i, \; d(x, \xi_{i,\varepsilon}) \le d(x,\xi_{j,\varepsilon})\}.$$
\end{definition}
\noindent We recall \cite{BKMP} that ${\cal B}(\xi_{i,\varepsilon},\varepsilon/2) \subseteq {\cal V}(\xi_{i,\varepsilon},\varepsilon) \subseteq {\cal B}(\xi_{i,\varepsilon},\varepsilon)$, hence $\text{Vol}({\cal V}(\xi_{i,\varepsilon},\varepsilon)) \approx \varepsilon^2$.
Let
$${\cal N}^c(f_\ell; {\cal V}(\xi_{i,\varepsilon},\varepsilon),I)=\#\{x \in {\cal V}(\xi_{i,\varepsilon},\varepsilon): f_\ell(x) \in I, \nabla f_\ell(x)=0\}.$$
Note that, almost surely, the summation of the critical points over the
Voronoi cells equals the total number of critical points:
\begin{align*}
{\cal N}^c_I(f_\ell)= \sum_{\xi_{i,\varepsilon} \in \Xi_\varepsilon } {\cal N}^c(f_\ell; {\cal V}(\xi_{i,\varepsilon},\varepsilon),I).
\end{align*}
Therefore, we have that
\begin{equation} \label{mir1}
\text{Var} \left( \mathcal{N}^{c}_I(f_{\ell })\right) =\sum_{\xi_{i,\varepsilon}, \xi_{j,\varepsilon} \in \Xi_\varepsilon } \text{Cov} \left( \mathcal{N}^{c}(f_{\ell }; {\cal V}(\xi _{i,\varepsilon }),I),\mathcal{N}%
^{c}(f_{\ell }; {\cal V}(\xi _{j,\varepsilon }),I)\right).
\end{equation}

\subsubsection{Proof of Proposition \ref{P-afkao}}

We divide the sum in \eqref{mir1} into terms with corresponding points at distance bigger
or smaller than $C/\ell$. For the former we will
exploit the precise Kac-Rice formula below, in contrast to the latter regime whose contribution is bounded; first we define
\begin{equation}
\varepsilon =c/\ell,  \label{numberedformula}
\end{equation}
where $c$ is a positive constant sufficiently small so that, we may apply Proposition
\ref{a} stating that for every $i$,
\begin{align}
\text{Var}\left( \mathcal{N}^{c}(f_{\ell };{\cal V}(\xi _{\varepsilon ,i}),I)\right)
&=\iint_{{\cal V} (\xi _{\varepsilon ,i})\times {\cal V}(\xi _{\varepsilon ,i})} \iint_{I\times
I}K_{2,\ell}(x,y;t_{1},t_{2})dt_{1}dt_{2}dxdy \nonumber \\
&\;\;+\mathbb{E}\left[ \mathcal{N}^{c}(f_{\ell }; {\cal V}(\xi _{\varepsilon ,i}),I)%
\right] -(\mathbb{E}\left[ \mathcal{N}^{c}(f_{\ell }; {\cal V}(\xi _{\varepsilon
,i}),I)\right])^{2};  \label{lario1}
\end{align}
and, by Proposition \ref{expectation copy(1)} and \eqref{numberedformula},
\begin{align}
\mathbb{E}\left[ \mathcal{N}^{c}(f_{\ell };V(\xi _{\varepsilon ,i}),I)\right]
\leq \mathbb{E}\left[ \mathcal{N}^{c}(f_{\ell };B(\xi _{\varepsilon
,i};\varepsilon ),I)\right]  \leq  \pi \varepsilon ^{2}\ell ^{2}=O(1).  \label{lario3}
\end{align}
Note that in the proof of Proposition \ref{a} we exploit the non-degeneracy of the covariance matrix for sufficiently close points $x,y$ established in Appendix \ref{AppC}, whence we can apply Kac-Rice as in formula \eqref{lario1}. Moreover, by Lemma \ref{b}, we have
\begin{equation}
\iint_{{\cal V}(\xi _{\varepsilon ,i})\times {\cal V} (\xi _{\varepsilon ,i})}\iint_{I\times
I}K_{2}(x,y;t_{1},t_{2})dt_{1}dt_{2}dxdy\leq \ell ^{4}\cdot (\pi \varepsilon
^{2})^{2}=O(1),  \label{nonnumberedyet}
\end{equation}
again by \eqref{numberedformula}. Substituting the estimates \eqref{lario3}
and \eqref{nonnumberedyet} into \eqref{lario1} yields
\begin{equation}
\label{seriousnumber}
\text{Var}\left( \mathcal{N}^{c}(f_{\ell }; {\cal V}(\xi _{\varepsilon ,i}),I)\right)
=O(1).
\end{equation}
uniformly for all $\ell,i$.
Using the latter with the Cauchy-Schwartz inequality we
may bound each individual summand in the summation \eqref{mir1} as
\begin{equation}
\label{csbound}
\left\vert \text{Cov} \left( \mathcal{N}^{c}(f_{\ell }; {\cal V} (\xi _{i,\varepsilon }),I),%
\mathcal{N}^{c}(f_{\ell }; {\cal V}(\xi _{j,\varepsilon }),I)\right) \right\vert
\leq \sqrt{\text{Var} \left( \mathcal{N}^{c}(f_{\ell }; {\cal V}(\xi _{i,\varepsilon
}),I)\right) }\cdot\sqrt{\text{Var}\left( \mathcal{N}^{c}(f_{\ell }; {\cal V}(\xi
_{j,\varepsilon }),I)\right) } = O(1).
\end{equation}


As there are $O(\ell ^{2})$ pairs of Voronoi cells
at distance smaller or equal than $C/\ell$, \eqref{csbound} implies that
the contribution of this range to \eqref{mir1} is
\begin{equation*}
\sum_{d( {\cal V}(\xi _{i,\varepsilon }),{\cal V}(\xi _{j,\varepsilon }))<C/\ell
}\left\vert \text{Cov}\left( \mathcal{N}^{c}(f_{\ell }; {\cal V}(\xi _{i,\varepsilon }),I),%
\mathcal{N}^{c}(f_{\ell }; {\cal V}(\xi _{j,\varepsilon }),I)\right) \right\vert
=O(\ell ^{2}).
\end{equation*}
For Voronoi cells that are at distance greater than $C/\ell$, we may use
the standard Kac-Rice formula in Lemma \ref{e}.
Applying Kac-Rice individually on each of the pairs $({\cal V}(\xi _{i,\varepsilon
}),{\cal V}(\xi _{j,\varepsilon }))$ such that $d( {\cal V}(\xi _{i,\varepsilon }), {\cal V}(\xi
_{j,\varepsilon }))\geq C/\ell$, we have
\begin{align}
\sum_{d({\cal V}(\xi _{i,\varepsilon }), {\cal V}(\xi _{j,\varepsilon }))\geq C/\ell
} \text{Cov} \left( \mathcal{N}^{c}(f_{\ell }; {\cal V}(\xi _{i,\varepsilon }),I),\mathcal{N}%
^{c}(f_{\ell }; {\cal V}(\xi _{j,\varepsilon }),I)\right)
=\int_{\mathcal{W}}\iint_{I\times I} K_{2,\ell}(x,y;t_{1},t_{2}) dt_{1}dt_{2}dxdy
\label{k2-p1}
\end{align}
where
\begin{equation*}
\mathcal{W}=\bigcup _{d( {\cal V}(\xi _{i,\varepsilon }), {\cal V}(\xi _{j,\varepsilon }))\geq
C/\ell } {\cal V}(\xi _{i,\varepsilon })\times {\cal V} (\xi _{j,\varepsilon }),
\end{equation*}%
is the union of all tuples of points belonging to Voronoi cells further apart than $C/\ell$.
Now with the help of Proposition \ref{key-c}, we may write this summation \eqref{k2-p1} as
\begin{align*}
&\sum_{d( {\cal V}(\xi _{i,\varepsilon }), {\cal V}(\xi _{j,\varepsilon }))\geq C/\ell
} \text{Cov} \left( \mathcal{N}^{c}(f_{\ell }; {\cal V}(\xi _{i,\varepsilon }),I),\mathcal{N}%
^{c}(f_{\ell }; {\cal V}(\xi _{j,\varepsilon }),I)\right) \nonumber \\
&\;\;= \int_{C/\ell}^{\pi/2} \iint_{I\times
I}L_{2,\ell}(\phi;t_{1},t_{2})  dt_{1}dt_{2} \sin \phi d \phi+\int_{\cal W} \iint_{I\times I}E_{2,\ell}(x,y;t_{1},t_{2})dt_{1}dt_{2}dxdy  \nonumber \\
&\;\;=\int_{C/\ell}^{\pi/2} \iint_{I\times
I}L_{2,\ell}(\phi;t_{1},t_{2}) dt_{1}dt_{2}  \sin \phi d \phi+O(\ell ^{5/2}),
\end{align*}
as claimed.

\section{Asymptotics of the two-point correlation function}
\label{asymptoticsection}

Here we prove the auxiliary results in section \ref{afkao}.

\subsection{Long-range asymptotics for the two-point correlation function} \label{llllllllong}

In this section we prove Lemma \ref{e}, Lemma \ref{d}, and Proposition \ref{key-c}.

\subsubsection{Conditional covariance matrix}

For $x,y\in {\cal S}^2$ we define the following random vector
\begin{equation*}
Z_{\ell ;x,y}=(\nabla f_{\ell }(x),\nabla f_{\ell }(y),\nabla ^{2}f_{\ell
}(x),\nabla ^{2}f_{\ell }(y)).
\end{equation*}
To write the Kac-Rice formula in coordinate system, given $x$, $y \in {\cal S}^2$, we consider two local orthogonal frames
$\{e_{1}^{x}, e_{2}^{x}\}$ and
$\{e_{1}^{y},e_{2}^{y}\}$ defined in some neighbourhood of $x$ and $y$ respectively.
This gives rise to the (local) identifications
\begin{equation} \label{ident_isom}
T_{x}({\cal S}^2)\cong \mathbb{R}^{2}\cong T_{y}({\cal S}^2),
\end{equation}
so that, as discussed earlier we do not have to work with
probability densities defined on tangent planes which depend on
the points $x$ and $y$ respectively. Under the identification
\eqref{ident_isom} the random vector $Z_{\ell ;x,y}$ is a
$\mathbb{R}^{10}$ centered Gaussian random vector.  By isotropy,
it is convenient to perform our computations along a specific
geodesic. In particular, we focus on the equatorial line  $x=(\pi
/2,\phi )$, $y=(\pi /2,0)$, and we work with the orthogonal frames
\begin{equation} \label{orthogonalf}
\left\{e_{1}^{x}=\frac{\partial }{\partial \theta _{x}},\;e_{2}^{x}=\frac{
\partial }{\partial \varphi _{x}}\right\},\hspace{0.5cm}\left\{e_{1}^{y}=\frac{
\partial }{\partial \theta _{y}},\;e_{2}^{y}=\frac{\partial }{\partial
\varphi _{y}}\right\}.
\end{equation}
In Appendix  \ref{cov_matx_ev} we compute the entries of the $10\times 10$ covariance matrix of  $Z_{\ell ;x,y}$, i.e.,
$$\Sigma _{\ell }(\phi)=\left(
\begin{array}{cc}
A_{\ell}(\phi) & B_{\ell}(\phi) \\
B^{t}_{\ell}(\phi) & C_{\ell}(\phi)
\end{array}
\right),$$
where $A_{\ell}(\phi)$, $B_{\ell}(\phi)$ and $C_{\ell}(\phi)$ are the covariance matrices of the gradient terms, first and second order derivatives, and second order derivatives, respectively. In
Appendix \ref{cov_matx} we compute the
 conditional covariance matrix $\Omega _{\ell }(\phi )$ of
the random vector
\begin{equation*}
(\nabla ^{2}f_{\ell }(x),\nabla ^{2}f_{\ell }(y)\big|\nabla f_{\ell
}(x)=\nabla f_{\ell }(y)={\bf 0}),
\end{equation*}%
i.e.,
\begin{equation} \label{2agosto}
\Omega _{\ell }(\phi )=C_{\ell }(\phi )-B_{\ell }(\phi
)^{t}A_{\ell }(\phi )^{-1}B_{\ell }(\phi ).
\end{equation}%
After scaling, we obtain
\begin{equation} \label{1agosto}
\Delta _{\ell }(\phi )=\frac{8}{   \lambda _{\ell }^2   }\Omega _{\ell }(\phi )=\left(
\begin{array}{cc}
\Delta _{1,\ell }(\phi ) & \Delta _{2,\ell }(\phi ) \\
\Delta _{2,\ell }(\phi ) & \Delta _{1,\ell }(\phi )%
\end{array}%
\right),
\end{equation}%
where the entries  $a_{i,\ell}(\phi)$, $i=1,\dots,8$, of  $\Delta _{1,\ell }(\phi )$ and $\Delta _{2,\ell }(\phi )$ are defined by
\begin{align}
\Delta _{1,\ell }(\phi )&=\left(
\begin{array}{ccc}
3-\frac{16\beta _{2,\ell }^{2} (\phi)}{\lambda _{\ell } (\lambda _{\ell
}^{2}-4\alpha _{2,\ell }^{2} (\phi))}-\frac{2}{\lambda _{\ell }} & 0 & 1-\frac{%
16\beta _{2,\ell } (\phi) \beta _{3,\ell }(\phi)}{\lambda _{\ell }(\lambda _{\ell
}^{2}-4\alpha _{2,\ell }^{2}(\phi))}+\frac{2}{\lambda _{\ell }} \\
0 & 1-\frac{16\beta _{1,\ell }^{2}(\phi)}{\lambda _{\ell } (\lambda _{\ell
}^{2}-4\alpha _{1,\ell }^{2}(\phi))} & 0 \\
1-\frac{16\beta _{2,\ell }(\phi)\beta _{3,\ell }(\phi)}{\lambda _{\ell } (\lambda
_{\ell }^{2}-4\alpha _{2,\ell }^{2}(\phi))}+\frac{2}{\lambda _{\ell }} & 0 & 3-%
\frac{16\beta _{3,\ell }^{2}(\phi)}{\lambda _{\ell } (\lambda _{\ell
}^{2}-4\alpha _{2,\ell }^{2}(\phi))}-\frac{2}{\lambda _{\ell }}%
\end{array}%
\right) \nonumber\\
&=\left(
\begin{array}{ccc}
3+a_{1,\ell }(\phi ) & 0 & 1+a_{4,\ell }(\phi ) \\
0 & 1+a_{2,\ell }(\phi ) & 0 \\
1+a_{4,\ell }(\phi ) & 0 & 3+a_{3,\ell }(\phi )%
\end{array}
\right), \label{Delta_1}
\end{align}%
and
\begin{align}
\Delta _{2,\ell }(\phi )&=\left(
\begin{array}{ccc}
8\frac{\gamma _{1,\ell }(\phi)+\frac{4\alpha _{2,\ell }(\phi)\beta _{2,\ell }^{2}(\phi)}{
4\alpha _{2,\ell }^{2}(\phi)-\lambda _{\ell }^{2}}}{\lambda _{\ell }^2} & 0 & 8\frac{\gamma _{3,\ell }(\phi)+\frac{4\alpha _{2,\ell }(\phi)\beta
_{2,\ell }(\phi)\beta _{3,\ell }(\phi)}{4\alpha _{2,\ell }^{2}(\phi)-\lambda _{\ell }^{2}}}{\lambda _{\ell }^2} \\
0 & 8\frac{\gamma _{2,\ell }(\phi)+\frac{4\alpha _{1,\ell }(\phi)\beta _{1,\ell }^{2}(\phi)}{%
4\alpha _{1,\ell }^{2}(\phi)-\lambda _{\ell }^{2}}}{\lambda _{\ell }^2} & 0 \\
8\frac{\gamma _{3,\ell }(\phi)+\frac{4\alpha _{2,\ell }(\phi)\beta _{2,\ell }(\phi)\beta
_{3,\ell }(\phi)}{4\alpha _{2,\ell }^{2}(\phi)-\lambda _{\ell }^{2}}}{\lambda _{\ell }^2} & 0 & 8\frac{\gamma _{4,\ell }(\phi)+\frac{4\alpha _{2,\ell
}(\phi) \beta _{3,\ell }^{2}(\phi)}{4\alpha _{2,\ell }^{2}(\phi)-\lambda _{\ell }^{2}}}{\lambda _{\ell }^2}%
\end{array}%
\right) \nonumber \\
&=\left(
\begin{array}{ccc}
a_{5,\ell }(\phi ) & 0 & a_{8,\ell }(\phi ) \\
0 & a_{6,\ell }(\phi ) & 0 \\
a_{8,\ell }(\phi ) & 0 & a_{7,\ell }(\phi )%
\end{array}%
\right), \label{Delta_2}
\end{align}
with
\begin{equation*}
\alpha _{1,\ell }(\phi )=P_{\ell }^{\prime }(\cos \phi ),
\end{equation*}
\begin{equation*}
\alpha _{2,\ell }(\phi )=-\sin ^{2}\phi P_{\ell }^{\prime \prime }(\cos \phi
)+\cos \phi P_{\ell }^{\prime }(\cos \phi ),
\end{equation*}

\begin{equation*}
\beta _{1,\ell }(\phi )=\sin \phi P_{\ell }^{\prime \prime }(\cos \phi ),
\end{equation*}%
\begin{equation*}
\beta _{2,\ell }(\phi )=\sin \phi \cos \phi P_{\ell }^{\prime \prime }(\cos
\phi )+\sin \phi P_{\ell }^{\prime }(\cos \phi ),
\end{equation*}%
\begin{equation*}
\beta _{3,\ell }(\phi )=-\sin ^{3}\phi P_{\ell }^{\prime \prime \prime
}(\cos \phi )+3\sin \phi \cos \phi P_{\ell }^{\prime \prime }(\cos \phi
)+\sin \phi P_{\ell }^{\prime }(\cos \phi ),
\end{equation*}
and
\begin{equation*}
\gamma _{1,\ell }(\phi )=(2+\cos ^{2}\phi )P_{\ell }^{\prime \prime }(\cos
\phi )+\cos \phi P_{\ell }^{\prime }(\cos \phi ),
\end{equation*}%
\begin{equation*}
\gamma _{2,\ell }(\phi )=-\sin ^{2}\phi P_{\ell }^{\prime \prime \prime
}(\cos \phi )+\cos \phi P_{\ell }^{\prime \prime }(\cos \phi ),
\end{equation*}%
\begin{equation*}
\gamma _{3,\ell }(\phi )=-\sin ^{2}\phi \cos \phi P_{\ell }^{\prime \prime
\prime }(\cos \phi )+(-2\sin ^{2}\phi +\cos ^{2}\phi )P_{\ell }^{\prime
\prime }(\cos \phi )+\cos \phi P_{\ell }^{\prime }(\cos \phi ),
\end{equation*}
\begin{equation*}
\gamma _{4,\ell }(\phi )=\sin ^{4}\phi P_{\ell }^{\prime \prime \prime
\prime }(\cos \phi )-6\sin ^{2}\phi \cos \phi P_{\ell }^{\prime \prime
\prime }(\cos \phi )+(-4\sin ^{2}\phi +3\cos ^{2}\phi )P_{\ell }^{\prime \prime }(\cos \phi
)+\cos \phi P_{\ell }^{\prime }(\cos \phi ).
\end{equation*}

\begin{proof}[Proof of Lemma \ref{e}]
Lemma \ref{e} follows from Theorem 11.5.1 in \cite{adlertaylor} provided that we show
that for $C>0$ sufficiently big the covariance matrix $\Sigma _{\ell }(\phi)$ is
nonsingular. To this end, we may write
\begin{equation*}
\det (\Sigma _{\ell }(\phi) )=\det (A_\ell(\phi) )\times \det (\Omega_\ell(\phi) ),
\end{equation*}
that is that each of the two matrices on the
right-hand side is nonsingular. The latter follows from the fact that the
perturbation terms $a_{i,\ell}$ defined in \eqref{Delta_1} and \eqref{Delta_2}, properly normalized, are decaying,
see Appendix \ref{estimates}, Lemma \ref{alphaalpha}.
\end{proof}

\subsubsection{Proof of Proposition \ref{key-c}}
First we recall that the two-point correlation
function is given by \eqref{k2elegant}, a Gaussian expectation related to a vector with covariance matrix $\Sigma_{\ell}(\phi)$.
To understand the asymptotic behaviour of  the function $K_{2,\ell}$ we will have to provide a more explicit formula by using the orthogonal frames \eqref{orthogonalf} chosen above. Then we will have a frame-dependent formula for the two-point correlation function depending on the geodesic distance $\phi=d(x,y)$. To define $K_{2,\ell }(\phi ; t_{1},t_{2})$ pointwise
(almost everywhere), we write, for any two intervals $I_{1},I_{2}$:
\begin{align*}
& \hspace{4cm}\iint_{I_{1}\times I_{2}}K_{2,\ell }(\phi ; t_{1},t_{2})dt_{1}dt_{2}\\
&=\frac{1}{(2\pi )^{2}\sqrt{\text{det} (A_{\ell }(\phi ))}}\iint_{{\mathbb{R}^{3}
}\times {\mathbb{R}^{3}}}|\zeta _{x,1}\zeta _{x,3}-\zeta _{x,2}^{2}|\cdot|\zeta
_{y,1}\zeta _{y,3}-\zeta _{y,2}^{2}|\cdot 1\hspace{-0.27em}\mbox{\rm l}_{\big\{
\frac{\zeta _{x,1}+\zeta _{x,3}}{\lambda _{\ell }}\in I_{1}\big\}}\cdot 1\hspace{%
-0.27em}\mbox{\rm l}_{\big\{\frac{\zeta _{y,1}+\zeta _{y,3}}{\lambda _{\ell
}}\in I_{2}\big\}}\\
&\hspace{0.2cm}\times \frac{1}{(2\pi )^{3}}\exp \left\{-\frac{1}{2}(\zeta
_{x,1},\zeta _{x,2},\zeta _{x,3},\zeta _{y,1},\zeta _{y,2},\zeta
_{y,3})\Omega _{\ell }(\phi )^{-1}(\zeta _{x,1},\zeta _{x,2},\zeta
_{x,3},\zeta _{y,1},\zeta _{y,2},\zeta _{y,3})^{t}\right\}\\
&\hspace{0.2cm} \times \frac{1}{\sqrt{\text{det} (\Omega _{\ell }(\phi ))}}d\zeta
_{x,1}\;d\zeta _{x,2}\;d\zeta _{x,3}\;d\zeta _{y,1}\;d\zeta _{y,2}\;d\zeta
_{y,3}.
\end{align*}%
Here we exploited the linear dependence \eqref{linear_dp}. Now we scale the variables: for $i=1,2,3$ introduce $\tilde{\zeta}_{x,i}$ and $\tilde{\zeta}_{y,i}$:
\begin{equation*}
\zeta _{x,i}= \frac{\lambda_\ell}{\sqrt 8} \tilde{%
\zeta}_{x,i},\hspace{1cm}\zeta _{y,i}=\frac{\lambda_\ell}{\sqrt 8} \tilde{\zeta}_{y,i}.
\end{equation*}%
With the new variables we have
\begin{align*}
& \iint_{I_{1}\times I_{2}}K_{2,\ell }(\phi; t_{1},t_{2})dt_{1}dt_{2}
=\frac{1 }{(2\pi )^{2}\sqrt{\text{det} (A_{\ell }(\phi) )}}\frac{\lambda
_{\ell }^{4}}{8^{2}}\iint_{{\mathbb{R}^{3}}\times {%
\mathbb{R}^{3}}}|\tilde{\zeta}_{x,1}\tilde{\zeta}_{x,3}-\tilde{\zeta}%
_{x,2}^{2}|\cdot |\tilde{\zeta}_{y,1}\tilde{\zeta}_{y,3}-\tilde{\zeta}_{y,2}^{2}|
\\
& \hspace{0.2cm}\times 1\hspace{-0.27em}\mbox{\rm l}_{\Big\{ \frac{   \tilde{\zeta%
}_{x,1}+\tilde{\zeta}_{x,3}}{\sqrt{8} }\in I\Big\}}1\hspace{-0.27em}\mbox{\rm l}_{\Big\{%
\frac{
\tilde{\zeta}_{y,1}+\tilde{\zeta}_{y,3}}{\sqrt{8}}\in I\Big\}} \\
& \hspace{0.2cm}\times \frac{1}{(2\pi )^{3}}\exp \left\{-\frac{1}{2}(\tilde{\zeta}%
_{x,1},\tilde{\zeta}_{x,2},\tilde{\zeta}_{x,3},\tilde{\zeta}_{y,1},\tilde{%
\zeta}_{y,2},\tilde{\zeta}_{y,3})\Delta _{\ell }(\phi )^{-1}(\tilde{\zeta}%
_{x,1},\tilde{\zeta}_{x,2},\tilde{\zeta}_{x,3},\tilde{\zeta}_{y,1},\tilde{%
\zeta}_{y,2},\tilde{\zeta}_{y,3})^{t}\right\} \\
& \hspace{0.2cm}\times \frac{1}{\sqrt{\text{det}(\Delta _{\ell }(\phi ))}}d%
\tilde{\zeta}_{x,1}\;d\tilde{\zeta}_{x,2}\;d\tilde{\zeta}_{x,3}\;d\tilde{%
\zeta}_{y,1}\;d\tilde{\zeta}_{y,2}\;d\tilde{\zeta}_{y,3}.
\end{align*}%
Making the substitutions
\begin{equation*}
\left(
\begin{array}{c}
\tilde{\zeta}_{x,1} \\
\tilde{\zeta}_{x,2} \\
\tilde{\zeta}_{x,3}
\end{array}
\right) =\left(
\begin{array}{ccc}
1 & 0 & 0 \\
0 & 1 & 0 \\
-1 & 0 & \sqrt 8
\end{array}
\right) \left(
\begin{array}{c}
z_{1} \\
z_{2} \\
t_{1}
\end{array}
\right),
\end{equation*}
\begin{equation*}
\left(
\begin{array}{c}
\tilde{\zeta}_{y,1} \\
\tilde{\zeta}_{y,2} \\
\tilde{\zeta}_{y,3}
\end{array}
\right) =\left(
\begin{array}{ccc}
1 & 0 & 0 \\
0 & 1 & 0 \\
-1 & 0 & \sqrt 8
\end{array}
\right) \left(
\begin{array}{c}
w_{1} \\
w_{2} \\
t_{2}
\end{array}
\right),
\end{equation*}
we obtain
\begin{align} \label{kkernelkk}
& \hspace{5cm} K_{2,\ell }(\phi ; t_{1},t_{2}) \nonumber \\
&=\frac{1 }{(2\pi )^{2}\sqrt{\text{det}(A_{\ell
}(\phi) )}}\frac{\lambda _{\ell }^{4}}{8^{2}}  \iint_{{\mathbb{R}^{2}}\times {\mathbb{R}^{2}}}
\left|
z_{1} \left( \sqrt 8
t_{1}-z_{1}\right)-z_{2}^{2} \right|
\cdot \left|w_{1}\big(  \sqrt 8 t_{2}-w_{1}\big)-w_{2}^{2}
\right|\nonumber \\
&\hspace{0.2cm}\times \frac{1}{(2\pi )^{3}}\exp \left\{-\frac{1}{2}%
v_{t_{1},t_{2}}(z_{1},z_{2},w_{1},w_{2})\Delta _{\ell }(\phi
)^{-1}v_{t_{1},t_{2}}(z_{1},z_{2},w_{1},w_{2})^{t}\right\}  \frac{1}{\sqrt{\text{det}%
(\Delta _{\ell }(\phi ))}} 8 \;
dz_{1}dz_{2}dw_{1}dw_{2},
\end{align}
where
\begin{equation*}
v_{t_{1},t_{2}}(z_{1},z_{2},w_{1},w_{2})=\left(z_{1},z_{2},\sqrt 8
t_{1}-z_{1},w_{1},w_{2}, \sqrt 8 t_{2}-w_{1}\right).
\end{equation*}%
Now let us observe that for the determinant of $A_\ell(\phi)$ we have
\begin{align}
\label{eq:sqrt(detA)}
(2\pi )^{2}\sqrt{\det (A_{\ell }(\phi ))} &={(2\pi )^{2}\sqrt{\frac{1}{16}%
(\lambda _{\ell }^{2}-4\alpha _{2,\ell }^{2}(\phi ))(\lambda _{\ell
}^{2}-4\alpha _{1,\ell }^{2}(\phi ))}} =\pi^{2} \sqrt{(\lambda _{\ell }^{2}-4\alpha _{2,\ell }^{2}(\phi
))(\lambda _{\ell }^{2}-4\alpha _{1,\ell }^{2}(\phi ))} .
\end{align}%
At this point we consider the $2$-point correlation function \eqref{kkernelkk} as a function
of the perturbing elements $\{a_{i,\ell}(\phi)\}$, $i=1,\dots,8$ defined in \eqref{Delta_1} and \eqref{Delta_2}; to this end
it is convenient to collect the elements into a single vector:
\begin{equation*}
\mathbf{a}=\mathbf{a}_{\ell }(\phi)=(a_{1,\ell }(\phi ),a_{2,\ell }(\phi ),a_{3,\ell
}(\phi ),a_{4,\ell }(\phi ),a_{5,\ell }(\phi ),a_{6,\ell }(\phi ),a_{7,\ell}(\phi ),a_{8,\ell }(\phi )),
\end{equation*}
and write
\begin{equation*}
\Delta(\mathbf{a})=\left(
\begin{array}{cc}
\Delta _{1}(\mathbf{a}) & \Delta _{2}(\mathbf{a}) \\
\Delta _{2}(\mathbf{a}) & \Delta _{1}(\mathbf{a})%
\end{array}%
\right),
\end{equation*}%
where
\begin{equation*}
\Delta _{1}(\mathbf{a})=\left(
\begin{array}{ccc}
3+a_{1} & 0 & 1+a_{4} \\
0 & 1+a_{2} & 0 \\
1+a_{4} & 0 & 3+a_{3}%
\end{array}
\right) \;\; 
\text{and} \;\;\;
\Delta _{2}(\mathbf{a})=\left(
\begin{array}{ccc}
a_{5} & 0 & a_{8} \\
0 & a_{6} & 0 \\
a_{8} & 0 & a_{7}%
\end{array}%
\right).
\end{equation*}
With this slight abuse of notation it is evident that
\begin{equation}
\label{eq:Deltal(phi)=Delta(al(phi))}
\Delta _{\ell }(\phi)=\Delta (\mathbf{a}_{\ell}(\phi)).
\end{equation}
We then introduce the functions
\begin{equation*}
\hat{q} (\mathbf{a};t_{1},t_{2};z_{1},z_{2},w_{1},w_{2})
=\frac{1}{\sqrt{\det (\Delta (\mathbf{a}))}}\exp \left\{-\frac{1}{2}%
v_{t_{1},t_{2}}(z_{1},z_{2},w_{1},w_{2})\Delta
(\mathbf{a})^{-1}v_{t_{1},t_{2}}(z_{1},z_{2},w_{1},w_{2})^{t}\right\},
\end{equation*}
and
\begin{align*}
q (\mathbf{a};t_1,t_2)&=\frac{1}{(2\pi )^{3}}\iint_{\mathbb{R}^{2}\times \mathbb{R}%
^{2}} \left|z_{1} \sqrt 8 t_1 -z_{1}^{2}-z_{2}^{2}\right|\cdot \left| w_{1} \sqrt 8 t_2 -w_{1}^{2}-w_{2}^{2}\right| \\
&  \hspace{0.2cm} \times \hat{q}
(\mathbf{a}; t_{1},t_{2};z_{1},z_{2},w_{1},w_{2})dz_{1}dz_{2}dw_{1}dw_{2}.
\end{align*}
Bearing in mind \eqref{kkernelkk}, \eqref{eq:sqrt(detA)} and \eqref{eq:Deltal(phi)=Delta(al(phi))} we have
\begin{equation}
\label{eq:=K2=*q}
K_{2,\ell }(\phi ; t_{1},t_{2}) = \frac{\lambda_{\ell}^{4}}{8\pi^{2} \sqrt{(\lambda _{\ell }^{2}-4\alpha _{2,\ell }^{2}(\phi
))(\lambda _{\ell }^{2}-4\alpha _{1,\ell }^{2}(\phi ))}} q (\mathbf{a}_{\ell}(\phi);t_1,t_2) .
\end{equation}

\begin{remark} \label{q(0)}
We note that

\begin{align*}
\iint_{I \times I} q ({\bf 0};t_1,t_2)dt_1dt_2 &= \frac{1}{8}
\frac{1}{(2\pi )^{3}} \Bigg[ \int_{I} d t \int_{\mathbb{R}^{2}}
\Big|z_{1} \sqrt{8}t-z_{1}^{2}-z_{2}^{2}\Big|   \exp
\left\{-\frac{3}{2}t^{2}\right\}\exp \left\{-\frac{1}{2}
(z_{1}^{2}+z_{2}^{2}-\sqrt{8}t z_{1})\right\} dz_{1}dz_{2}  \Bigg]^2 \nonumber \\
&=\frac{1}{8} \left[ \int_I p_{1}^{c}(t) d t  \right]^2.
\end{align*}
\end{remark}

Our next step is to study the asymptotic behaviour of the
functions $q$, by means of a Taylor expansion around the origin
$\mathbf{a}=\mathbf{0}$.

\vspace{0.5cm}

\noindent{\it {\bf  Taylor expansion of the two-point correlation function}}\\

To understand the behaviour of the two-point correlation
function in the long-range regime, we have to investigate the high energy asymptotic
behaviour of the integrals
\begin{align}  \label{sitt}
16 \pi^2 \int_{C/\ell}^{\pi/2} K_{2,\ell }(\phi; t_{1},t_{2}) \sin \phi d \phi=2 \lambda _{\ell }^{2} \int_{C/\ell}^{\pi/2}  \frac{\sin \phi}{ \sqrt{(1-4\alpha _{2,\ell }^{2}(\phi)/\lambda _{\ell }^{2})(1-4\alpha _{1,\ell }^{2}(\phi )/\lambda _{\ell }^{2})} }  q (\mathbf{a} ; t_{1},t_{2}) d \phi,
\end{align}
recalling \eqref{eq:=K2=*q}.
In the range $\phi >C/\ell $ the covariance matrices we shall deal
with are perturbations of the values they would have under
independence between values at the points $x,y.$ We can hence
exploit perturbation theory (see \cite{kato}, Theorem 1.5) to
yield that the Gaussian expectations are analytic functions of the
covariance matrix elements. Hence $q(\cdot ;t_{1},t_{2})$ is a smooth function,
defined on some neighbourhood of the origin (its arguments are
uniformly small for $\phi >C/\ell $), and we can expand it into a
finite Taylor polynomial around the origin, as follows:

\begin{align}  \label{taylor_ex}
q (\mathbf{a};t_1,t_2)&=q (\mathbf{0};t_1,t_2)+\sum_{i=1}^8 a_{i}
\frac{\partial }{\partial a_{i}}q (\mathbf{0};t_1,t_2)
+\sum_{i\neq j}a_{i} a_{j} \frac{\partial^2}{
\partial a_{i} \partial a_{j} }q (\mathbf{0};t_1,t_2)  \nonumber \\
&\;\;+\frac{1}{2} \sum_{i=1}^{8} a_{i}^2 \frac{\partial^2}{
\partial a_{i}^2} q (\mathbf{0};t_1,t_2)+O(w(t_1,t_2)||\mathbf{a}||^3), \mbox{ as $||\mathbf{a}|| \rightarrow 0$,}
\end{align}
for some $w(t_1,t_2)\ge 0$.
Below we will evaluate all the derivatives involved in \eqref{taylor_ex},
and show in addition that
\begin{equation}
\label{eq:w in L1}
w(\cdot ,\cdot)\in L^1(\mathbb{R}^2),
\end{equation}
important for integrating \eqref{taylor_ex} w.r.t. $t$.
We will
see that the variance of critical points is dominated by three
terms of order $O(\ell ^{3})$ in \eqref{taylor_ex}. \vspace{0.5cm}

\noindent{\it {\bf Asymptotic behaviour of the integrals}}\\

We shall now introduce the following notation: for $i,j=1,2,\dots 8$,
\begin{align*}
A_{0,\ell }& =\int_{C/\ell }^{\frac{\pi }{2}}\frac{\sin \phi }{\sqrt{%
(1-4\alpha _{2,\ell }^{2}(\phi )/\lambda _{\ell }^{2})(1-4\alpha _{1,\ell
}^{2}(\phi )/\lambda _{\ell }^{2})}}d\phi , \\
A_{i,\ell }& =\int_{C/\ell }^{\frac{\pi }{2}}\frac{a_{i,\ell }(\phi )}{\sqrt{%
(1-4\alpha _{2,\ell }^{2}(\phi )/\lambda _{\ell }^{2})(1-4\alpha _{1,\ell
}^{2}(\phi )/\lambda _{\ell }^{2})}}\sin \phi \;d\phi , \\
A_{ij,\ell }& =\int_{C/\ell }^{\frac{\pi }{2}}\frac{a_{i,\ell }(\phi
)a_{j,\ell }(\phi )}{\sqrt{(1-4\alpha _{2,\ell }^{2}(\phi )/\lambda _{\ell
}^{2})(1-4\alpha _{1,\ell }^{2}(\phi )/\lambda _{\ell }^{2})}}\sin \phi
\;d\phi .
\end{align*}%
As a consequence, we may write
\begin{align*}
16 \pi^2 \int_{C/\ell}^{\pi/2} K_{2,\ell}(\phi;t_1,t_2) \sin \phi
d \phi &=2 \lambda _{\ell }^{2} \Big\{A_{0,\ell } \; q
(\mathbf{0};t_1,t_2)+\sum_{i=1}^{8} A_{i,\ell
}\;\Big[\frac{\partial }{\partial a_{i} }q (\mathbf{a};t_1,t_2)\Big]_{\mathbf{a}=\mathbf{0}}\\
&\;\;+\frac{1}{2}\sum_{i,j=1}^{8}A_{ij,\ell }\;\Big[\frac{\partial
^{2}}{\partial a_{i} \partial a_{j}}q (\mathbf{a}
;t_1,t_2)\Big]_{\mathbf{a}=\mathbf{0}}\\
&\;\;+\int_{C/\ell }^{\frac{\pi }{2}}\frac{\sin \phi
}{\sqrt{(1-4\alpha _{2,\ell }^{2}(\phi )/\lambda _{\ell
}^{2})(1-4\alpha _{1,\ell }^{2}(\phi )/\lambda _{\ell
}^{2})}}O(||\mathbf{a}_{\ell }(\phi)||^{3})\;d\phi \Big\}.
\end{align*}

\noindent We shall now study the high frequency asymptotic behaviour of the terms $%
A_{0,\ell }$, $A_{i,\ell }$, and $A_{ij,\ell }$, for $i,j=1,2,\dots 8$.
First we shall show that the
first term in the expansion cancels out with the squared expectation.
More precisely, we shall prove the following lemma:

\begin{lemma} \label{A0}
\label{6.2} \label{A_0} As $\ell \rightarrow \infty$, we have
\begin{equation*}
2 \lambda _{\ell }^{2} \; A_{0,\ell }\; \iint_{I \times I} q (
\mathbf{0}; t_{1},t_{2}) d t_1 d t_2- \big(\mathbb{E}[{\cal N}^c_I(f_{\ell})]\big)^2=\frac{\ell
^{3}}{4} \Big[\int_I p_{1}^{c}(t ) d t \Big]^2+O(\ell^{2} \log \ell ).
\end{equation*}
\end{lemma}

\noindent We shall then show that all linear terms $A_{i,\ell }$
with $i \ne 3$ are indeed subdominant. The bound $O(\ell ^{-3/2})$
may not be optimal; it is probably possible to improve it to
$O(\log \ell /\ell ^{2})$ by working somewhat harder; we postpone
this analysis to future research.

\begin{lemma} \label{Ai}
\label{firstO} As $\ell \rightarrow \infty $, for all $i=1,\dots 8$, such
that $i\neq 3$, we have%
\begin{equation*}
A_{i,\ell }= O(\ell ^{-3/2}),
\end{equation*}%
whereas for $i=3$, we get
\begin{equation*}
A_{3,\ell }=-8\ell ^{-1}+O(\ell ^{-2}\log \ell ).
\end{equation*}
\end{lemma}

\noindent Finally, in the next lemma we study the asymptotic behaviour of the second
order terms $A_{ij,\ell }$, for $i,j=1,\dots ,8$; again they are all
subdominant, but for the term with index $(7,7)$:

\begin{lemma} \label{Aij}
\label{secondO} As $\ell \rightarrow \infty $, for $(i,j)\neq (7,7)$, we
have
\begin{equation*}
A_{ij,\ell }=O(\ell ^{-2}\log \ell ),
\end{equation*}%
and for $(i,j)=(7,7)$ we have
\begin{equation*}
A_{77,\ell }=32\ell ^{-1}+O(\ell ^{-2}\log \ell ).
\end{equation*}
\end{lemma}
\noindent The proofs of Lemma \ref{A0}, Lemma \ref{Ai} and Lemma \ref{Aij} are in Appendix \ref{A_terms}. We have proved that
\begin{align} \label{dominantttttt}
&16 \pi^2 \int_{C/\ell}^{\pi/2} \iint_{I \times I} K_{2,\ell}(\phi; t_1,t_2) d t_1 d t_2 \sin \phi d \phi -
 \big(\mathbb{E}[{\cal N}^c_I(f_{\ell})]\big)^2 \nonumber \\
&= \ell^3 \Big\{  \frac{1}{4}  \Big[\int_I p_{1}^{c}(t ) d t \Big]^2 -16 \iint_{I \times I} \Big[\frac{%
\partial }{\partial a_{3} }q (\mathbf{a}; t_{1},t_{2})\Big]_{\mathbf{a}=\mathbf{0}} d t_1 d t_2 \nonumber \\
&\;\;+32 \iint_{I \times I}   \Big[\frac{\partial ^{2}}{\partial
a_{7}^2}q (\mathbf{a}; t_{1},t_{2})\Big]_{\mathbf{a}=\mathbf{0}} d
t_1 d t_2 \Big\}  +O(\ell^{5/2}).
\end{align}
We may rewrite the latter result as
\begin{align*}
16 \pi^2 \int_{C/\ell}^{\pi/2} \iint_{I \times I} K_{2,\ell}(\phi; t_1,t_2) d t_1 d t_2   \sin \phi d \phi -  \big(\mathbb{E}[{\cal N}^c_I(f_{\ell})]\big)^2
&= \int_{C/\ell}^{\pi/2} \iint_{I \times I}  L_{2,\ell}(\phi;t_1,t_2) d t_1 d t_2  \sin \phi d \phi   +O(\ell^{5/2}),
\end{align*}
with $L_{2,\ell}$ defined by \eqref{bachata}. In fact, once we
have isolated the dominant terms in \eqref{dominantttttt} (see
Appendix \ref{A_terms}), we can write them, as function of $\phi$,
in the form given in \eqref{bachata}. In Appendix \ref{A_terms},
the remainder $E_{2,\ell}(\phi; t_1,t_2)$ is computed explicitly,
and the bound \eqref{integralerror} is also established. To prove
the statement of Proposition \ref{key-c} we now compute the values
of the derivatives of $\hat{q}$ to obtain $v_2$ and $v_3$ in
\eqref{bachata}.

\vspace{0.5cm}

\noindent{\it {\bf Derivatives of $\hat{q}$}}\\

The relevant derivatives can be evaluated explicitly as
follows. Let
\begin{equation*}
\mathbf{a}_{i}=(0,\dots ,0,a_{i},0,\dots ,0),
\end{equation*}%
for $i=1,\dots ,8$; recall that $\hat{q} (\mathbf{a}; t_{1},t_{2};
z_1,z_2,w_1,w_2) $ is an analytic function of the elements of the
vector $ \mathbf{a}$, see \cite{kato}, Theorem 1.5, so that we can
write
\begin{equation*}
\Big[\frac{\partial^{j}}{\partial a_{i}^{j}} \hat{q}(\mathbf{a}; t_{1},t_{2} ; z_1,z_2,w_1,w_2 )\Big]_{\mathbf{a}=\mathbf{0}}=%
\Big[\frac{\partial ^{j}}{\partial a_{i}^{j}} \hat{q}_{\ell}(\mathbf{a}%
_{i}; t_{1},t_{2} ; z_1,z_2,w_1,w_2 )\Big]_{\mathbf{a}_i=\mathbf{0}},
\end{equation*}%
for $j=1,2$. Using Leibnitz integral rule and some tedious but
mechanical computations, we obtain
\begin{align}
\left[\frac{\partial }{\partial a_{3}} q
(\mathbf{a}_{3}; t_{1},t_{2})\right]_{\mathbf{a}_3=\mathbf{0}}
&=\frac{1}{2 \cdot 8^2}   \frac{1}{(2 \pi)^3} \iint_{\mathbb{R}^2
\times \mathbb{R}^2}
\left|z_1 \sqrt 8 t_1-z_1^2-z_2^2\right| \cdot \left|w_1 \sqrt 8 t_2-w_1^2-w_2^2\right| \nonumber \\
& \;\; \times \exp \left\{-\frac{3}{2}t_{1}^{2}\right\}\exp \left\{-\frac{1}{2}(z_{1}^{2}+z_{2}^{2}-\sqrt{8}%
t_{1}z_{1})\right\} \nonumber \\
&\;\; \times \exp \left\{-\frac{3}{2} t_{2}^{2}\right\}\exp \left\{-\frac{1}{2}(w_{1}^{2}+w_{2}^{2}-
\sqrt{8}t_{2}w_{1})\right\}  \nonumber  \\
& \;\;\times \left[ -6+(3t_{1}-\sqrt{2}z_{1})^{2}+(3t_{2}-\sqrt{2}%
w_{1})^{2}\right] d z_1 d z_2 d w_1 d w_2, \label{dominantttttt1}
\end{align}
\begin{align}
\Big[\frac{\partial ^{2}}{\partial a_{7}^{2}} q (\mathbf{a}%
_{7} ; t_{1},t_{2})\Big]_{\mathbf{a}_7=\mathbf{0}} &= \frac{1}{8^3}
\frac{1}{(2 \pi)^3} \iint_{\mathbb{R}^2 \times \mathbb{R}^2}
\left|z_1 \sqrt 8 t_1-z_1^2-z_2^2\right|\cdot \left|w_1 \sqrt 8 t_2-w_1^2-w_2^2\right| \nonumber \\
& \;\; \times \exp \left\{-\frac{3}{2}t_1^{2}\right\} \exp\left\{-\frac{1}{2}
(z_{1}^{2}+z_{2}^{2}-\sqrt{8}t_1 z_{1})\right\}  \nonumber \\
&\;\; \times  \exp \left\{-\frac{3}{2}t_2 ^{2}\right\} \exp\left\{-\frac{1}{2}
(z_{1}^{2}+z_{2}^{2}-\sqrt{8}t_2 z_{1})\right\}  \nonumber \\
& \;\; \times \left[3-(3t_1 -\sqrt{2}z_{1})^{2}\right] \cdot
\left[3-(3t_2 -\sqrt{2}z_{1})^{2}\right] d z_1 d z_2 d w_1 d w_2.
\label{dominantttttt2}
\end{align}

\noindent Performing similar computations reveals that on a sufficiently small neighborhood
of $\mathbf{a}=\mathbf{0}$ in $\R^{8}$ the function $w(t_1,t_2)$ appearing
in \eqref{taylor_ex} has Gaussian tails w.r.t. $(t_1,t_2)$, and hence
\eqref{eq:w in L1}, i.e. $w(\cdot, \cdot)$ belongs to $L^1(\mathbb{R}^2)$.
Indeed, the inverse matrix
$\Delta(\bf{a})^{-1}$ appearing in the definition of $\hat{q}$ is
a perturbation of the identity, whence a Gaussian term in
$(t_1,t_2)$ factors out from its derivatives of every order. This
concludes the proof of Proposition \ref{key-c}.

\subsubsection{Proof of Lemma \ref{d}}

To prove Lemma \ref{d}, it is sufficient to show that for $\phi >C/\ell $ we have%
\begin{equation*}
\sin ^{4}\phi \; [P_{\ell }^{\prime \prime }(\cos \phi )]^{2}, \hspace{0.2cm}
\frac{1}{\ell ^{2}} \sin ^{6}\phi [ P_{\ell }^{\prime \prime \prime }(\cos \phi
) ]^{2},\hspace{0,2cm} \frac{1}{\ell ^{4}} \sin ^{8}\phi [ P_{\ell }^{\prime
\prime \prime \prime }(\cos \phi )]^{2} \hspace{0.2cm} =O(\ell ^{4}),
\end{equation*}%
uniformly in $\ell$. To establish these bounds, we note that, from Lemma \ref{hilbs},
\begin{align*}
\sin ^{4}\phi [P_{\ell }^{\prime \prime }(\cos \phi )]^{2}
&=\sin ^{4}\phi \frac{\ell ^{3}}{\sin ^{5}\phi }+O(\ell ^{3})=O(\ell ^{4}),
\\
\frac{1}{\ell ^{2}} \sin ^{6}\phi [P_{\ell }^{\prime \prime \prime }(\cos \phi
)]^{2} &=\frac{\sin ^{6}\phi }{\ell ^{2}}\frac{\ell ^{5}%
}{\sin ^{7}\phi }+O(\ell ^{3})=O(\ell ^{4}), \\
\frac{1}{\ell ^{4}} \sin ^{8}\phi [P_{\ell }^{\prime \prime \prime \prime }(\cos
\phi )]^{2}&=\frac{\sin ^{8}\phi }{\ell ^{4}}\frac{%
\ell ^{7}}{\sin ^{9}\phi }+O(\ell ^{3})=O(\ell ^{4}),
\end{align*}%
as claimed.

\subsection{Short-range application of Kac-Rice}

In this section, we provide the proofs of Proposition \ref{a} and Lemma \ref{b}.

\subsubsection{Conditional covariance matrix}

With the scaling
\begin{equation*}
\phi =\psi / \ell,
\end{equation*}%
the matrix $\Sigma _{\ell }$ becomes
\begin{equation*}
\Sigma _{\ell }(\psi )=\left(
\begin{array}{cc}
A_\ell(\psi) & B_\ell(\psi) \\
B^{t}_\ell(\psi) & C_\ell(\psi)
\end{array}%
\right),
\end{equation*}%
where
\begin{equation*}
A_{\ell }(\psi )_{4\times 4}=\left(
\begin{array}{cccc}
\frac{\ell +1}{2\ell } & 0 & \alpha _{1,\ell }(\psi ) & 0 \\
0 & \frac{\ell +1}{2\ell } & 0 & \alpha _{2,\ell }(\psi ) \\
\alpha _{1,\ell }(\psi ) & 0 & \frac{\ell +1}{2\ell } & 0 \\
0 & \alpha _{2,\ell }(\psi ) & 0 & \frac{\ell +1}{2\ell }%
\end{array}%
\right),
\end{equation*}%
and the elements of the off-diagonal $2\times 2$ terms of $A$ are given by
\begin{equation*}
\alpha _{1,\ell }(\psi )=\frac{1}{\ell ^{2}}P_{\ell }^{\prime }(\cos (\psi
/\ell )),
\end{equation*}%
\begin{equation*}
\alpha _{2,\ell }(\psi )=-\frac{1}{\ell ^{2}}\sin ^{2}(\psi /\ell )P_{\ell
}^{\prime \prime }(\cos (\psi /\ell ))+\frac{1}{\ell ^{2}}\cos (\psi /\ell
)P_{\ell }^{\prime }(\cos (\psi /\ell )).
\end{equation*}%
The matrix $B_{\ell}(\psi)$ is given by
\begin{equation*}
B_{\ell }(\psi )_{4\times 6}=\left(
\begin{array}{cccccc}
0 & 0 & 0 & 0 & \beta _{1,\ell }(\psi ) & 0 \\
0 & 0 & 0 & \beta _{2,\ell }(\psi ) & 0 & \beta _{3,\ell }(\psi ) \\
0 & -\beta _{1,\ell }(\psi ) & 0 & 0 & 0 & 0 \\
-\beta _{2,\ell }(\psi ) & 0 & -\beta _{3,\ell }(\psi ) & 0 & 0 & 0%
\end{array}%
\right) ,
\end{equation*}%
with elements%
\begin{equation*}
\beta _{1,\ell }(\psi )=\frac{1}{\ell ^{3}}\sin (\psi /\ell )P_{\ell
}^{\prime \prime }(\cos (\psi /\ell )),
\end{equation*}%
\begin{equation*}
\beta _{2,\ell }(\psi )=\sin (\psi /\ell )\cos (\psi /\ell )\frac{1}{\ell
^{3}}P_{\ell }^{\prime \prime }(\cos (\psi /\ell ))+\sin (\psi /\ell )\frac{1%
}{\ell ^{3}}P_{\ell }^{\prime }(\cos (\psi /\ell )),
\end{equation*}%
\begin{equation*}
\beta _{3,\ell }(\psi )=-\sin ^{3}(\psi /\ell )\frac{1}{\ell ^{3}}P_{\ell
}^{\prime \prime \prime }(\cos (\psi /\ell ))+3\sin (\psi /\ell )\cos (\psi
/\ell )\frac{1}{\ell ^{3}}P_{\ell }^{\prime \prime }(\cos (\psi /\ell
))+\sin (\psi /\ell )\frac{1}{\ell ^{3}}P_{\ell }^{\prime }(\cos (\psi /\ell
)).
\end{equation*}
Finally, for the matrix $C_{\ell}(\psi)$, we have
\begin{equation*}
C_{\ell }(\psi )_{6\times 6}=\left(
\begin{array}{cc}
c_\ell(0) & c_\ell(\psi ) \\
c_\ell(\psi ) & c_\ell(0)
\end{array}
\right),
\end{equation*}
where
\begin{equation*}
c_\ell(0)=\left(
\begin{array}{ccc}
\frac{\ell (1+3\ell (\ell +2))-2}{8\ell ^{3}} & 0 & \frac{(\ell +1)(\ell
^{2}+\ell +2)}{8\ell ^{3}} \\
0 & \frac{(\ell +1)(\ell ^{2}+\ell -2)}{8\ell ^{3}} & 0 \\
\frac{(\ell +1)(\ell ^{2}+\ell +2)}{8\ell ^{3}} & 0 & \frac{\ell (1+3\ell
(\ell +2))-2}{8\ell ^{3}}%
\end{array}%
\right),
\end{equation*}%
and
\begin{equation*}
c_{\ell}(\psi )=\left(
\begin{array}{ccc}
\gamma _{1,\ell }(\psi ) & 0 & \gamma _{3,\ell }(\psi ) \\
0 & \gamma _{2,\ell }(\psi ) & 0 \\
\gamma _{3,\ell }(\psi ) & 0 & \gamma _{4,\ell }(\psi )%
\end{array}%
\right) ,
\end{equation*}%
with%
\begin{align*}
\gamma _{1,\ell }(\psi )=(2+\cos ^{2}(\psi /\ell ))\frac{1}{\ell ^{4}}%
P_{\ell }^{\prime \prime }(\cos (\psi /\ell ))+\cos (\psi /\ell )\frac{1}{%
\ell ^{4}}P_{\ell }^{\prime }(\cos (\psi /\ell )),
\end{align*}%
\begin{align*}
\gamma _{2,\ell }(\psi )=-\sin ^{2}(\psi /\ell )\frac{1}{\ell ^{4}}P_{\ell
}^{\prime \prime \prime }(\cos (\psi /\ell ))+\cos (\psi /\ell )\frac{1}{%
\ell ^{4}}P_{\ell }^{\prime \prime }(\cos (\psi /\ell )),
\end{align*}%
\begin{align*}
\gamma _{3,\ell }(\psi ) &=-\sin ^{2}(\psi /\ell )\cos (\psi /\ell )\frac{1%
}{\ell ^{4}}P_{\ell }^{\prime \prime \prime }(\cos (\psi /\ell )) +(-2\sin ^{2}(\psi /\ell )+\cos ^{2}(\psi /\ell ))\frac{1}{\ell ^{4}}%
P_{\ell }^{\prime \prime }(\cos (\psi /\ell ))\\
&\;\;+\cos (\psi /\ell )\frac{1}{%
\ell ^{4}}P_{\ell }^{\prime }(\cos (\psi /\ell )),
\end{align*}%
\begin{align*}
\gamma _{4,\ell }(\psi )& =\sin ^{4}(\psi /\ell )\frac{1}{\ell ^{4}}P_{\ell
}^{\prime \prime \prime \prime }(\cos (\psi /\ell ))-6\sin ^{2}(\psi /\ell
)\cos (\psi /\ell )\frac{1}{\ell ^{4}}P_{\ell }^{\prime \prime \prime }(\cos
(\psi /\ell )) \\
& +(-4\sin ^{2}(\psi /\ell )+3\cos ^{2}(\psi /\ell ))\frac{1}{\ell ^{4}}%
P_{\ell }^{\prime \prime }(\cos (\psi /\ell )) +\cos (\psi /\ell )\frac{1}{\ell ^{4}}P_{\ell }^{\prime }(\cos (\psi
/\ell )).
\end{align*}
The conditional covariance matrix $\Delta_{\ell}(\psi)$ is given by
\begin{equation*}
\Delta _{\ell }(\psi )=C_{\ell }(\psi )-B_{\ell }^{t}(\psi )A_{\ell
}^{-1}(\psi )B_{\ell }(\psi )=\left(
\begin{matrix}
\Delta _{1,\ell }(\psi ) & \Delta _{2,\ell }(\psi ) \\
\Delta _{2,\ell }(\psi ) & \Delta _{1,\ell }(\psi )%
\end{matrix}%
\right);
\end{equation*}%
we shall use below only the explicit expression only for  $\Delta _{1,\ell }(\psi )$ which is given by
\begin{align*}
&\hspace{3cm} \Delta _{1,\ell }(\psi )\\
&=\left(
\begin{array}{ccc}
\frac{2 \ell ( \ell+1)\beta _{2, \ell}(\psi )^{2}}{4 \ell^{2}\alpha _{2, \ell}(\psi )^{2}-(\ell+1)^{2}}+%
\frac{\ell(3\ell(\ell+2)+1)-2}{8\ell^{3}} & 0 & \frac{(\ell+1)\left( \frac{16\beta
_{2, \ell}(\psi )\beta _{3, \ell}(\psi ) \ell^{4}}{4 \ell^{2}\alpha _{2, \ell}(\psi )^{2}-(\ell+1)^{2}}%
+\ell^{2}+\ell+2\right) }{8 \ell^{3}} \\
0 & \frac{2 \ell(\ell+1)\beta _{1, \ell}(\psi )^{2}}{4\ell^{2}\alpha _{1, \ell}(\psi
)^{2}-(\ell+1)^{2}}+\frac{(\ell-1)(\ell+1)(\ell+2)}{8\ell ^{3}} & 0 \\
\frac{(\ell+1)\left( \frac{16\beta _{2, \ell}(\psi )\beta _{3, \ell}(\psi ) \ell^{4}}{%
4 \ell^{2}\alpha _{2, \ell}(\psi )^{2}-(\ell+1)^{2}}+\ell^{2}+\ell+2\right) }{8 \ell^{3}} & 0 &
\frac{2\ell (\ell+1)\beta _{3, \ell}(\psi )^{2}}{4 \ell^{2}\alpha _{2, \ell}(\psi )^{2}-(\ell+1)^{2}}+%
\frac{\ell(3\ell(\ell+2)+1)-2}{8 \ell^{3}}%
\end{array}%
\right).
\end{align*}

\begin{proof}[Proof of Proposition \ref{a}]
The statement of Proposition \ref{a} is an application of
Theorem 11.5.1 in \cite{adlertaylor}, provided that we
check that of the $10\times 10$ covariance matrix $\Sigma_\ell(\psi)$ of
the first and the second order derivatives of $f_{\ell}$ is nonsingular for
sufficiently $x,y$ satisfying $d(x,y)<c/\ell$ with $c>0$ sufficiently small.
The latter is shown in Appendix \ref{AppC}, with the
aid of specialized computer software, by Taylor expanding
the relevant determinant around the diagonal $x=y$.
\end{proof}

\subsubsection{Proof of Lemma \ref{b}}

We now need to study the high energy asymptotic behaviour of the kernel
\begin{equation*}
 \iint_{\mathbb{R}\times \mathbb{R}%
}K_{2,\ell}(\phi; t_{1},t_{2})dt_{1}dt_{2}
\end{equation*}
for $\phi<c/\ell$. In view of the equality \eqref{kkernelkk} we may proceed directly to bounding
\begin{equation}
\frac{\ell^4}{\sqrt{\text{det} (A_{\ell
}(\psi ))}}    \iint_{\mathbb{R}\times \mathbb{R}}\rho_{\ell
}(\psi ; t_{1},t_{2})dt_{1}dt_{2},  \label{lemmab}
\end{equation}%
where $\psi=\ell \phi$, and
\begin{align*}
\rho_{\ell }(\psi ; t_{1},t_{2})&= \iint_{\mathbb{R}^{2}\times \mathbb{R}^{2}}\Big|z_{1} t_1
-z_{1}^{2}-z_{2}^{2}\Big|\cdot
 \Big|w_{1} t_2 -w_{1}^{2}-w_{2}^{2}\Big| \frac{1}{\sqrt{\det(\Delta _{\ell }(\psi ))}} \\
 &\;\; \times  \exp \left\{-\frac{1}{2}%
(z_{1},z_{2},t_1-z_1,w_{1},w_{2}, t_2-w_1 ) \Delta _{\ell }(\psi
)^{-1} (z_{1},z_{2},t_1-z_1,w_{1},w_{2}, t_2-w_1)^{t}\right\}  dz_{1}dz_{2}dw_{1}dw_{2}.
\end{align*}%
Let $a_{i,\ell}(\psi)$ be the entries of the matrix $\Delta _{1,\ell }(\psi )$ defined by
\begin{align}
\Delta _{1,\ell }(\psi )
&=\left(
\begin{array}{ccc}
3+{a}_{1,\ell }(\psi ) & 0 & 1+{a}_{4,\ell }(\psi ) \\
0 & 1+{a}_{2,\ell }(\psi ) & 0 \\
1+{a}_{4,\ell }(\psi ) & 0 & 3+{a}_{3,\ell }(\psi )%
\end{array}%
\right). \label{deltapsi}
\end{align}
The main technical difficulty of this proof is
that, as Lemma \ref{18luglioa} below shows, the term $\sqrt{\text{det}(A_{\ell }(\psi ))}$ appearing in the
denominator is of the order $\psi^{2}$ around the origin. As a
consequence, a very delicate bound must be established on $\rho_{\ell}(
\psi; t_1,t_2)$ to ensure the convergence and the boundedness of the integral with respect to $\psi$.

\begin{lemma} \label{18luglioa}
Uniformly in $\ell$, we have for $\psi \leq c$
\begin{equation}
\det (A_{\ell }(\psi ))\geq c\psi ^{4},  \label{lemmab1}
\end{equation}
for some universal $c>0$.
\end{lemma}

\begin{lemma} \label{dettttA}
Uniformly in $\ell$, we have%
\begin{equation}
\iint_{\mathbb{R\times R}}\rho_{\ell }(\psi
; t_{1},t_{2})dt_{1}dt_{2}=O(\psi ^{2}).  \label{x}
\end{equation}
\end{lemma}

\begin{proof}[Proof of Proposition \ref{b} assuming Lemmas \ref{18luglioa} and \ref{dettttA}]
Since we have shown that the numerator is uniformly bounded by terms of order $\psi ^{2}$, the statement of Lemma \ref{b} follows at once upon substituting the estimates \eqref{lemmab1} and \eqref{x} into \eqref{lemmab}.
\end{proof}

\begin{proof}[Proof of Lemma \ref{18luglioa}]  First we note that
\begin{align*}
\text{det}(A_{\ell }(\psi ) )&=\frac{1}{16\ell ^{4}}((\ell +1)^{2}-4\ell
^{2}\alpha _{1,\ell}^{2}(\psi ))((\ell +1)^{2}-4\ell ^{2}\alpha _{2.\ell}^{2}(\psi ))
\\
&\geq  \frac{1}{16}(1-2\ell \alpha _{2,\ell }(\psi )/(\ell +1))(1-2\ell \alpha _{1,\ell
}(\psi )/(\ell +1)).
\end{align*}
By exploiting the Taylor expansions given in Appendix \ref{AppC},
for $\alpha_{1,\ell}(\psi)$ and  $\alpha_{2,\ell}(\psi)$, we
obtain
\begin{align*}
(1-2\ell \alpha_{1,\ell }(\psi )/(\ell +1))(1-2\ell \alpha _{2,\ell }(\psi )/(\ell +1))&=
\left( \frac{3}{64}+\frac{3}{32 \ell}-\frac{5}{64\ell^{2}}-\frac{1}{%
8 \ell^{3}}+\frac{1}{16 \ell^{4}}\right) \psi ^{4}+O(\psi ^{6}),
\end{align*}%
which certainly implies the statement of the present lemma.
\end{proof}

\begin{proof}[Proof of Lemma \ref{dettttA}]
We can bound \eqref{x} by
\begin{equation*}
\iint_{\mathbb{R\times R}}\rho_{\ell }(\psi
; t_{1},t_{2})dt_{1}dt_{2} \le \mathbb{E}\Big[%
|X_{1}X_{3}||Y_{1}Y_{3}|+|X_{1}X_{3}|Y_{2}^{2}+|Y_{1}Y_{3}|X_{2}^{2}+Y_{2}^{2}X_{2}^{2}%
\Big],
\end{equation*}%
where the random vector $(X_{1},X_{2},X_{3},Y_{1},Y_{2},Y_{3})$ is a
multivariate Gaussian with zero mean and covariance matrix $\Delta_{\ell}(\psi)$.
Applying repeatedly Cauchy-Schwarz and recalling that for Gaussian random
variables we have $\mathbb{E}[X_{1}^{4}]=3(\mathbb{E}[X_{1}^{2}])^{2}$, we
obtain
\begin{align*}
\mathbb{E}\left[|X_{1}X_{3}||Y_{1}Y_{3}|\right]& \leq \mathbb{E}
\left[|X_{1}||X_{3}||Y_{1}||Y_{3}|\right]\leq \left(\mathbb{E}\left[X_{1}^{2}X_{3}^{2}\right]
\cdot \mathbb{E}\left[Y_{1}^{2}Y_{3}^{2}\right]\right)^{1/2}\leq \left(\mathbb{E}[X_{1}^{4}]
\mathbb{E}\left[X_{3}^{4}\right]\mathbb{E}\left[Y_{1}^{4}\right]\mathbb{E}\left[Y_{3}^{4}\right]\right)^{1/4} \\
& =3\left(\mathbb{E}[X_{1}^{2}]\mathbb{E}\left[X_{3}^{2}\right]\mathbb{E}\left[Y_{1}^{2}\right]
\mathbb{E}\left[Y_{3}^{2}\right]\right)^{1/2}=3(3+a_{1,\ell }(\psi ))(3+a_{3,\ell }(\psi)),
\end{align*}
where in the last equality we write explicitly the variances by replacing
the elements of the covariance matrix $\Delta _{1,\ell }(\psi)$. Analogously we get
\begin{equation*}
\mathbb{E}\left[|X_{1}X_{3}|Y_{2}^{2}\right]\leq 3\left(\mathbb{E}\left[X_{1}^{2}\right]\mathbb{E}%
\left[X_{3}^{2}\right]\right)^{1/2}\mathbb{E}\left[Y_{2}^{2}\right]=3((3+a_{1,\ell }(\psi
))\left(3+a_{3,\ell }(\psi /\ell ))\right)^{1/2}(1+a_{2,\ell }(\psi )),
\end{equation*}
\begin{equation*}
\mathbb{E}\left[|Y_{1}Y_{3}|X_{2}^{2}\right]\leq 3\left(\mathbb{E}\left[Y_{1}^{2}\right]\mathbb{E}%
\left[Y_{3}^{2}\right]\right)^{1/2}\mathbb{E}\left[X_{2}^{2}\right]=3\left((3+a_{1,\ell }(\psi
))(3+a_{3,\ell }(\psi ))\right)^{1/2}(1+a_{2,\ell }(\psi )),
\end{equation*}
\begin{equation*}
\mathbb{E}\left[Y_{2}^{2}X_{2}^{2}\right]\leq 3\mathbb{E}\left[Y_{2}^{2}\right]\mathbb{E}%
\left[X_{2}^{2}\right]=3(1+a_{2,\ell }(\psi ))^{2}.
\end{equation*}
Collecting the previous results, and after some direct calculations, we obtain
\begin{align*}
&\mathbb{E}\left[
|X_{1}X_{3}|\cdot |Y_{1}Y_{3}|+|X_{1}X_{3}|\cdot Y_{2}^{2}+|Y_{1}Y_{3}|\cdot X_{2}^{2}+Y_{2}^{2}X_{2}^{2}%
\right] \\
& \leq 3(3+a_{1,\ell }(\psi ))(3+a_{3,\ell }(\psi ))+6((3+a_{1,\ell }(\psi
))(3+a_{3,\ell }(\psi /\ell )))^{1/2}(1+a_{2,\ell }(\psi ))+3(1+a_{2,\ell
}(\psi ))^{2}\\
&=\frac{3(\ell-2)(\ell-1)(\ell+1)^{3}(\ell+2)(\ell+3)\left( \ell^{2}+\ell-2\right) \psi ^{2}}{%
128 \ell^{7}(3 \ell(\ell+1)-2)}+O\left( \psi ^{3}\right).
\end{align*}%
\end{proof}

\begin{proof}[Proof of Proposition \ref{b}]
We have hence shown that the numerator is uniformly
bounded by terms of order $\psi ^{2}$. The statement of Lemma \ref{b} follows at once upon substituting the estimates \eqref{lemmab1} and \eqref{x} into \eqref{lemmab}.

\end{proof}

\begin{remark}
The geometric intuition of the previous result can be
explained as follows. We impose the condition that two critical points are
at (scaled) distance $\psi ,$ and study the asymptotic behaviour of the Hessian
for small values of $\psi .$ In this regime, the two critical points collide,
and hence the Hessian approaches zero with locally
quadratic behaviour. This is exactly the term we were looking for to cancel the
determinant term of order $\psi^{-2}$ in Lemma \ref{dettttA}.
\end{remark}

\section{Asymptotic expression for the variance} \label{cinque}

Here we find the analytic expression for the variance stated in Theorem \ref{th_variance copy(1)}.

\begin{proof}[Proof of Theorem \ref{th_variance copy(1)}]
In view of  \eqref{dominantttttt}, \eqref{dominantttttt1} and \eqref{dominantttttt2} we can write

\begin{align*}
&\text{Var}(\mathcal{N}_{I}^{c}(f_{\ell }))= \frac{\ell^{3}}{4} \left(  \left[ \int_I p^c_1(t) d t \right]^2-    \iint_{I \times I}
g^c_2(t_1,t_2) d t_1 dt_2+ 16  \left[ \int_{I} g^c_3(t) dt \right]^2 \right) +
O(\ell^{5/2}),
\end{align*}
where
\begin{align*}
g^c_2(t_1,t_2)&= \frac{1}{2} \frac{1}{(2 \pi)^{3}}\iint_{\mathbb{R}^2 \times
\mathbb{R}^2} \left| z_1 \sqrt 8 t_1  -z_1^2-z_2^2\right| \exp\left\{ - \frac{3 }{2} t_1^2 \right\}
\exp\left\{-\frac 1 2 (z_1^2+z_2^2- \sqrt 8 t_1 z _1)\right\} \\
&\;\; \times \left| w_1 \sqrt 8 t_2  -w_1^2-w_2^2\right| \exp\left\{ - \frac{3 }{2} t_2^2 \right\}
\exp\left\{-\frac 1 2 (w_1^2+w_2^2- \sqrt 8 t_2 w _1)\right\} \\
&\;\; \times \left[-6+(3 t_1-\sqrt{2} z_1)^2+(3 t_2 -\sqrt{2} w_1)^2\right] d z_1 d z_2
d w_1 d w_2,
\end{align*}
and
\begin{align*}
g^c_3(t)= \frac{1}{8} \frac{1}{ (2 \pi)^{3/2}}\int_{\mathbb{R}^2} \left|z_1 \sqrt 8 t -z_1^2-z_2^2\right|
\exp\left\{- \frac 3{2} t^2\right\} \exp\left\{-\frac 1 2 (
z_1^2+z_2^2- \sqrt 8 t z _1)\right\} \left[3-(3 t -\sqrt{2} z_1)^2\right] d z_1 d z_2.
\end{align*}
Let
\begin{equation*}
k(z_{1},z_{2},t)=\left|z_{1} \sqrt 8 t-z_{1}^{2}-z_{2}^{2}\right|\exp \left\{-\frac{3}{2}
t^{2}\right\}\exp \left\{-\frac{1}{2}(z_{1}^{2}+z_{2}^{2}-\sqrt 8 t z_{1})\right\};
\end{equation*}%
we note that
\begin{align*}
g^c_{2}(t_{1},t_{2})& =-\frac{6}{2}p^c_1(t_{1})p^c_1(t_{2}) \\
& \;\;+\frac{1}{2}\frac{1}{(2 \pi)^{3}}\int_{\mathbb{R}^{2}}(3t_{1}-\sqrt{2}%
z_{1})^{2}k(z_{1},z_{2},t_{1})dz_{1}dz_{2}\int_{\mathbb{R}%
^{2}}k(w_{1},w_{2},t_{2})dw_{1}dw_{2} \\
& \;\;+\frac{1}{2}\frac{1}{(2 \pi)^{3}}\int_{\mathbb{R}%
^{2}}k(z_{1},z_{2},t_{1})dz_{1}dz_{2}\int_{\mathbb{R}^{2}}(3t_{2}-\sqrt{2}%
w_{1})^{2}k(w_{1},w_{2},t_{2})dw_{1}dw_{2} \\
& =-3p^c_{1}(t_{1})p^c_{1}(t_{2})+\frac{1}{2}p^c_2(t_{1})p^c_{1}(t_{2})+\frac{1}{2}%
p^c_{1}(t_{1})p^c_2(t_{2}),
\end{align*}%
where%
\begin{align*}
p^c_2(t)&=\frac{1}{(2 \pi)^{3/2}}\int_{\mathbb{R}^{2}}(3t-\sqrt{2}%
z_{1})^{2}k(z_{1},z_{2},t)dz_{1}dz_{2}.
\end{align*}%
Note also that
\begin{equation*}
g^c_{3}(t)=\frac{1}{8}\frac{1}{(2 \pi)^{3/2}}\int_{\mathbb{R}%
^{2}}k(z_{1},z_{2},t)\left[3-(3t-\sqrt{2}z_{1})^{2}\right]dz_{1}dz_{2}=\frac{3}{8}%
p^c_{1}(t)-\frac{1}{8}p^c_2(t).
\end{equation*}%
Hence%
\begin{align*}
\left[\int_{I}{p}^c_{1}(t)dt\right]^{2}-\iint_{I\times I}
g^c_{2}(t_{1},t_{2})dt_{1}dt_{2}+16\left[\int_{I}
g^c_{3}(t)dt\right]^{2} &
=\left[\int_{I}{p}^c_{1}(t)dt\right]^{2}+3\left[\int_{I}{p}^c_{1}(t)dt\right]
^{2}-\int_{I}{p}^c_{1}(t)dt\int_{I}p^c_2(t)dt
\end{align*}

\begin{equation*}
+\frac{9}{4}\left[\int_{I}{p}^c_{1}(t)dt\right]^{2}+\frac{1}{4}\left[
\int_{I}p^c_2(t)dt\right]^{2}-\frac{3}{2}\int_{I}{p}^c_{1}(t)dt\int_{I}p^c_2(t)dt \\
=\frac{1}{4}\left[5\int_{I}p^c_{1}(t)dt-\int_{I}
p^c_2(t)dt\right]^{2},
\end{equation*}%
i.e.
\begin{align*}
\text{Var}(\mathcal{N}_{I}^{c}(f_{\ell }))&=\frac{\ell^3}{16}
 \Big[5\int_{I}p^c_{1}(t)dt-\int_{I} p^c_2(t)dt\Big]^{2} + O(\ell^{5/2})
\end{align*}
with
\begin{align*}
p^c_1(t)&=\sqrt 8 \mathbb{E} \left. \big[  |Y_1 Y_3-Y_2^2| \right| Y_1+Y_3=\sqrt 8 t  \big] \phi_{Y_1+Y_3}(\sqrt 8 t)=\frac{\sqrt 2}{ \sqrt \pi} (2 e^{-t^2}+t^2-1) e^{-\frac{t^2}{2}},\\
p^c_2(t)&=\sqrt 8 \mathbb{E} \left. \big[  (3 t-\sqrt 2 Y_1)^2 |Y_1 Y_3-Y_2^2| \right| Y_1+Y_3=\sqrt 8 t  \big] \phi_{Y_1+Y_3}(\sqrt 8 t).
\end{align*}
We now derive an analytic expression for $p^c_2(t)$. As in the
proof of Proposition \ref{expectation copy(1)} we first write
$p^c_2(t)$ as
\begin{align*}
p^c_2(t)&=\sqrt 8 \mathbb{E}  \left[  \left(3 t-\sqrt 2 (Z_1+\sqrt 2 t) \right)^2 \cdot\left| \sqrt 8 t (Z_1+\sqrt 2 t)- (Z_1+\sqrt 2 t)^2-Z_2^2 \right|  \right] \cdot\phi_{Y_1+Y_3}(\sqrt 8 t)\\
&=\sqrt 8 \mathbb{E}  \left[  \left( t-\sqrt 2 Z_1  \right)^2 \cdot\left| -Z_1^2-Z_2^2+2
t^2 \right|  \right]\cdot \phi_{Y_1+Y_3}(\sqrt 8 t),
\end{align*}
where $Z_1$, $Z_2$ denote standard independent Gaussian variables. Now
\begin{align*}
\phi_{Y_1+Y_3}(\sqrt 8 t)=\frac{1}{ 4 \sqrt{ \pi}} e^{-\frac{t^2}{2}},
\end{align*}
and we need to compute
$$\mathbb{E}  \big[  ( t-\sqrt 2 Z_1  )^2 | -Z_1^2-Z_2^2+2 t^2 |  \big].$$
The joint density function of $\xi=Z_1$ and $\zeta=Z_1^2+Z_2^2$ is given by
\begin{align*}
f_{(\xi,\zeta)}(u,v)&= \frac{\partial^2}{\partial u \partial v} \mathbb{P} [\xi < u, \zeta< v]=\frac{\partial^2}{\partial u \partial v} \mathbb{P} [Z_1 < u, Z_1^2+Z_2^2< v] =\frac{1}{2 \pi} \frac{\partial^2}{\partial u \partial v} \int\limits_{\stackrel{Z_1 < u}{0\le Z_1^2+Z_2^2 <v}}  e^{-\frac{z_1^2+z_2^2}{2}} d z_1 d z_2,
\end{align*}
i.e.,
\begin{align*}
f_{(\xi,\zeta)}(u,v)&=  \begin{cases}
0 & u \le - \sqrt v,\\
\frac{1}{2 \pi} \frac{\partial^2}{\partial u \partial v} \int_{-\sqrt v}^{u} d z_1 \int_{-\sqrt{v-z_1^2}}^{\sqrt{v-z_1^2}} e^{-\frac{z_1^2+z_2^2}{2}} d z_2&  -\sqrt{v}< u \le 0,\\
\frac{1}{2 \pi} \frac{\partial^2}{\partial u \partial v} \int_{-\sqrt v}^{u} d z_1 \int_{-\sqrt{v-z_1^2}}^{\sqrt{v-z_1^2}} e^{-\frac{z_1^2+z_2^2}{2}} d z_2 & 0< u < \sqrt v,   \\
  0 & u \ge \sqrt v,
 \end{cases}\\
 & =\frac{1}{2 \pi} \frac{e^{- \frac v 2}}{\sqrt{v-u^2}} \ind_{\{v\ge 0,\; u \in (-\sqrt v, \sqrt v )\}}.
\end{align*}
Then
\begin{align*}
\mathbb{E}  \big[  ( t-\sqrt 2 \xi  )^2 | -\zeta+2 t^2 |  \big]&= \int\limits_{0}^{\infty} d v \int\limits_{-\sqrt{v}}^{\sqrt v}
( t-\sqrt 2 u  )^2 \cdot \left| -v+2 t^2 \right| \frac{1}{2 \pi} \frac{e^{- \frac v 2}}{\sqrt{v-u^2}}  d u\\
&= \int_{0}^{2 t^2} d v \int_{-\sqrt{v}}^{\sqrt v}
( t-\sqrt 2 u  )^2 ( -v+2 t^2 ) \frac{1}{2 \pi} \frac{e^{- \frac v 2}}{\sqrt{v-u^2}} d u \\
&+ \int_{2 t^2}^{\infty} d v \int_{-\sqrt{v}}^{\sqrt v}
( t-\sqrt 2 u  )^2 ( v-2 t^2 ) \frac{1}{2 \pi} \frac{e^{- \frac v 2}}{\sqrt{v-u^2}} d u.
\end{align*}
Now
\begin{align*}
\int \frac{(t-\sqrt 2 u)^2}{\sqrt{v-u^2}} du&=t^2 \int \frac{1}{\sqrt{v-u^2}} du-2 \sqrt 2 t \int \frac{u}{\sqrt{v-u^2}} du + 2 \int \frac{ u^2}{\sqrt{v-u^2}} du,
\end{align*}
where for any symmetric interval
\begin{align*}
 \int_{-\sqrt{v}}^{\sqrt{v}} \frac{u}{\sqrt{v-u^2}} du =0,
\end{align*}
while
\begin{align*}
\int_{-\sqrt{v}}^{\sqrt{v}} \frac{1}{\sqrt{v-u^2}} du=\left[ -\arctan \Big(\frac{u \sqrt{v-u^2}}{u^2-v}\Big)\right]_{-\sqrt{v}}^{\sqrt{v}}=\pi,
\end{align*}
and
\begin{align*}
 \int_{-\sqrt{v}}^{\sqrt{v}}  \frac{ u^2}{\sqrt{v-u^2}} du=\left[\frac 1 2 \left(-u \sqrt{v-u^2}+ v \arctan \Big(\frac{u }{v-u^2}\Big) \right) \right]_{-\sqrt{v}}^{\sqrt{v}}=\frac {\pi v} 2.
\end{align*}
Hence we have
\begin{align*}
\mathbb{E}  \big[  ( t-\sqrt 2 \xi  )^2 | -\zeta+2 t^2 |  \big]
&= \frac{1}{2} \int_{0}^{2 t^2} e^{- \frac v 2} (-v+2 t^2) (v+t^2) d v   + \frac 1 2 \int_{2 t^2}^{\infty} e^{-\frac v 2} (v-2t^2) (v+t^2) d v  \\
&= 2 [-4+t^2+t^4+e^{-t^2} (4+3 t^2)]+2 e^{-t^2} (4+3 t^2),
\end{align*}
which leads to
\begin{align*}
p^c_2(t)= [-4+t^2+t^4+e^{- t^2} 2(4+3 t^2)] \frac{\sqrt 2}{  \sqrt{ \pi}} e^{-\frac{t^2}{2}}.
\end{align*}
Finally,
\begin{align*}
5 p_1^c(t)-p_2^c(t)
&=\frac{\sqrt 2}{ \sqrt \pi} e^{- \frac{3 }{2}t^2} [2-6 t^2-e^{t^2} (1-4 t^2+t^4)].
\end{align*}
Similarly, for the extrema, we have
\begin{align*}
p_2^e(t)&=\sqrt 8 \mathbb{E}[(3 t -\sqrt 2 Y_1) |Y_1 Y_3-Y_2^2| \ind_{\{Y_1 Y_3-Y_2^2>0\}} | Y_1+Y_3=\sqrt 8 t] \phi_{Y_1+Y_3} (\sqrt 8 t)\\
&=\sqrt 8 \mathbb{E}[(t-\sqrt 2 \xi)^2 |-\xi+2 t^2| \ind_{\{-\xi +2 t^2>0\}}] \phi_{Y_1+Y_3} (\sqrt 8 t),
\end{align*}
where
\begin{align*}
\mathbb{E}[(t-\sqrt 2 \xi)^2 |-\xi+2 t^2| \ind_{\{-\xi +2 t^2>0\}}] 
&=2 [-4+t^2+t^4+e^{- t^2} (4+3 t^2)],
\end{align*}
so that
\begin{align*}
p_2^e(t)= \left[-4+t^2+t^4+e^{- t^2} (4+3 t^2)\right] \frac{\sqrt 2}{ \sqrt \pi} e^{-\frac {t^2} {2}}
\end{align*}
and
\begin{align*}
5 p_1^e(t)-p_2^e(t)=\frac{\sqrt 2}{\sqrt \pi} e^{- \frac{3 }{2}t^2} [1-3 t^2-e^{t^2} (1-4 t^2+t^4)].
\end{align*}
Finally, applying the same methods for the saddles, we have
\begin{align*}
\mathbb{E}[(3 t -\sqrt 2 Y_1) |Y_1 Y_3-Y_2^2| \ind_{\{Y_1 Y_3-Y_2^2<0\}} | Y_1+Y_3=\sqrt 8 t]
&=\frac{1}{2 } \int_{2 t^2}^{\infty}  (t^2+v)   (v-2 t^2)   e^{-\frac v 2}   d v\\
&=2(4+3 t^2) e^{-t^2},
\end{align*}
yielding
\begin{align*}
p^s_2(t)=\frac{\sqrt 2}{ \sqrt \pi} (4+3 t^2) e^{-\frac {3}{2} t^2},
\end{align*}
and
\begin{align*}
5 p^s_1(t)-p^s_2(t)=\frac{\sqrt 2}{\sqrt \pi} (1-3 t^2) e^{-\frac {3}{2} t^2},
\end{align*}
as claimed.

\end{proof}

\section{Convergence of empirical measures} \label{sei}

The following auxiliary lemma shows that the empirical measures under random
and deterministic normalizations are asymptotically equivalent, uniformly in
$z$.

\begin{lemma}
\label{lma} For all $\varepsilon >0$, as $\ell \rightarrow \infty $,
\begin{equation*}
\mathbb{P}\{\sup_{z}|F_{\ell }(z)-F_{\ell }^{\ast }(z)|\geq \varepsilon
\}\rightarrow 0,
\end{equation*}
\end{lemma}
\begin{proof}
We first note that
\begin{align*}
|F_\ell(z)-F_\ell^*(z)| &= F_\ell^*(z) \left| 1- \frac{F_\ell(z)}{F_\ell^*(z)} \right| = F_\ell^*(z) \left| 1- \frac{{\cal N}^c_{\mathbb{R}}(f_\ell)}{\mathbb{E}[{\cal N}^c_{\mathbb{R}}(f_\ell)} \right| \le  \left| 1- \frac{{\cal N}^c_{\mathbb{R}}(f_\ell)}{\mathbb{E}[{\cal N}^c_{\mathbb{R}}(f_\ell)} \right|.
\end{align*}
The statement of the present lemma follows by observing that
 from Proposition \ref{expectation copy(1)} and Theorem \ref{th_variance copy(1)} we have that
\begin{align*}
\frac{{\cal N}^c_{\mathbb{R}}(f_\ell)}{\mathbb{E}[{\cal N}^c_{\mathbb{R}}(f_\ell)]}
\end{align*}
is a random variable with unitary mean and variance $O(\ell^{-1})$.
\end{proof}

\noindent We can now provide the proof of Proposition \ref{emprcl}.

\begin{proof}[Proof of Proposition \ref{emprcl}]
We first note that in view Lemma \ref{lma} proving Proposition \ref{emprcl} is equivalent to proving that for all $\varepsilon >0$ and $\delta>0$ there exists $\ell_{\varepsilon,\delta}$ such that for all $\ell>\ell_{\varepsilon, \delta }$ we have
\begin{align*}
\mathbb{P}\{ \sup_z |  F_\ell(z) -\Phi_{\infty}(z)|>\varepsilon\} \le \delta.
\end{align*}
Fix $\varepsilon>0$ and choose $K_{\varepsilon}>0$ sufficiently big such that $1/ K_{\varepsilon}< \varepsilon/2$. Now we define the partitions
$$
-\infty=x_1 \le x_2 \le \dots \le x_{K_{\varepsilon}}= \infty,
$$
such that
\begin{align*}
\Phi_{\infty}(x_{k+1})-\Phi_{\infty}(x_k)< \frac{\varepsilon}{2}.
\end{align*}
For every $z$ there exist $i^{-}_{K_{\varepsilon}}(z)$
and $i^{+}_{K_{\varepsilon}}(z) \in \{x_1, x_2,  \dots, x_{K_{\varepsilon}} \}$ such that
$$z \in (i^{-}_{K_{\varepsilon}}(z), i^{+}_{K_{\varepsilon}}(z)).$$
Then, since $F_\ell(z)$ and $\Phi_{\infty}(z)$ are both non decreasing in $z$,  we have
\begin{align*}
F_\ell(z) -\Phi_{\infty}(z) &\le F(i^{+}_{K_{\varepsilon}}(z) ;f_\ell) -\Phi_{\infty}(i^{+}_{K_{\varepsilon}}(z))+\frac{\varepsilon}{2},\\
F_\ell(z) -\Phi_{\infty}(z) &\ge F(i^{-}_{K_{\varepsilon}}(z) ;f_\ell) - \Phi_{\infty}(i^{-}_{K_{\varepsilon}}(z))-\frac{\varepsilon}{2},
\end{align*}
so that
\begin{align*}
\sup_{z} \left|F_\ell(z) -\Phi_{\infty}(z)  \right| \le \max_{k=1, \dots, K_{\varepsilon}}  \left|F_\ell(x_k) -\Phi_{\infty}(x_k)  \right|+ \frac {\varepsilon} 2.
\end{align*}
Then
\begin{align*}
\mathbb{P}\{ \sup_z |  F_\ell(z) -\Phi_{\infty}(z)|>\varepsilon\} \le \mathbb{P}\left\{  \max_{k=1, \dots, K_{\varepsilon}}  \left|F_\ell(x_k) -\Phi_{\infty}(x_k)  \right| > \frac{\varepsilon}{2}\right\}
\end{align*}
and, in view of Proposition \ref{th_variance copy(1)}, each of the $K_{\varepsilon}$ random variables
$$\left|F_\ell(x_k) -\Phi_{\infty}(x_k)  \right|,$$
converges in probability to zero.
\end{proof}

\begin{appendices}

\section{Evaluation of covariance matrices}

\label{cov_matx_ev}

In this section we compute the covariance matrix  $\Sigma _{\ell
}(x,y)$ for the $10$-dimensional random vector $Z_{\ell ;x,y}$,
which combines the gradient and the elements of the Hessian
evaluated at $x,y$. $\Sigma _{\ell }(x,y)$ depends only on the
geodesic distance $\phi =d(x,y)$, so, abusing notation, we shall
write $\Sigma _{\ell }(x,y)=\Sigma _{\ell }(\phi )$ whenever
convenient, and similarly for the other functions we shall deal
with.
 The computations are quite lengthy, but they do not require
sophisticated arguments, other than iterative derivations of Legendre
polynomials. It is convenient
to write these matrices in block-diagonal form, i.e.
\begin{equation*}
\Sigma _{\ell }(\phi)=\left(
\begin{array}{cc}
A_{\ell}(\phi) & B_{\ell}(\phi) \\
B^{t}_{\ell}(\phi) & C_{\ell}(\phi)
\end{array}
\right).
\end{equation*}
In particular the $A_\ell$ component collects the variances of the
gradient
terms, and it is given by%
\begin{align*}
A_{\ell }(x,y)_{4\times 4}&=\left. \mathbb{E}\left[ \left(
\begin{array}{c}
\nabla f_{\ell }(\bar{x})^{t} \\
\nabla f_{\ell }(\bar{y})^{t}%
\end{array}%
\right) \left(
\begin{array}{cc}
\nabla f_{\ell }(x) & \nabla f_{\ell }(y)%
\end{array}%
\right) \right] \right\vert _{{x=\bar{x}},{y=\bar{y}}}\\
&=\left(
\begin{array}{cc}
a_{\ell }(x,x) & a_{\ell }(x,y) \\
a_{\ell }(y,x) & a_{\ell }(y,y)%
\end{array}%
\right),
\end{align*}%
where
\begin{equation*}
a_{\ell }(x,x)=\left. \left(
\begin{array}{cc}
e_{1}^{\bar{x}}e_{1}^{x}r_{\ell }(\bar{x},x) & e_{1}^{\bar{x}%
}e_{2}^{x}r_{\ell }(\bar{x},x) \\
e_{2}^{\bar{x}}e_{1}^{x}r_{\ell }(\bar{x},x) & e_{2}^{\bar{x}%
}e_{2}^{x}r_{\ell }(\bar{x},x)%
\end{array}%
\right) \right\vert _{x=\bar{x}}, \hspace{0,5cm} a_{\ell }(x,y)=\left. \left(
\begin{array}{cc}
e_{1}^{\bar{x}}e_{1}^{y}r_{\ell }(\bar{x},y) & e_{1}^{\bar{x}%
}e_{2}^{y}r_{\ell }(\bar{x},y) \\
e_{2}^{\bar{x}}e_{1}^{y}r_{\ell }(\bar{x},y) & e_{2}^{\bar{x}%
}e_{2}^{y}r_{\ell }(\bar{x},y)%
\end{array}%
\right) \right\vert _{x=\bar{x}},
\end{equation*}
and $$r_{\ell}(x,y)=\mathbb{E}[f_{\ell}(x) f_{\ell}(y)]=P_{\ell}(\cos d(x,y)),$$
 with $h(x,y)=\cos d(x,y)=\cos \theta_x \cos \theta_y+\sin \theta_x \sin \theta_y \cos(\varphi_x-\varphi_y)$.  Then, for example,  computing explicitly the derivatives, we have
 \begin{align*}
 e_{1}^{\bar{x}}e_{1}^{x}r_{\ell }(\bar{x},x)=P''_{\ell}(h(\bar{x},x)) \frac{\partial}{\partial \theta_{\bar{x}}} h(\bar{x},x)  \frac{\partial}{\partial \theta_{x}} h(\bar{x},x)+P'_{\ell}(h(\bar{x},x)) \frac{\partial}{\partial \theta_{\bar{x}}}  \frac{\partial}{\partial \theta_{x}} h(\bar{x},x),
 \end{align*}
 where
 \begin{align*}
 \frac{\partial}{\partial \theta_{\bar{x}}} h(\bar{x},x)&=\left. -\cos \theta_{{x}} \sin \theta_{\bar{x}}+\cos \theta_{\bar{x}} \sin \theta_{{x}} \cos( \varphi_{x}-\varphi_{\bar{x}}) \right\vert _{x=\bar{x}=(\pi /2,\varphi _{x})}=0,\\
 \frac{\partial}{\partial \theta_{x}} h(\bar{x},x) &=\left.  -\cos \theta_{\bar{x}} \sin \theta_{{x}}+\cos \theta_{{x}} \sin \theta_{\bar{x}} \cos( \varphi_{x}-\varphi_{\bar{x}}) \right\vert _{x=\bar{x}=(\pi /2,\varphi _{x})}=0,\\
 \frac{\partial}{\partial \theta_{\bar{x}}}  \frac{\partial}{\partial \theta_{x}} h(\bar{x},x) &=\left.  \sin \theta_{\bar{x}} \sin \theta_{{x}}+\cos \theta_{{x}} \cos \theta_{\bar{x}} \cos( \varphi_{x}-\varphi_{\bar{x}}) \right\vert _{x=\bar{x}=(\pi /2,\varphi _{x})}=1.
 \end{align*}
 We write then
 \begin{equation*}
\left. a_{\ell }(x,x)\right\vert _{x=(\pi /2,\varphi _{x})}=\left. a_{\ell
}(y,y)\right\vert _{y=(\pi /2,0)}=\left(
\begin{array}{cc}
P_{\ell }^{\prime }(1) & 0 \\
0 & P_{\ell }^{\prime }(1)%
\end{array}%
\right),
\end{equation*}%
and, again with some slight abuse of notation,
\begin{align*}
\left. a_{\ell }(x,y)\right\vert _{{x=(\pi /2,\varphi _{x})},{y=(\pi /2,0)}%
}&=\left. a_{\ell }(y,x)\right\vert _{{x=(\pi /2,\varphi _{x})},{%
y=(\pi /2,0)}}=\left(
\begin{array}{cc}
\alpha _{1,\ell }(\phi ) & 0 \\
0 & \alpha _{2,\ell }(\phi )%
\end{array}%
\right),
\end{align*}%
where as we recalled before $\phi =d(x,y)$ and
\begin{equation*}
\alpha _{1,\ell }(\phi )=P_{\ell }^{\prime }(\cos \phi ),
\end{equation*}
\begin{equation*}
\alpha _{2,\ell }(\phi )=-\sin ^{2}\phi P_{\ell }^{\prime \prime }(\cos \phi
)+\cos \phi P_{\ell }^{\prime }(\cos \phi ).
\end{equation*}
Now recall that $P_{\ell }^{\prime }(1)=\frac{\ell (\ell +1)}{2}$, for $
\lambda _{\ell }=\ell (\ell +1)$; hence we have
\begin{equation*}
A_{\ell }(\phi )=\left(
\begin{array}{cccc}
\frac{\lambda _{\ell }}{2} & 0 & \alpha _{1,\ell }(\phi ) & 0 \\
0 & \frac{\lambda _{\ell }}{2} & 0 & \alpha _{2,\ell }(\phi ) \\
\alpha _{1,\ell }(\phi ) & 0 & \frac{\lambda _{\ell }}{2} & 0 \\
0 & \alpha _{2,\ell }(\phi ) & 0 & \frac{\lambda _{\ell }}{2}%
\end{array}%
\right) .
\end{equation*}%

\noindent The matrix $B_{\ell}$ collects the covariances between first and second order
derivatives, and is given by
\begin{align*}
B_{\ell }(x,y)_{4\times 6}&=\left. \mathbb{E}\left[ \left(
\begin{array}{c}
\nabla f_{\ell }(\bar{x})^{t} \\
\nabla f_{\ell }(\bar{y})^{t}%
\end{array}%
\right) \left(
\begin{array}{cc}
\nabla ^{2}f_{\ell }(x) & \nabla ^{2}f_{\ell }(y)%
\end{array}%
\right) \right] \right\vert _{{x=\bar{x}},{y=\bar{y}}}\\
&=\left(
\begin{array}{cc}
b_{\ell }(x,x) & b_{\ell }(x,y) \\
b_{\ell }(y,x) & b_{\ell }(y,y)%
\end{array}%
\right).
\end{align*}
It is well-known that for Gaussian isotropic processes, for $i,j=1,2$, the second derivatives $e^x_i e^x_j f_\ell(x)$ are independent of $e_i^x f_\ell(x)$ at every fixed point $x \in {\cal S}^2$ see, e.g, \cite{adlertaylor} section 5.5; we have then
\begin{equation*}
\left. b_{\ell }(x,x)\right\vert _{x=(\pi /2,\varphi _{x})}=\left. b_{\ell
}(y,y)\right\vert _{y=(\pi /2,0)}=\left(
\begin{array}{ccc}
0 & 0 & 0 \\
0 & 0 & 0%
\end{array}%
\right),
\end{equation*}%
while
\begin{align*}
\left. b_{\ell }(x,y)\right\vert _{{x=(\pi /2,\varphi _{x})},{y=(\pi /2,0)}%
}& =\left(
\begin{array}{ccc}
0 & \beta _{1,\ell }(\phi ) & 0 \\
\beta _{2,\ell }(\phi ) & 0 & \beta _{3,\ell }(\phi )%
\end{array}%
\right) \\
& =-\left. b_{\ell }(y,x)\right\vert _{{x=(\pi /2,\varphi _{x})},{y=(\pi
/2,0)}}.
\end{align*}%
Here we have introduced the functions%
\begin{equation*}
\beta _{1,\ell }(\phi )=\sin \phi P_{\ell }^{\prime \prime }(\cos \phi ),
\end{equation*}%
\begin{equation*}
\beta _{2,\ell }(\phi )=\sin \phi \cos \phi P_{\ell }^{\prime \prime }(\cos
\phi )+\sin \phi P_{\ell }^{\prime }(\cos \phi ),
\end{equation*}%
\begin{equation*}
\beta _{3,\ell }(\phi )=-\sin ^{3}\phi P_{\ell }^{\prime \prime \prime
}(\cos \phi )+3\sin \phi \cos \phi P_{\ell }^{\prime \prime }(\cos \phi
)+\sin \phi P_{\ell }^{\prime }(\cos \phi ).\newline
\end{equation*}

\noindent Finally, the matrix $C_{\ell}$ contains the variances of second-order
derivatives, and we have
\begin{align*}
C_{\ell }(x,y)_{6\times 6}&=\left. \mathbb{E}\left[
\begin{array}{c}
\left(
\begin{array}{c}
\nabla ^{2}f_{\ell }(\bar{x})^{t} \\
\nabla ^{2}f_{\ell }(\bar{y})^{t}%
\end{array}%
\right) \left(
\begin{array}{cc}
\nabla ^{2}f_{\ell }({x}) & \nabla ^{2}f_{\ell }(\bar{y})%
\end{array}%
\right)%
\end{array}%
\right] \right\vert _{{x=\bar{x}},{y=\bar{y}}}\\
&=\left(
\begin{array}{cc}
c_{\ell }(x,x) & c_{\ell }(x,y) \\
c_{\ell }(y,x) & c_{\ell }(y,y)%
\end{array}%
\right).
\end{align*}%
Direct calculations yield
\begin{align*}
\left. c_{\ell }(x,y)\right\vert _{{x=(\pi /2,\varphi _{x})},{y=(\pi /2,0)}%
}&=\left. c_{\ell }(y,x)\right\vert _{{x=(\pi /2,\varphi _{x})},{%
y=(\pi /2,0)}}\\
&=\left(
\begin{array}{ccc}
\gamma _{1,\ell }(\phi ) & 0 & \gamma _{3,\ell }(\phi ) \\
0 & \gamma _{2,\ell }(\phi ) & 0 \\
\gamma _{3,\ell }(\phi ) & 0 & \gamma _{4,\ell }(\phi )%
\end{array}%
\right),
\end{align*}%
with%
\begin{equation*}
\gamma _{1,\ell }(\phi )=(2+\cos ^{2}\phi )P_{\ell }^{\prime \prime }(\cos
\phi )+\cos \phi P_{\ell }^{\prime }(\cos \phi ),
\end{equation*}%
\begin{equation*}
\gamma _{2,\ell }(\phi )=-\sin ^{2}\phi P_{\ell }^{\prime \prime \prime
}(\cos \phi )+\cos \phi P_{\ell }^{\prime \prime }(\cos \phi ),
\end{equation*}%
\begin{equation*}
\gamma _{3,\ell }(\phi )=-\sin ^{2}\phi \cos \phi P_{\ell }^{\prime \prime
\prime }(\cos \phi )+(-2\sin ^{2}\phi +\cos ^{2}\phi )P_{\ell }^{\prime
\prime }(\cos \phi )+\cos \phi P_{\ell }^{\prime }(\cos \phi ),
\end{equation*}
\begin{equation*}
\gamma _{4,\ell }(\phi )=\sin ^{4}\phi P_{\ell }^{\prime \prime \prime
\prime }(\cos \phi )-6\sin ^{2}\phi \cos \phi P_{\ell }^{\prime \prime
\prime }(\cos \phi )+(-4\sin ^{2}\phi +3\cos ^{2}\phi )P_{\ell }^{\prime \prime }(\cos \phi
)+\cos \phi P_{\ell }^{\prime }(\cos \phi ).
\end{equation*}
Since $P_{\ell }^{\prime \prime }(1)=\frac{\lambda_\ell}{8}(\lambda_\ell-2)$, it immediately follows that
\begin{align*}
\left. c_{\ell }(x,x)\right\vert _{x=(\pi /2,\varphi _{x})}& =\left(
\begin{array}{ccc}
3P_{\ell }^{\prime \prime }(1)+P_{\ell }^{\prime }(1) & 0 & P_{\ell
}^{\prime \prime }(1)+P_{\ell }^{\prime }(1) \\
0 & P_{\ell }^{\prime \prime }(1) & 0 \\
P_{\ell }^{\prime \prime }(1)+P_{\ell }^{\prime }(1) & 0 & 3P_{\ell
}^{\prime \prime }(1)+P_{\ell }^{\prime }(1)%
\end{array}%
\right) \\
& =\left(
\begin{array}{ccc}
\frac{\lambda _{\ell }}{8}[3\lambda _{\ell }-2] & 0 & \frac{\lambda _{\ell }%
}{8}[\lambda _{\ell }+2] \\
0 & \frac{\lambda _{\ell }}{8}[\lambda _{\ell }-2] & 0 \\
\frac{\lambda _{\ell }}{8}[\lambda _{\ell }+2] & 0 & \frac{\lambda _{\ell }}{%
8}[3\lambda _{\ell }-2]%
\end{array}%
\right) =\left. c_{\ell }(y,y)\right\vert _{y=(\pi /2,0)}.
\end{align*}

\section{The conditional covariance matrix $\Delta_{\ell }(\phi )$}

\label{cov_matx}

In this section we compute the conditional covariance matrices
$\Omega _{\ell }(\phi )$ and $\Delta _{\ell }(\phi )$ (eqs.
\ref{2agosto}, \ref{1agosto}). To simplify the notation we will
write
 $\alpha _{i}, \beta_{i}, \gamma_{i}$ for  $\alpha _{i,\ell }(\phi )$, $\beta _{i,\ell }(\phi )$ and $\gamma
_{i,\ell }(\phi )$; likewise we will adopt the shorthand notation $A,B,C,\Omega, \Delta$ for $A_{\ell }(\phi )$, $%
B_{\ell }(\phi )$, $C_{\ell }(\phi )$, $\Omega _{\ell }(\phi
) $ and $\Delta _{\ell }(\phi )$, respectively.

Let us first compute explicitly the inverse matrix $A^{-1}$; we
write $A$ as a block matrix
\begin{equation*}
A=\left(
\begin{array}{cc}
a_1 & a_2 \\
a_2 & a_1%
\end{array}%
\right)
\end{equation*}%
where
\begin{equation*}
a_1=\left(
\begin{array}{cc}
\frac{\lambda _{\ell }}{2} & 0 \\
0 & \frac{\lambda _{\ell }}{2}%
\end{array}%
\right) ,\hspace{1cm}a_2=\left(
\begin{array}{cc}
\alpha _{1} & 0 \\
0 & \alpha _{2}%
\end{array}%
\right) ,
\end{equation*}%
and we evaluate the following components:
\begin{equation*}
(a_1-a_2a_1^{-1}a_2)^{-1}=\left(
\begin{array}{cc}
\frac{2\lambda _{\ell }}{\lambda _{\ell }^{2}-4\alpha _{1}^{2}} & 0 \\
0 & \frac{2\lambda _{\ell }}{\lambda _{\ell }^{2}-4\alpha _{2}^{2}}%
\end{array}%
\right) ,
\end{equation*}%
and
\begin{equation*}
a_1^{-1}a_2=a_2a_1^{-1}=\left(
\begin{array}{cc}
\frac{2\alpha _{1}}{\lambda _{\ell }} & 0 \\
0 & \frac{2\alpha _{2}}{\lambda _{\ell }}%
\end{array}%
\right) .
\end{equation*}%
Now, to invert blockwise $A$, we need to compute the main diagonal
blocks
\begin{equation*}
a_1^{-1}+a_1^{-1}a_2(a_1-a_2a_1^{-1}a_2)^{-1}a_2a_1^{-1}
=(a_1-a_2a_1^{-1}a_2)^{-1},
\end{equation*}%
and the off-diagonal blocks
\begin{equation*}
-a_1^{-1}a_2(a_1-a_2a_1^{-1}a_2)^{-1}=-(a_1-a_2a_1^{-1}a_2)^{-1}a_2a_1^{-1}=-\left(
\begin{array}{cc}
\frac{4\alpha _{1}}{\lambda _{\ell }^{2}-4\alpha _{1}^{2}} & 0 \\
0 & \frac{4\alpha _{2}}{\lambda _{\ell }^{2}-4\alpha _{2}^{2}}%
\end{array}%
\right) .
\end{equation*}
We have then
\begin{equation*}
A^{-1}=\left(
\begin{array}{cccc}
\frac{2\lambda _{\ell }}{\lambda _{\ell }^{2}-4\alpha _{1}^{2}} & 0 & -\frac{%
4\alpha _{1}}{\lambda _{\ell }^{2}-4\alpha _{1}^{2}} & 0 \\
0 & \frac{2\lambda _{\ell }}{\lambda _{\ell }^{2}-4\alpha _{2}^{2}} & 0 & -%
\frac{4\alpha _{2}}{\lambda _{\ell }^{2}-4\alpha _{2}^{2}} \\
-\frac{4\alpha _{1}}{\lambda _{\ell }^{2}-4\alpha _{1}^{2}} & 0 & \frac{%
2\lambda _{\ell }}{\lambda _{\ell }^{2}-4\alpha _{1}^{2}} & 0 \\
0 & -\frac{4\alpha _{2}}{\lambda _{\ell }^{2}-4\alpha _{2}^{2}} & 0 & \frac{%
2\lambda _{\ell }}{\lambda _{\ell }^{2}-4\alpha _{2}^{2}}%
\end{array}%
\right) .
\end{equation*}%
We are now in the position to compute the matrix $B^{t}A^{-1}B$;
indeed we get:
\begin{equation*}
B^{t}A^{-1}B=\left(
\begin{array}{cccc}
0 & 0 & 0 & -\beta _{2} \\
0 & 0 & -\beta _{1} & 0 \\
0 & 0 & 0 & -\beta _{3} \\
0 & \beta _{2} & 0 & 0 \\
\beta _{1} & 0 & 0 & 0 \\
0 & \beta _{3} & 0 & 0%
\end{array}%
\right) A^{-1}\left(
\begin{array}{cccccc}
0 & 0 & 0 & 0 & \beta _{1} & 0 \\
0 & 0 & 0 & \beta _{2} & 0 & \beta _{3} \\
0 & -\beta _{1} & 0 & 0 & 0 & 0 \\
-\beta _{2} & 0 & -\beta _{3} & 0 & 0 & 0%
\end{array}%
\right)
\end{equation*}%
\begin{equation*}
=\left(
\begin{array}{cccccc}
\frac{2\lambda _{\ell }\beta _{2}^{2}}{\lambda _{\ell }^{2}-4\alpha _{2}^{2}}
& 0 & \frac{2\lambda _{\ell }\beta _{2}\beta _{3}}{\lambda _{\ell
}^{2}-4\alpha _{2}^{2}} & \frac{4\alpha _{2}\beta _{2}^{2}}{\lambda _{\ell
}^{2}-4\alpha _{2}^{2}} & 0 & \frac{4\alpha _{2}\beta _{2}\beta _{3}}{%
\lambda _{\ell }^{2}-4\alpha _{2}^{2}} \\
0 & \frac{2\lambda _{\ell }\beta _{1}^{2}}{\lambda _{\ell }^{2}-4\alpha
_{1}^{2}} & 0 & 0 & \frac{4\alpha _{1}\beta _{1}^{2}}{\lambda _{\ell
}^{2}-4\alpha _{1}^{2}} & 0 \\
\frac{2\lambda _{\ell }\beta _{2}\beta _{3}}{\lambda _{\ell }^{2}-4\alpha
_{2}^{2}} & 0 & \frac{2\lambda _{\ell }\beta _{3}^{2}}{\lambda _{\ell
}^{2}-4\alpha _{2}^{2}} & \frac{4\alpha _{2}\beta _{2}\beta _{3}}{\lambda
_{\ell }^{2}-4\alpha _{2}^{2}} & 0 & \frac{4\alpha _{2}\beta _{3}^{2}}{%
\lambda _{\ell }^{2}-4\alpha _{2}^{2}} \\
\frac{4\alpha _{2}\beta _{2}^{2}}{\lambda _{\ell }^{2}-4\alpha _{2}^{2}} & 0
& \frac{4\alpha _{2}\beta _{2}\beta _{3}}{\lambda _{\ell }^{2}-4\alpha
_{2}^{2}} & \frac{2\lambda _{\ell }\beta _{2}^{2}}{\lambda _{\ell
}^{2}-4\alpha _{2}^{2}} & 0 & \frac{2\lambda _{\ell }\beta _{2}\beta _{3}}{%
\lambda _{\ell }^{2}-4\alpha _{2}^{2}} \\
0 & \frac{4\alpha _{1}\beta _{1}^{2}}{\lambda _{\ell }^{2}-4\alpha _{1}^{2}}
& 0 & 0 & \frac{2\lambda _{\ell }\beta _{1}^{2}}{\lambda _{\ell
}^{2}-4\alpha _{1}^{2}} & 0 \\
\frac{4\alpha _{2}\beta _{2}\beta _{3}}{\lambda _{\ell }^{2}-4\alpha _{2}^{2}%
} & 0 & \frac{4\alpha _{2}\beta _{3}^{2}}{\lambda _{\ell }^{2}-4\alpha
_{2}^{2}} & \frac{2\lambda _{\ell }\beta _{2}\beta _{3}}{\lambda _{\ell
}^{2}-4\alpha _{2}^{2}} & 0 & \frac{2\lambda _{\ell }\beta _{3}^{2}}{\lambda
_{\ell }^{2}-4\alpha _{2}^{2}}%
\end{array}%
\right)\mbox{,}
\end{equation*}%
From section A in this appendix we have%
\begin{equation*}
C=\left(
\begin{array}{cccccc}
3\frac{\lambda _{\ell }(\lambda _{\ell }-2)}{8}+\frac{\lambda _{\ell }}{2} &
0 & \frac{\lambda _{\ell }(\lambda _{\ell }-2)}{8}+\frac{\lambda _{\ell }}{2}
& \gamma _{1} & 0 & \gamma _{3} \\
0 & \frac{\lambda _{\ell }(\lambda _{\ell }-2)}{8} & 0 & 0 & \gamma _{2} & 0
\\
\frac{\lambda _{\ell }(\lambda _{\ell }-2)}{8}+\frac{\lambda _{\ell }}{2} & 0
& 3\frac{\lambda _{\ell }(\lambda _{\ell }-2)}{8}+\frac{\lambda _{\ell }}{2}
& \gamma _{3} & 0 & \gamma _{4} \\
\gamma _{1} & 0 & \gamma _{3} & 3\frac{\lambda _{\ell }(\lambda _{\ell }-2)}{%
8}+\frac{\lambda _{\ell }}{2} & 0 & \frac{\lambda _{\ell }(\lambda _{\ell
}-2)}{8}+\frac{\lambda _{\ell }}{2} \\
0 & \gamma _{2} & 0 & 0 & \frac{\lambda _{\ell }(\lambda _{\ell }-2)}{8} & 0
\\
\gamma _{3} & 0 & \gamma _{4} & \frac{\lambda _{\ell }(\lambda _{\ell }-2)}{8%
}+\frac{\lambda _{\ell }}{2} & 0 & 3\frac{\lambda _{\ell }(\lambda _{\ell
}-2)}{8}+\frac{\lambda _{\ell }}{2}%
\end{array}%
\right)\mbox{;}
\end{equation*}
The remaining computations to obtain $\Omega $ and $\Delta$ are
straightforward.

\section{Some estimates on Legendre polynomials} \label{estimates}

\noindent Let us first recall the following:

\begin{lemma}[Hilb's asymptotics, \cite{szego}, page 195, Theorem
8.21.6.]
\label{hilb0} For any $\varepsilon >0$ and any constant $C>0$, we have
\begin{equation*}
P_\ell(\cos \phi)=\left( \frac{\phi}{\sin \phi}\right)^{1/2} J_0((\ell+1/2)
\phi)+\delta_\ell(\phi),
\end{equation*}
where $J_\nu$ is the Bessel function of the first kind, $P_\ell$
denotes Legendre polynomials, and the error term satisfies
\begin{align*}
\delta_\ell(\phi) \ll
\begin{cases}
\phi^2 O(1), & 0<\phi<C/\ell, \\
\phi^{1/2} O(\ell^{-3/2}), & C/ \ell \le \phi,%
\end{cases}%
\end{align*}
uniformly w.r.t. $\ell \ge 1$ and $\phi \in [0, \pi-\varepsilon]$.
\end{lemma}

\begin{lemma}
\label{bessel} The following asymptotic reppresentation for the Bessel
functions of the first kind holds:
\begin{align*}
J_0(x)&=\left( \frac{2}{\pi x} \right)^{1/2} \cos(x- \pi/4) %
\left[ \sum_{k=0}^n (-1)^k g(2k) \; (2 x)^{-2 k} +O(|x|^{-2 n -2})\right] \\
&\;\;-\left( \frac{2}{\pi x} \right)^{1/2} \sin(x- \pi/4) \left[
\sum_{k=0}^n (-1)^k g(2k+1)\; (2 x)^{-2 k-1} +O(|x|^{-2 n -3})\right]
\end{align*}
where $\varepsilon>0$, $|\arg x|\le \pi-\varepsilon$, $(0,0)=1$ and $g(k)=\frac{(-1)(-3^2) \cdots (-(2k-1)^2)}{2^{2k} k!}=(-1)^k \frac{[(2k)!!]^2}{2^{2k} k!}$.
\end{lemma}

\noindent For a proof of Lemma \ref{bessel} see \cite{lebedev}
, Section 5.11.\newline

\noindent For instance when $n=0$ we have
\begin{align}
J_0(x)&=\left( \frac{2}{\pi x} \right)^{1/2} \cos(x- \pi/4) -\frac 1 8 \left(
\frac{2}{\pi x} \right)^{1/2} \frac{ 1 }{x} \cos(x+ \pi/4)+O(|x|^{-5/2}).
\label{J0}
\end{align}
In the rest of the paper we use the following notation:
\begin{align*}
&R_1(\ell,\phi)=O(\ell^{-1/2} \phi^{-5/2}), &
R_2(\ell,\phi) =O(\ell^{1/2} \phi^{-7/2}), \\
&R_3(\ell,\phi)=O( \ell^{1/2} \phi^{-11/2}), & R_4(\ell,\phi)=O(\ell^{1/2} \phi^{-15/2} ).
\end{align*}

\begin{lemma}
\label{hilbs} For any constant $C>0$, we have, uniformly for $\ell \geq 1$ and $\phi \in \lbrack C/\ell ,\pi /2]$: 
\begin{equation} \label{P}
P_{\ell }(\cos \phi )=\sqrt{\frac{2}{\pi }}\frac{\ell ^{-1/2}}{\sin
^{1/2}\phi }\left[\cos \psi _{\ell }^{-}-\frac{1}{8\ell \phi }\cos \psi
_{\ell }^{+}\right]+O(\ell ^{-5/2}\phi ^{-5/2})+O(\ell ^{-3/2}\phi ^{-1/2}),
\end{equation}
\begin{equation} \label{P'}
P_{\ell }^{\prime }(\cos \phi )=\sqrt{\frac{2}{\pi }}\frac{\ell ^{1-1/2}}{%
\sin ^{1+1/2}\phi }\left[\sin \psi _{\ell }^{-}-\frac{1}{8\ell \phi }\sin
\psi _{\ell }^{+}\right]+R_{1}(\ell ,\phi ),
\end{equation}
\begin{equation} \label{P''}
P_{\ell }^{\prime \prime }(\cos \phi )=\sqrt{\frac{2}{\pi }}\frac{\ell
^{2-1/2}}{\sin ^{2+1/2}\phi }\left[-\cos \psi _{\ell }^{-}+\frac{1}{8\ell
\phi }\cos \psi _{\ell }^{+}\right]-\sqrt{\frac{2}{\pi }}\frac{\ell ^{1-1/2}}{\sin ^{3+1/2}\phi }\left[\cos
\psi _{\ell -1}^{+}+\frac{1}{8\ell \phi }\cos \psi _{\ell -1}^{-}\right]%
+R_{2}(\ell ,\phi ),
\end{equation}
\begin{align}
P_{\ell }^{\prime \prime \prime }(\cos \phi )& =\sqrt{\frac{2}{\pi }}\frac{%
\ell ^{3-1/2}}{\sin ^{3+1/2}\phi }\left[\cos \psi _{\ell }^{+}+\frac{1}{8\ell
\phi }\cos \psi _{\ell }^{-}\right]  \nonumber  \\
& \;\;-\sqrt{\frac{2}{\pi }}\frac{\ell ^{2-1/2}}{\sin ^{4+1/2}\phi }\left[\frac{1%
}{2}(\cos \psi _{\ell +1}^{-}+5\cos \psi _{\ell -1}^{-})-\frac{1}{8\ell \phi
}\frac{1}{2}(\cos \psi _{\ell +1}^{+}+5\cos \psi _{\ell -1}^{+})\right]  \notag \\
& \;\;+\sqrt{\frac{2}{\pi }}\frac{\ell ^{1-1/2}}{\sin ^{5+1/2}\phi }\left[3\cos
\phi \sin \psi _{\ell -1}^{-}-\frac{1}{8\ell \phi }3\cos \phi \sin \psi
_{\ell -1}^{+}\right]+R_{3}(\ell ,\phi ),  \label{P'''}
\end{align}
\begin{align}
P_{\ell }^{\prime \prime \prime \prime }(\cos \phi )& =\sqrt{\frac{2}{\pi }}%
\frac{\ell ^{4-1/2}}{\sin ^{4+1/2}\phi }\left[\cos \psi _{\ell }^{-}-\frac{1}{%
8\ell \phi }\cos \psi _{\ell }^{+}\right]  \notag \\
& +\sqrt{\frac{2}{\pi }}\frac{\ell ^{3-1/2}}{\sin ^{5+1/2}\phi }\left[-\frac{3}{2}%
(\sin \psi _{\ell +1}^{-}+3\sin \psi _{\ell -1}^{-})+\frac{1}{8\ell \phi }%
\frac{3}{2}(\sin \psi _{\ell +1}^{+}+3\sin \psi _{\ell -1}^{+})\right]  \notag \\
& +\sqrt{\frac{2}{\pi }}\frac{\ell ^{2-1/2}}{\sin ^{6+1/2}\phi }\left[-\frac{1}{2}%
(\cos \psi _{\ell +2}^{-}+16\cos \psi _{\ell }^{-}+13\cos \psi _{\ell
-2}^{-})+\frac{1}{8\ell \phi }\frac{1}{2}(\cos \psi _{\ell +2}^{+}+16\cos
\psi _{\ell }^{+}+13\cos \psi _{\ell -2}^{+})\right]  \notag \\
& +\sqrt{\frac{2}{\pi }}\frac{\ell ^{1-1/2}}{\sin ^{7+1/2}\phi }\left[-3(5-4\sin
^{2}\phi )\cos \psi _{\ell -1}^{+}-\frac{1}{8\ell \phi }3(5-4\sin ^{2}\phi
)\cos \psi _{\ell -1}^{-}\right]+R_{4}(\ell ,\phi ),  \label{P''''}
\end{align}%
where $\psi _{\ell +k}^{\pm }=(\ell +k+1/2)\phi \pm \pi /4$.
\end{lemma}

\begin{proof}
By applying the Hilb's asymptotics in Lemma \ref{hilb0} and, in view of formula \eqref{J0}, we obtain
\begin{align*}
P_\ell(\cos \phi)&=\left( \frac{\phi}{\sin \phi}\right)^{1/2} \left( \frac{2}{\pi (\ell+1/2) \phi } \right)^{1/2} \Big[  \sin((\ell+1/2) \phi+ \pi/4) -\frac 1 8  \frac{ 1 }{(\ell+1/2) \phi} \cos((\ell+1/2) \phi+ \pi/4) \Big] \\
&\;\; + \left( \frac{\phi}{\sin \phi}\right)^{1/2} O(|(\ell+1/2) \phi|^{-5/2}) +\phi^{1/2} O(\ell^{-3/2}) \\
&= \sqrt{\frac{2}{\pi}} \left( \frac{1}{\sin \phi}\right)^{1/2} \left( \frac{1}{\sqrt \ell}+O(\ell^{-3/2}) \right) \Big[  \sin((\ell+1/2) \phi+ \pi/4) -\frac{1}{8 \phi}  \left(  \frac 1 {\ell} +O(\ell^{-2}) \right) \cos((\ell+1/2) \phi+ \pi/4) \Big] \\
&\;\; + \left( \frac{\phi}{\sin \phi}\right)^{1/2} O( \ell^{-5/2} \phi^{-5/2})  +\phi^{1/2} O(\ell^{-3/2}) \\
&= \sqrt{\frac{2}{\pi}} \left( \frac{1}{ \ell \sin \phi } \right)^{1/2}  \big[ \sin \psi_\ell^+ - \frac{ 1 }{8 \ell \phi} \cos \psi_\ell^+ \big] +  O(\ell^{-5/2} \phi^{-5/2}) +  O( \ell^{-3/2} \phi^{-1/2}),
\end{align*}
so that \eqref{P}, is established.
Let us introduce some more notation:
\begin{equation*}
s_1(\ell,\phi)= \sqrt{ \frac{2}{\pi } } \frac{\ell^{-1/2}}{\sin^{1/2} \phi}, \;\;\;\;\; s_2(\ell, \phi)=     \sqrt{\frac{2}{\pi}} \frac{ 1 }{8 \ell \phi}   \frac{\ell^{-1/2}}{\sin^{1/2} \phi},
\end{equation*}
and
\begin{equation*}
s_3(\ell,\phi)=P_\ell(\cos \phi)- s_1(\ell,\phi) \cos \psi^-_\ell - s_2(\ell,\phi) \cos \psi^+_\ell;
\end{equation*}
note that
\begin{equation*}
s_3(\ell,\phi)= O(\ell^{-5/2} \phi^{-5/2}) + O(\ell^{-3/2} \phi^{-1/2}).
\end{equation*}
We can now rewrite \eqref{P} as follows:
\begin{align*}
P_\ell(\cos \phi)=s_1(\ell,\phi) \cos \psi^-_\ell - s_2(\ell,\phi) \cos \psi^+_\ell+s_3(\ell,\phi).
\end{align*}
To obtain the asymptotic behaviour of the first derivative in \eqref{P'} we first note that
\begin{align*}
P'_\ell(\cos \phi)&= \frac{\ell+1}{\sin^2 \phi} [\cos \phi P_\ell(\cos \phi)-P_{\ell+1}(\cos \phi)].
\end{align*}
where
\begin{align*}
\cos \phi P_\ell(\cos \phi)-P_{\ell+1}(\cos \phi)&= s_1(\ell, \phi) \cos \phi \cos \psi^-_\ell-  s_2(\ell, \phi)  \cos \phi \cos \psi^+_\ell+    s_3(\ell, \phi) \cos \phi  \\
&\;\; - s_1(\ell+1, \phi) \cos \psi^-_{\ell+1} +  s_2(\ell+1, \phi) \cos \psi^+_{\ell+1}- s_3(\ell+1, \phi).
\end{align*}
Now note that
\begin{align}
s_1(\ell+k,\phi)&=s_1(\ell,\phi)+\sqrt{\frac{2}{\pi}} \frac{1}{\sqrt{\sin \phi}}   \sum_{n=1}^{m}\binom{-1/2}{n} k^n \ell^{-n-1/2}+O(\ell^{-(m+1)-1/2} \phi^{-1/2}),\label{infinite-jest-1}\\
s_2(\ell+k,\phi)&=s_2(\ell,\phi)+ \sqrt{\frac{2}{\pi}}  \frac{1}{\sqrt{\sin \phi}} \frac{1}{8 \ell \phi} \sum_{n=1}^{m}\binom{-1/2}{n} k^n \ell^{-n-1/2} +O(\ell^{-(m+1)-3/2} \phi^{-3/2}), \label{infinite-jest-2}
\end{align}
and
\begin{align}
s_3(\ell+k,\phi)&=s_3(\ell,\phi)+O(\ell^{-7/2} \phi^{-5/2})+O(\ell^{-5/2} \phi^{-1/2}); \label{infinite-jest-3}
\end{align}
so that, using \eqref{infinite-jest-1} and \eqref{infinite-jest-2}, with $m=0$, we have
\begin{align*}
\cos \phi P_\ell(\cos \phi)-P_{\ell+1}(\cos \phi)&=   s_1(\ell, \phi) [\cos \phi \cos \psi^-_\ell-\cos \psi^-_{\ell+1}]-  s_2(\ell, \phi) [\cos \phi \cos \psi^+_\ell-\cos \psi^+_{\ell+1}] \\
&\;\;+ s_3(\ell, \phi) [\cos \phi-1] +O(\ell^{-3/2} \phi^{-1/2}).
\end{align*}
Now observe that
\begin{align*}
\cos \phi \cos \psi^{\pm}_\ell-\cos \psi^{\pm}_{\ell+1} &= \sin
\phi \sin\psi^{\pm}_\ell;
\end{align*}
we obtain
\begin{align*}
\cos \phi P_\ell(\cos \phi)-P_{\ell+1}(\cos \phi)
&=   s_1(\ell, \phi) \sin \phi \sin\psi^{-}_\ell  -  s_2(\ell, \phi) \sin \phi \sin\psi^{+}_\ell + s_3(\ell, \phi) [\cos \phi-1]+O(\ell^{-3/2} \phi^{-1/2}).
\end{align*}
Then \eqref{P'} easily follows, since we get
\begin{align*}
P'_\ell(\cos \phi)&=  \sqrt{\frac 2  \pi} \frac{\ell^{1-1/2}}{\sin^{1+1/2} \phi}  [\sin \psi^-_\ell- \frac{1}{8 \ell \phi}\sin \psi^+_\ell  ] +O(  \ell^{-1/2} \phi^{-5/2}).
\end{align*}
To prove the asymptotic behaviour of the second derivative in \eqref{P''} we start from
\begin{align} \label{c2}
P''_\ell(\cos \phi)&=\frac{\ell+1}{\sin^4 \phi} [ (1+2 \cos^2 \phi +\ell \cos^2 \phi) P_\ell(\cos \phi)-(5+2 \ell) \cos \phi  P_{\ell+1}(\cos \phi)+(\ell+2) P_{\ell+2}(\cos \phi)  ].
\end{align}
We first note that
\begin{align*}
& (1+2 \cos^2 \phi +\ell \cos^2 \phi) P_\ell(\cos \phi)-(5+2 \ell) \cos \phi  P_{\ell+1}(\cos \phi)+(\ell+2) P_{\ell+2}(\cos \phi)  \\
&=    s_1(\ell, \phi)(1+2 \cos^2 \phi +\ell \cos^2 \phi)  \cos \psi^-_\ell -  s_2(\ell, \phi) (1+2 \cos^2 \phi +\ell \cos^2 \phi)  \cos \psi^+_\ell+ s_3(\ell, \phi) (1+2 \cos^2 \phi +\ell \cos^2 \phi)  \\
&\;\;- (5+2 \ell)   s_1(\ell+1, \phi)  \cos \phi    \cos \psi^-_{\ell+1}  + (5+2 \ell)  s_2(\ell+1, \phi) \cos \phi   \cos \psi^+_{\ell+1}  - (5+2 \ell)  s_3(\ell+1, \phi) \cos \phi   \\
&\;\; + (\ell+2)  s_1(\ell+2, \phi) \cos \psi^-_{\ell+2}  - (\ell+2)  s_2(\ell+2, \phi) \cos \psi^+_{\ell+2} +(\ell+2)  s_3(\ell+2, \phi),
\end{align*}
now, applying \eqref{infinite-jest-1} and \eqref{infinite-jest-2} with $m=1$, we obtain for example that
\begin{align*}
&-(5+2 \ell)s_1(\ell+1,\phi) \cos \phi    \cos \psi^-_{\ell+1}+ (\ell+2)  s_1(\ell+2, \phi) \cos \psi^-_{\ell+2}\\
&=-(5+2 \ell)s_1(\ell,\phi) \cos \phi    \cos \psi^-_{\ell+1}+ (\ell+2)  s_1(\ell, \phi) \cos \psi^-_{\ell+2}\\
&\;\;-(5+2 \ell) \left[- \sqrt{\frac{2}{\pi}} \frac{1}{\sqrt{\sin \phi}} \frac 1 2 \ell^{-3/2} +O(\ell^{-2-1/2} \phi^{-1/2})  \right]  \cos \phi    \cos \psi^-_{\ell+1}\\
&\;\;+(\ell+2) \left[- \sqrt{\frac{2}{\pi}} \frac{1}{\sqrt{\sin \phi}}  \ell^{-3/2} +O(\ell^{-2-1/2} \phi^{-1/2})  \right]
\cos \psi^-_{\ell+2}
\end{align*}
i.e.,
\begin{align*}
&-(5+2 \ell)s_1(\ell,\phi) \cos \phi    \cos \psi^-_{\ell+1}+ (\ell+2)  s_1(\ell, \phi) \cos \psi^-_{\ell+2}\\
&+ \left(\frac{5}{2}+ \ell \right) \sqrt{\frac{2}{\pi}} \frac{1}{\sqrt{\sin \phi}} \ell^{-3/2} \cos \phi    \cos \psi^-_{\ell+1}  \\
&- (\ell+2)  \sqrt{\frac{2}{\pi}} \frac{1}{\sqrt{\sin \phi}}  \ell^{-3/2}
\cos \psi^-_{\ell+2} +O(\ell^{-3/2} \phi^{-1/2})
\end{align*}
now since $\cos \psi^-_{\ell+2}=\cos \phi \cos \psi^-_{\ell+1}-\sin \phi \cos \psi^-_{\ell+1}$, we have
\begin{align*}
&-(5+2 \ell)s_1(\ell,\phi) \cos \phi    \cos \psi^-_{\ell+1}+ (\ell+2)  s_1(\ell, \phi) \cos \psi^-_{\ell+2}\\
&+ \frac{5}{2}  \sqrt{\frac{2}{\pi}} \frac{1}{\sqrt{\sin \phi}} \ell^{-3/2} \cos \phi    \cos \psi^-_{\ell+1} + \ell   \sqrt{\frac{2}{\pi}} \frac{1}{\sqrt{\sin \phi}}  \ell^{-3/2} \sin \phi \sin \psi^-_{\ell+1}  \\
&- 2  \sqrt{\frac{2}{\pi}} \frac{1}{\sqrt{\sin \phi}}  \ell^{-3/2}
\cos \psi^-_{\ell+2} +O(\ell^{-3/2} \phi^{-1/2}) \\
&=-(5+2 \ell)s_1(\ell,\phi) \cos \phi    \cos \psi^-_{\ell+1}+ (\ell+2)  s_1(\ell, \phi) \cos \psi^-_{\ell+2}\\
&+O(\ell^{-1/2} \phi^{1/2}).
\end{align*}
In the same way we deal with $(5+2 \ell)s_2(\ell+1,\phi) \cos \phi
\cos \psi^+_{\ell+1}- (\ell+2)  s_2(\ell+2, \phi) \cos
\psi^+_{\ell+2}$ and $-(5+2 \ell)s_3(\ell+1,\phi) \cos \phi   +
(\ell+2)  s_3(\ell+2, \phi)$, so that we arrive at
\begin{align*}
& (1+2 \cos^2 \phi +\ell \cos^2 \phi) P_\ell(\cos \phi)-(5+2 \ell) \cos \phi  P_{\ell+1}(\cos \phi)+(\ell+2) P_{\ell+2}(\cos \phi)  \\
&=  \ell \big\{  s_1(\ell, \phi)  [\cos^2 \phi  \cos \psi^-_\ell  -2 \cos \phi    \cos \psi^-_{\ell+1} + \cos \psi^-_{\ell+2} ]    -  s_2(\ell, \phi) [ \cos^2 \phi  \cos \psi^+_\ell  - 2  \cos \phi   \cos \psi^+_{\ell+1} +\cos \psi^+_{\ell+2}  ]     \\
&\;\;+  s_3(\ell, \phi) [ \cos^2 \phi - 2 \cos \phi +1 ]  \big\}  +     s_1(\ell, \phi)[  (1+2 \cos^2 \phi )  \cos \psi^-_\ell  - 5   \cos \phi    \cos \psi^-_{\ell+1} + 2  \cos \psi^-_{\ell+2}  ] \\
&\;\;- s_2(\ell, \phi)[ (1+2 \cos^2 \phi )  \cos \psi^+_\ell - 5  \cos \phi   \cos \psi^+_{\ell+1} + 2  \cos \psi^+_{\ell+2}   ]    + s_3(\ell, \phi) [1+2 \cos^2 \phi  - 5  \cos \phi +    2]\\
&\;\;+O(\ell^{-1/2} \phi^{1/2})  .
\end{align*}
Since
\begin{align*}
&\cos^2 \phi  \cos \psi^{\pm}_\ell  -2 \cos \phi    \cos \psi^{\pm}_{\ell+1} + \cos \psi^{\pm}_{\ell+2} = -\sin^2 \phi \cos \psi^{\pm}_\ell, \\
  &(1+2 \cos^2 \phi )  \cos \psi^{\pm}_\ell  - 5   \cos \phi    \cos \psi^{\pm}_{\ell+1} + 2  \cos \psi^{\pm}_{\ell+2}  =\pm \sin \phi \cos \psi^{\mp}_{\ell-1},
\end{align*}
we obtain
\begin{align}
& (1+2 \cos^2 \phi +\ell \cos^2 \phi) P_\ell(\cos \phi)-(5+2 \ell) \cos \phi  P_{\ell+1}(\cos \phi)+(\ell+2) P_{\ell+2}(\cos \phi) \nonumber  \\
&=  \ell \big\{ s_1(\ell, \phi)  [ -\sin^2 \phi \cos \psi^-_\ell ]   -  s_2(\ell, \phi) [  -\sin^2 \phi \cos \psi^+_\ell ]     +  s_3(\ell, \phi) [ \cos^2 \phi - 2 \cos \phi +1 ]  \big\}\nonumber \\
&\;\; + s_1(\ell, \phi)[ -\sin \phi \cos \psi^+_{\ell-1}  ]  - s_2(\ell, \phi)[ \sin \phi \cos \psi^-_{\ell-1}]       + s_3(\ell, \phi) [1+2 \cos^2 \phi  - 5  \cos \phi +    2] \nonumber \\
&\;\;+O(\ell^{-1/2} \phi^{1/2})  .    \label{c1}
\end{align}
Plugging  \eqref{c1} in \eqref{c2}, we finally have
\begin{align*}
P''_\ell(\cos \phi)
&= \frac{(\ell+1) \ell}{\sin^2 \phi} \{   - s_1(\ell, \phi)  \cos \psi^-_\ell   +  s_2(\ell, \phi)  \cos \psi^+_\ell      \} + \frac{\ell+1}{\sin^3 \phi} \{    - s_1(\ell, \phi) \cos \psi^+_{\ell-1}    - s_2(\ell, \phi) \cos \psi^-_{\ell-1}       \}\\
&\;\;+ \frac{(\ell+1) \ell}{\sin^4 \phi}    s_3(\ell, \phi) [ \cos^2 \phi - 2 \cos \phi +1 ]  + \frac{\ell+1}{\sin^4 \phi}  s_3(\ell, \phi) [1+2 \cos^2 \phi  - 5  \cos \phi +    2]   +O(\ell^{1/2} \phi^{-7/2})  \\
&= \sqrt{\frac 2 \pi } \frac{ \ell^{2-1/2}}{\sin^{2+1/2} \phi} \{   -   \cos \psi^-_\ell   + \frac{1}{8 \ell \phi}  \cos \psi^+_\ell      \} +\sqrt{\frac 2 \pi }  \frac{\ell^{1-1/2}}{\sin^{3+1/2} \phi} \{    - \cos \psi^+_{\ell-1}    -\frac{1}{8 \ell \phi} \cos \psi^-_{\ell-1}       \}\\
&\;\;+ \ell^2   s_3(\ell, \phi)   + \frac{\ell}{\sin^2 \phi}  s_3(\ell, \phi) +O(\ell^{1/2} \phi^{-7/2}),
\end{align*}
where
\begin{align*}
 \ell^2   s_3(\ell, \phi)   + \frac{\ell}{\sin^2 \phi}  s_3(\ell, \phi)&=\ell^2  [O(\ell^{-5/2} \phi^{-5/2}) + O(\ell^{-3/2} \phi^{-1/2})]   + \frac{\ell}{\sin^2 \phi} [O(\ell^{-5/2} \phi^{-5/2}) + O(\ell^{-3/2} \phi^{-1/2})]\\
&= O(\ell^{1/2}\phi^{-1/2})  +O(\ell^{-3/2} \phi^{-9/2}).
\end{align*}
To obtain the asymptotic behaviour of the third derivative in \eqref{P'''}, we write
\begin{align*}
P'''_\ell(x)&=\frac{\ell+1}{(x^2-1)^3} [(-\ell^2 x^3-\ell (3 x + 5 x^3) -9x-6 x^3) P_\ell(x)+(3 \ell^2 x^2+ \ell(3+18 x^2)+6+27 x^2) P_{\ell+1}(x)\\
&\;\;-(3 \ell^2 x+18 \ell x+24 x)P_{\ell+2}(x)+(\ell^2+5 \ell+6) P_{\ell+3}(x)]\\
&=\frac{\ell^2(\ell+1)}{(x^2-1)^3} [- x^3 P_\ell(x)+3  x^2 P_{\ell+1}(x)-3  x P_{\ell+2}(x)+ P_{\ell+3}(x)]\\
&\;\;+\frac{\ell (\ell+1)}{(x^2-1)^3} [- (3 x + 5 x^3) P_\ell(x)+ (3+18 x^2) P_{\ell+1}(x)- 18 xP_{\ell+2}(x)+5  P_{\ell+3}(x)]\\
&\;\;+\frac{\ell+1}{(x^2-1)^3} [-(9x+6 x^3) P_\ell(x)+(6+27 x^2) P_{\ell+1}(x)- 24 xP_{\ell+2}(x)+6 P_{\ell+3}(x)].
\end{align*}
Now, writing Legendre polynomials in terms of $s_1, s_2$ and $s_3$
and applying \eqref{infinite-jest-1} and \eqref{infinite-jest-2}
with $m=2$, we have, for example, that
\begin{align*}
&3 s_1(\ell+1,\phi) \cos^2 \phi \cos \psi^-_{\ell+1}-3 s_1(\ell+2,\phi) \cos \phi \cos \psi^-_{\ell+2}+s_1(\ell+3,\phi) \cos \psi^-_{\ell+3}\\
&=3 s_1(\ell,\phi) \cos^2 \phi \cos \psi^-_{\ell+1}-3 s_1(\ell,\phi) \cos \phi \cos \psi^-_{\ell+2}+s_1(\ell,\phi) \cos \psi^-_{\ell+3}\\
&+ \sqrt{\frac{2}{\pi}} \frac{1}{\sqrt{\sin \phi}} \Big\{  3 \cos^2 \phi \cos \psi^-_{\ell+1} \Big[ -\frac 1 2 \ell^{-3/2} +\frac 3 2 \ell^{-5/2}+O(\ell^{-7/2}) \Big] \\
&- 3 \cos \phi \cos \psi^-_{\ell+2} \Big[  -\ell^{-3/2} +\frac 3 2 2^2 \ell^{- 5/2}+O(\ell^{-7/2}) \Big]\\
&+ \cos \psi^-_{\ell+3} \Big[  -\frac 3 2 \ell^{-3/2} + \frac 3 2 3^2 \ell^{-5/2}+O(\ell^{-7/2})\Big] \Big\}
\end{align*}
and since we also have $\cos \psi^-_{\ell+3}=\cos^2 \phi \cos
\psi^-_{\ell+1}- \sin^2 \phi \cos \psi^{-}_{\ell+1}-2 \sin \phi
\cos \phi \sin \psi^-_{\ell+1} $, we have

\begin{align*}
&3 s_1(\ell,\phi) \cos^2 \phi \cos \psi^-_{\ell+1}-3 s_1(\ell,\phi) \cos \phi \cos \psi^-_{\ell+2}+s_1(\ell,\phi) \cos \psi^-_{\ell+3}\\
&+ \sqrt{\frac{2}{\pi}} \frac{1}{\sqrt{\sin \phi}} \Big\{  O(\ell^{-7/2})\\
&+3 \sin \phi \cos \phi \sin \psi^-_{\ell+1} \Big[  \frac 3 2 2^2 \ell^{- 5/2}+O(\ell^{-7/2}) \Big]\\
&-\sin^2 \phi \cos \psi^-_{\ell+1} \Big[  -\frac 3 2 \ell^{-3/2} + \frac 3 2 3^2 \ell^{-5/2}+O(\ell^{-7/2})\Big]\\
&-2 \sin \phi \cos \phi \sin \psi^-_{\ell+1} \Big[\frac 3 2 3^2 \ell^{-5/2}+O(\ell^{-7/2})  \Big]  \Big\}\\
&=3 s_1(\ell,\phi) \cos^2 \phi \cos \psi^-_{\ell+1}-3 s_1(\ell,\phi) \cos \phi \cos \psi^-_{\ell+2}+s_1(\ell,\phi) \cos \psi^-_{\ell+3}\\
&+O(\ell^{-5/2} \phi^{1/2}).
\end{align*}
Exploiting the identities
\begin{align*}
&-\cos^3 \phi \cos \psi^{\pm}_{\ell}+3 \cos^2 \phi \cos \psi^{\pm}_{\ell+1}-3 \cos \phi \cos \psi^{\pm}_{\ell+2}+\cos \psi^{\pm}_{\ell+3}= {\pm} \sin^3 \phi \cos \psi_\ell^{\mp},\\
&-(3 \cos \phi+5 \cos^3 \phi) \cos \psi^{\pm}_{\ell}+(3+18 \cos^2 \phi) \cos \psi^{\pm}_{\ell+1}-18
\cos \phi \cos \psi^{\pm}_{\ell+2}+5 \cos \psi^{\pm}_{\ell+3}=\frac 1 2 \sin^2 \phi [\cos \psi^{\pm}_{\ell+1}+5 \cos \psi^{\pm}_{\ell-1}],\\
&(-9 \cos \phi-6 \cos^3 \phi) \cos \psi^{\pm}_\ell+(6+27 \cos^2 \phi) \cos \psi^{\pm}_{\ell+1}-24 \cos \phi
\cos \psi^{\pm}_{\ell+2}+6 \cos \psi^{\pm}_{\ell+3}=-3 \sin \phi \cos \phi \sin \psi^{\pm}_{\ell-1},
\end{align*}
we get
\begin{align*}
&\hspace{6cm}P'''_\ell(\cos \phi)\\
&\;\;=-\frac{\ell^2(\ell+1)}{\sin^6 \phi} \big\{  - s_1(\ell, \phi)  \sin^3 \phi \cos \psi_{\ell}^+ -  s_2(\ell, \phi)  \sin^3 \phi \cos \psi_\ell^-  +  s_3(\ell, \phi) [-\cos^3 \phi+3 \cos^2 \phi-3 \cos \phi+1] \big\}\\
&\;\;-\frac{\ell (\ell+1)}{\sin^6 \phi} \big\{   s_1(\ell, \phi) \frac 1 2 \sin^2 \phi [\cos \psi^{-}_{\ell+1}+5 \cos \psi^{-}_{\ell-1}]  - s_2(\ell, \phi) \frac 1 2 \sin^2 \phi [\cos \psi^{+}_{\ell+1}+5 \cos \psi^{+}_{\ell-1}]         \\
&\;\; +  s_3(\ell, \phi) [-(3 \cos \phi+5 \cos^3 \phi)+3+18 \cos^2 \phi-18 \cos \phi+5] \big\}\\
&\;\;-\frac{\ell+1}{\sin^6 \phi} \big\{  - s_1(\ell, \phi) 3 \sin \phi \cos \phi \sin \psi^{-}_{\ell-1} + s_2(\ell, \phi) 3 \sin \phi \cos \phi \sin \psi^{+}_{\ell-1}  \\
&\;\;+  s_3(\ell, \phi) [-9 \cos \phi-6 \cos^3 \phi+6+27 \cos^2 \phi-24 \cos \phi+6] \big\}+O(\ell^{1/2} \phi^{-11/2}) \\
&=\sqrt{\frac 2 \pi}\frac{\ell^{3-1/2}}{\sin^{3+1/2} \phi} [    \cos \psi_{\ell}^+ + \frac{1}{8 \ell \phi}    \cos \psi_{\ell}^- ]- \sqrt{\frac 2 \pi} \frac{\ell^{2-1/2} }{\sin^{4+1/2} \phi} [  \frac 1 2  (\cos \psi^{-}_{\ell+1}+5 \cos \psi^{-}_{\ell-1})  -\frac{1}{8 \ell \phi} \frac 1 2  (\cos \psi^{+}_{\ell+1}+5 \cos \psi^{+}_{\ell-1})     ]\\
&\;\;+\sqrt{\frac 2 \pi} \frac{\ell^{1-1/2}}{\sin^{5+1/2} \phi} [  3  \cos \phi \sin \psi^{-}_{\ell-1} - \frac{1}{8 \ell \phi} 3  \cos \phi \sin \psi^{+}_{\ell-1} ] + \ell^3  s_3(\ell, \phi)-\frac{\ell^2}{\sin^2 \phi}   s_3(\ell, \phi)-\frac{\ell}{\sin^4 \phi}  s_3(\ell, \phi) +O(\ell^{1/2} \phi^{-11/2}) .
\end{align*}
Finally, to prove \eqref{P''''}, we start from
\begin{align*}
&\hspace{6cm}P''''_\ell(x)\\
&=\frac{\ell+1}{(x^2-1)^4} [(\ell^3 x^4+9 \ell^2 x^4+6 \ell^2 x^2+26 \ell x^4+42 \ell x^2+3 \ell+24 x^4+72 x^2+9) P_\ell(x)\\
& \;\;+(-4 \ell^3 x^3-42 \ell^2 x^3-12 \ell^2 x-146 \ell x^3-78 \ell x-168x^3-111 x) P_{\ell+1}(x)\\
&\;\;+(6 \ell^3 x^2+66 \ell^2 x^2+6 \ell^2+231 \ell x^2+30 \ell+246 x^2+36)P_{\ell+2}(x)\\
&\;\;+(-4 \ell^3 x-42 \ell^2 x-134\ell x-132 x)P_{\ell+3}(x)\\
&\;\;+(\ell^3+9\ell^2+26 \ell+24)P_{\ell+4}(x)] \\
&=\frac{\ell^3(\ell+1)}{(x^2-1)^4} [ x^4 P_\ell(x)-4  x^3 P_{\ell+1}(x)+6  x^2 P_{\ell+2}(x)-4  x P_{\ell+3}(x)+P_{\ell+4}(x)] \\
&+\frac{\ell^2(\ell+1)}{(x^2-1)^4} [(9  x^4+6  x^2) P_\ell(x)-(42  x^3+12  x) P_{\ell+1}(x)+(66  x^2+6 )P_{\ell+2}(x)-42  x P_{\ell+3}(x)+9 P_{\ell+4}(x)] \\
&+\frac{\ell(\ell+1)}{(x^2-1)^4} [(26  x^4+42  x^2+3) P_\ell(x)-(146  x^3+78  x) P_{\ell+1}(x)+(231 x^2+30 )P_{\ell+2}(x)-134  xP_{\ell+3}(x)+26 P_{\ell+4}(x)] \\
&+\frac{\ell+1}{(x^2-1)^4} [(24 x^4+72 x^2+9) P_\ell(x)-(168x^3+111 x) P_{\ell+1}(x)+(246 x^2+36)P_{\ell+2}(x)-132 x P_{\ell+3}(x)+24 P_{\ell+4}(x)].
\end{align*}
Let us recall the further identities
\begin{align*}
&\cos^4 \phi \cos \psi^{\pm}_\ell-4 \cos^3 \phi \cos \psi^{\pm}_{\ell+1}+6 \cos^2 \phi \cos \psi^{\pm}_{\ell+2}-4 \cos \phi \cos \psi^{\pm}_{\ell+3}+\cos \psi^{\pm}_{\ell+4}=\sin^4 \phi \cos \psi^{\pm}_\ell, \\
&\cos^4 \phi -4 \cos^3 \phi+6 \cos^2 \phi -4 \cos \phi +1=(\cos \phi-1)^4,
\end{align*}
\begin{align*}
&(6 \cos^2\phi+9 \cos^4 \phi) \cos \psi^{\pm}_{\ell}-(12 \cos \phi+42\cos^3\phi) \cos \psi^{\pm}_{\ell+1}+(6+66 \cos^2 \phi) \cos \psi^{\pm}_{\ell+2} -42 \cos \phi \cos \psi^{\pm}_{\ell+3} \\
&\;\;+9 \cos\psi^{\pm}_{\ell+4}=- \frac 3 2 \sin^3 \phi   [\sin \psi^{\pm}_{\ell+1} +3 \sin \psi^{\pm}_{\ell-1} ],\\
&(6 \cos^2\phi+9 \cos^4 \phi) -(12 \cos \phi+42\cos^3\phi) +(6+66 \cos^2 \phi)  -42 \cos \phi  +9 =3 (\cos \phi-1)^3 (3 \cos \phi-5),
\end{align*}
\begin{align*}
&(3+42 \cos^2 \phi+26 \cos^4\phi) \cos \psi^{\pm}_{\ell}+(-78 \cos \phi-146 \cos^3 \phi) \cos \psi^{\pm}_{\ell+1}+(30+231 \cos^2\phi) \cos \psi^{\pm}_{\ell+2}\\
&\;\;-134 \cos \phi \cos \psi^{\pm}_{\ell+3}+26 \cos \psi^{\pm}_{\ell+4}=-\frac 1 2 \sin^2 \phi  [\cos \psi^{\pm}_{\ell+2}+16 \cos \psi^{\pm}_{\ell} +13 \cos \psi^{\pm}_{\ell-2} ],\\
&(3+42 \cos^2 \phi+26 \cos^4\phi)+(-78 \cos \phi-146 \cos^3 \phi) +(30+231 \cos^2\phi)\\
&\;\; -134 \cos \phi +26 =(\cos \phi-1)^2 (59-94 \cos \phi+26 \cos^2 \phi),
\end{align*}
 \begin{align*}
& (9+72 \cos^2 \phi+24 \cos^4 \phi) \cos \psi^{\pm}_{\ell}+(-111 \cos \phi-168 \cos^3 \phi) \cos \psi^{\pm}_{\ell+1}+(36+246 \cos^2\phi) \cos\psi^{\pm}_{\ell+2}\\
  &\;\; -132 \cos \phi \cos \psi^{\pm}_{\ell+3}+24 \cos \psi^{\pm}_{\ell+4}={\pm}3  \sin \phi  (5-4\sin^2 \phi) \cos \psi^{\mp}_{\ell-1},
   \end{align*}
from which we obtain
\begin{align*}
&\hspace{6cm}P''''_\ell(\cos \phi)\\
&=\frac{\ell^3(\ell+1)}{\sin^8 \phi} \{ s_1(\ell, \phi) \sin^4 \phi \cos \psi^{-}_\ell - s_2(\ell, \phi) \sin^4 \phi \cos \psi^{+}_\ell+ s_3(\ell, \phi) (\cos \phi-1)^4\}\\
&\;\;+\frac{\ell^2(\ell+1)}{\sin^8 \phi} \{  - s_1(\ell, \phi)  \frac 3 2 \sin^3 \phi   [\sin \psi^{-}_{\ell+1} +3 \sin \psi^{-}_{\ell-1} ]+ s_2(\ell, \phi)  \frac 3 2 \sin^3 \phi   [\sin \psi^{+}_{\ell+1} +3 \sin \psi^{+}_{\ell-1} ]+ s_3(\ell, \phi) 3 (\cos \phi-1)^3 (3 \cos \phi-5)\} \\
&\;\;+\frac{\ell(\ell+1)}{\sin^8 \phi} \{ - s_1(\ell, \phi) \frac 1 2 \sin^2 \phi  [\cos \psi^{-}_{\ell+2}+16 \cos \psi^{-}_{\ell} +13 \cos \psi^{-}_{\ell-2} ]+ s_2(\ell, \phi) \frac 1 2 \sin^2 \phi  [\cos \psi^{+}_{\ell+2}+16 \cos \psi^{+}_{\ell} +13 \cos \psi^{+}_{\ell-2} ]\\
&\;\;\;\;\;\;+ s_3(\ell, \phi) (\cos \phi-1)^2 (59-94 \cos \phi+26 \cos^2 \phi)\} \\
&+\frac{\ell+1}{\sin^8 \phi} \{  - s_1(\ell, \phi) 3  \sin \phi  (5-4\sin^2 \phi) \cos \psi^+_{\ell-1}- s_2(\ell, \phi) 3  \sin \phi  (5-4\sin^2 \phi) \cos \psi^-_{\ell-1}\\
&\;\;\;\;\;\;+ s_3(\ell, \phi) 3 [79-123 \cos \phi+57 \cos(2 \phi) -14 \cos(3 \phi)+\cos(4 \phi)] \}  +O(\ell^{1/2} \phi^{-15/2})\\
&=\frac{\ell^3(\ell+1)}{\sin^4 \phi} \{ s_1(\ell, \phi)  \cos \psi^{-}_\ell - s_2(\ell, \phi)  \cos \psi^{+}_\ell \} +\ell^4   s_3(\ell, \phi)  \\
&\;\;+\frac{\ell^2(\ell+1)}{\sin^5 \phi} \{  - s_1(\ell, \phi)  \frac 3 2    [\sin \psi^{-}_{\ell+1} +3 \sin \psi^{-}_{\ell-1} ]+ s_2(\ell, \phi)  \frac 3 2   [\sin \psi^{+}_{\ell+1} +3 \sin \psi^{+}_{\ell-1} ] \} +\frac{\ell^3}{\sin^2 \phi}  s_3(\ell, \phi) 3  (3 \cos \phi-5) \\
&\;\;+\frac{\ell(\ell+1)}{\sin^6 \phi} \{ - s_1(\ell, \phi) \frac 1 2   [\cos \psi^{-}_{\ell+2}+16 \cos \psi^{-}_{\ell} +13 \cos \psi^{-}_{\ell-2} ]+ s_2(\ell, \phi) \frac 1 2   [\cos \psi^{+}_{\ell+2}+16 \cos \psi^{+}_{\ell} +13 \cos \psi^{+}_{\ell-2} ]\} \\
&\;\;\;\;\;\;+\frac{\ell^2}{\sin^4 \phi}   s_3(\ell, \phi)  (59-94 \cos \phi+26 \cos^2 \phi) \\
&\;\;+\frac{\ell+1}{\sin^7 \phi} \{ - s_1(\ell, \phi) 3    (5-4\sin^2 \phi) \cos \psi^+_{\ell-1}- s_2(\ell, \phi) 3    (5-4\sin^2 \phi) \cos \psi^-_{\ell-1}\}+\frac{\ell}{\sin^6 \phi}  s_3(\ell, \phi)+O(\ell^{1/2} \phi^{-15/2})\\
&=\sqrt{\frac 2 \pi} \frac{\ell^{4-1/2}}{\sin^{4+1/2} \phi} \left[  \cos \psi^{-}_\ell -\frac{1}{8 \ell \phi}  \cos \psi^{+}_\ell\right] +\sqrt{\frac 2 \pi}\frac{\ell^{3-1/2}}{\sin^{5+1/2} \phi} \left[- \frac 3 2    \left[\sin \psi^{-}_{\ell+1} +3 \sin \psi^{-}_{\ell-1} \right]+\frac{1}{8 \ell \phi}   \frac 3 2   \left[\sin \psi^{+}_{\ell+1} +3 \sin \psi^{+}_{\ell-1} \right]\right]   \\
&\;\;+\sqrt{\frac 2 \pi}\frac{\ell^{2-1/2}}{\sin^{6+1/2} \phi} \left[- \frac 1 2   [\cos \psi^{-}_{\ell+2}+16 \cos \psi^{-}_{\ell} +13 \cos \psi^{-}_{\ell-2} ]+\frac{1}{8 \ell \phi}  \frac 1 2   \left[\cos \psi^{+}_{\ell+2}+16 \cos \psi^{+}_{\ell} +13 \cos \psi^{+}_{\ell-2} \right]\right]\\
&\;\;+\sqrt{\frac 2 \pi}\frac{\ell^{1-1/2}}{\sin^{7+1/2} \phi} \left[- 3    (5-4\sin^2 \phi) \cos \psi^+_{\ell-1}-\frac{1}{8 \ell \phi}  3    (5-4\sin^2 \phi) \cos \psi^-_{\ell-1}\right]\\
&\;\;+\ell^4   s_3(\ell, \phi)+\frac{\ell^3}{\sin^2 \phi}  s_3(\ell, \phi)   +\frac{\ell^2}{\sin^4 \phi}   s_3(\ell, \phi)   +\frac{\ell}{\sin^6 \phi}  s_3(\ell, \phi)+O(\ell^{1/2} \phi^{-15/2}).
\end{align*}
\end{proof}

\noindent In what follows we write $f_{\ell}(\phi) \simeq
g_{\ell}(\phi)$, if there exists a constant $c_0>0$ such that
\begin{equation*}
-c_0 \; g_{\ell}(\phi) \le f_{\ell}(\phi) \le c_0 \;
g_{\ell}(\phi),
\end{equation*}
for all $\ell \ge 1$ and $\phi \in [C/\ell, \pi/2]$. From Lemma \ref{hilbs}
it follows immediately that:

\begin{lemma}
\label{hilbs_apx} Uniformly in $\ell \ge 1$ and $\phi \in [C/\ell,
\pi/2]$, we have
\begin{align*}
&P^{\prime}_\ell(\cos \phi) \simeq \frac{\ell^{1-1/2}}{ \sin^{1+1/2} \phi}
+R_1(\ell,\phi), & P^{\prime\prime}_\ell(\cos \phi) \simeq \sum_{i=0}^1 \frac{
\ell^{2-i-1/2}}{ \sin^{2+i+1/2} \phi} +R_2(\ell,\phi), \\
&P^{\prime\prime\prime}_\ell(\cos \phi) \simeq \sum_{i=0}^2 \frac{
\ell^{3-i-1/2}}{ \sin^{3+i+1/2} \phi} +R_3(\ell,\phi), & P^{\prime\prime%
\prime\prime}_\ell(\cos \phi) \simeq \sum_{i=0}^3 \frac{ \ell^{4-i-1/2}}{%
\sin^{4+i+1/2} \phi}+R_4(\ell,\phi).
\end{align*}
\end{lemma}

\begin{proof}
From \eqref{P'} and since $\phi \in [C/\ell, \pi/2]$, we obtain
\begin{align*}
P^{\prime}_\ell(\cos \phi)&\le  \sqrt{\frac 2 \pi} \frac{\ell^{1-1/2}}{%
\sin^{1+1/2} \phi} \big[ 1+ \frac{1}{8 \ell \phi}  \big] +R_1(\ell,\phi) \le \sqrt{\frac 2 \pi} \frac{\ell^{1-1/2}}{%
\sin^{1+1/2} \phi} \big[ 1+ \frac{1}{8 C}  \big] +R_1(\ell,\phi),
\end{align*}
and
\begin{align*}
P^{\prime}_\ell(\cos \phi)&\ge \sqrt{\frac 2 \pi} \frac{\ell^{1-1/2}}{%
\sin^{1+1/2} \phi} \big[ -1 - \frac{1}{8 \ell \phi}  \big] +R_1(\ell,\phi) \ge - \sqrt{\frac 2 \pi} \frac{\ell^{1-1/2}}{%
\sin^{1+1/2} \phi} \big[ 1+ \frac{1}{8 C}  \big] +R_1(\ell,\phi).
\end{align*}
The proof is analogous in the other cases.
\end{proof}

From Lemma \ref{hilbs_apx} we obtain the following asymptotics for the
elements of the covariance matrix $\Sigma$.

\begin{lemma}
\label{alphaalpha} For every $\ell \ge 1$ and $\phi \in [C/\ell, \pi/2]$, we
have
\begin{align*}
\frac{\alpha_{1,\ell}(\phi)}{\ell^2}&\simeq \frac{1}{ \ell^{1+1/2}
\sin^{1+1/2} \phi}  + O(\ell^{-5/2} \phi^{-5/2}) \\
\frac{\alpha_{2,\ell}(\phi)}{\ell^2}&\simeq \sum_{i=0}^1 \frac{ 1}{
\ell^{i+1/2} \sin^{i+1/2} \phi} + O(\ell^{-3/2}\phi^{-3/2})  \\
\frac{\beta_{1,\ell}(\phi)}{\ell^3}&\simeq \sum_{i=0}^1 \frac{ 1}{
\ell^{1+i+1/2} \sin^{1+i+1/2} \phi} +O(\ell^{-5/2}\phi^{-5/2})  \\
\frac{\beta_{2,\ell}(\phi)}{\ell^3}& \simeq \sum_{i=0}^1 \frac{1}{%
\ell^{1+i+1/2} \sin^{1+i +1/2} \phi}+ \frac{1}{\ell^{2+1/2} \sin^{1/2} \phi}%
+ O(\ell^{-5/2} \phi^{-5/2})
\\
\frac{\beta_{3,\ell}(\phi)}{\ell^3} & \simeq \sum_{i=0}^2 \frac{1}{%
\ell^{i+1/2} \sin^{i+1/2} \phi} + \frac{1}{\ell^{2+1/2} \sin^{1/2} \phi} +
O(\ell^{-5/2} \phi^{-5/2}) \\
\frac{\gamma_{1,\ell}(\phi)}{\ell^4} & \simeq \sum_{i=0}^1 \frac{1}{%
\ell^{2+i+1/2} \sin^{2+i+1/2} \phi}+ \frac{1}{\ell^{3+1/2} \sin^{1+1/2} \phi}%
+O(\ell^{-7/2} \phi^{-7/2})
\\
\frac{\gamma_{2,\ell}(\phi)}{\ell^4} &\simeq \sum_{i=0}^1 \frac{1}{%
\ell^{1+i+1/2} \sin^{1+i +1/2} \phi}+ \frac{1}{\ell^{2+1/2} \sin^{1/2} \phi}
+O(\ell^{-7/2} \phi^{-7/2}) \\
\frac{\gamma_{3,\ell}(\phi)}{\ell^4} &\simeq \sum_{i=0}^2 \frac{1}{%
\ell^{1+i+1/2} \sin^{1+i+1/2} \phi} + \sum_{i=0}^1 \frac{1}{\ell^{2+i+1/2}
\sin^{i+1/2} \phi} + O(\ell^{-7/2} \phi^{-7/2}) \\
\frac{\gamma_{4,\ell}(\phi)}{\ell^4} &\simeq \sum_{i=0}^3 \frac{1}{%
\ell^{i+1/2} \sin^{i+1/2} \phi}+ \sum_{i=0}^1 \frac{1}{\ell^{2+i+1/2}
\sin^{i+1/2} \phi} + O(\ell^{-7/2} \phi^{-7/2}).
\end{align*}
\end{lemma}

\begin{lemma}
\label{sine2} For $k=0,1,2,\dots $ and $n=1,2,3, \dots$, we have
\begin{align*}
\frac{1}{\ell^k} \int_{C/\ell}^{\pi/2} \frac{1}{\sin^n \phi} \sin \phi d
\phi=
\begin{cases}
O(\ell^{-k}) & \mathrm{for\;} n=1, \\
O(\ell^{-k} \log \ell) & \mathrm{for\;} n=2, \\
O(\ell^{n-k-2}) & \mathrm{for\;} n\ge 3.%
\end{cases}%
\end{align*}
\end{lemma}

\begin{lemma}
\label{sine5} For $n=0,1,2, \dots$, we have
\begin{align*}
\frac{1}{\ell^{n+1/2}} \int_{C/\ell}^{\pi/2} \frac{1}{\sin^{n+1/2} \phi}
\sin \phi d \phi=
\begin{cases}
O(\ell^{-1/2}) & \mathrm{for\;} n=0, \\
O(\ell^{-1-1/2}) & \mathrm{for\;} n=1, \\
O(\ell^{-2}) & \mathrm{for\;} n\ge 2.%
\end{cases}%
\end{align*}
\end{lemma}

\begin{lemma}
\label{sine4} For $k,n=1,2,3, \dots$, we have
\begin{align*}
\frac{1}{\ell^k} \int_{C/\ell}^{\pi/2} \frac{1}{\sin^n \phi} \frac{1}{\phi^k}
\sin \phi d \phi=
\begin{cases}
O(\ell^{-1} \log \ell) & \mathrm{for\;} n+k=2, \\
O(\ell^{n-2}) & \mathrm{for\;} n+k\ge 3.%
\end{cases}%
\end{align*}
\end{lemma}

\begin{proof}
We first note that
\begin{align*}
\frac{1}{\ell^k} \int_{C/\ell}^{\pi/2}  \frac{1}{\phi^{n+k-1}}  d \phi \le \frac{1}{\ell^k} \int_{C/\ell}^{\pi/2} \frac{1}{\sin^n \phi} \frac{1}{\phi^k} \sin \phi d \phi \le \frac{1}{\ell^k} \int_{C/\ell}^{\pi/2} \frac{1}{\sin^{n+k-1} \phi}  d \phi,
\end{align*}
then the conclusion follows from Lemma \ref{sine2} and by observing that
\begin{align*}
\frac{1}{\ell^k} \int_{C/\ell}^{\pi/2}  \frac{1}{\phi^{n+k-1}}  d \phi =
\begin{cases}
O(\ell^{-1} \log \ell) & {\rm for\;} n+k=2,\\
O(\ell^{n-2})& {\rm for\;} n+k\ge 3.
\end{cases}
\end{align*}
\end{proof}

\section{Bounds for the terms $A_{0,\ell}$, $A_{i,\ell}$ and $A_{ij,\ell}$} \label{A_terms}

\begin{proof} [Proof of Lemma \ref{A0}]
By expanding the denominator in $A_{0,\ell}$, we write
\begin{align*}
A_{0,\ell}&=\int_{C/\ell}^{\frac{\pi }{2}}\Big[1+2\frac{\alpha _{2,\ell}^{2}(\phi)}{\lambda _{\ell
}^{2}}+O\Big(\frac{\alpha
_{2,\ell}^{4}(\phi)}{\lambda _{\ell }^{4}}\Big)\Big] \Big[1+2\frac{\alpha _{1,\ell}^{2}(\phi)}{\lambda _{\ell }^{2}}+O\Big(\frac{\alpha _{1,\ell}^{4}(\phi)}{\lambda _{\ell }^{4}}\Big)\Big] \sin \phi d\phi\\
&=\cos \big( \frac{C}{\ell} \big)+\int_{C/ \ell}^{\frac{\pi }{2}}\Big[ 2\frac{\alpha _{2,\ell}^{2}(\phi)}{\lambda
_{\ell }^{2}}+O\Big(\frac{
\alpha _{2,\ell}^{4}(\phi)}{\lambda _{\ell }^{4}}\Big) +2\frac{\alpha _{1,\ell}^{2}(\phi)}{\lambda _{\ell }^{2}}+O\Big(\frac{\alpha _{1,\ell}^{4}(\phi)}{\lambda _{\ell }^{4}}%
\Big)\\
&\;\; +\Big(2\frac{\alpha _{2,\ell}^{2}(\phi)}{\lambda _{\ell }^{2}}+O\Big(\frac{\alpha _{2,\ell}^{4}(\phi)}{\lambda _{\ell }^{4}}%
\Big)\Big) \Big(2\frac{\alpha _{1,\ell}^{2}(\phi)}{\lambda _{\ell }^{2}}+O\Big(\frac{\alpha _{1,\ell}^{4}(\phi)}{\lambda _{\ell }^{4}}\Big)%
\Big)\Big]\sin \phi d\phi.
\end{align*}
The idea is that the leading term in this expansion produces the
cancellation of the component
$\big( \mathbb{E}[{\cal N}^c_I(f_{\ell})] \big)^2$.
 Indeed, in view of Remark \ref{q(0)}, we have
\begin{align*}
2 \lambda _{\ell }^{2} \cos\big( \frac C  \ell \big) \; \iint_{I \times I} q({\bf 0}; t_1,t_2) d t_1 dt_2-
\frac{\lambda_{\ell}^2}{4} \Big[ \int_I p_1^c(t) d t \Big]^2 &=O(\ell^{2}).
\end{align*}
We consider now the rate of the terms
\begin{align*}
&2 \lambda _{\ell }^{2}  \int_{C/\ell}^{\frac{\pi }{2}}\Big[ 2\frac{\alpha _{2,\ell}^{2}(\phi)}{\lambda
_{\ell }^{2}}+2\frac{\alpha _{1,\ell}^{2}(\phi)}{\lambda _{\ell }^{2}}+ 4\frac{\alpha _{2,\ell}^{2}(\phi)}{\lambda _{\ell }^{2}} \frac{\alpha _{1,\ell}^{2}(\phi)}{\lambda _{\ell }^{2}} +O\Big(\frac{\alpha _{1,\ell}^{4}(\phi)}{\lambda _{\ell }^{4}}
\Big) +O\Big(\frac{\alpha _{2,\ell}^{4}(\phi)}{\lambda _{\ell }^{4}}\Big)\Big] \sin \phi \; d\phi\; \iint_{I \times I}q({\bf 0}; t_1,t_2)  d t_1 dt_2.
\end{align*}
We apply here Lemma \ref{alphaalpha} and Lemma \ref{sine2} to
identify the rate of the dominant term. In fact, from Lemma
\ref{alphaalpha}, we obtain the asymptotic behaviour of each
addend of the integrand function, then Lemma \ref{sine2} gives the
asymptotic behaviour of their integrals. We immediately see that
\begin{align} \label{13:02}
\frac{1}{\lambda _{\ell }^{2}}\int_{C/\ell}^{
\frac{\pi }{2}}  \alpha^{2}_{2,\ell }(\phi
) \sin \phi d\phi =O(\ell^{-1}).
\end{align}
To obtain the multiplicative constant of the leading term \eqref{13:02} note that, from Lemma \ref{hilbs} and applying again Lemma \ref{sine2}, to determine the non-dominant terms, we can write
\begin{align} \label{o}
\frac{2}{\lambda _{\ell }^{2}}\int_{C/\ell}^{
\frac{\pi }{2}}\alpha^2_{2,\ell }(\phi) \sin \phi d\phi&=\frac{2}{\lambda _{\ell }^{2}}\int_{C/\ell}^{
\frac{\pi }{2}}  [-\sin ^{2}\phi P_{\ell}^{\prime \prime }(\cos \phi )+\cos \phi P_{\ell }^{\prime }(\cos \phi
)]^{2} \sin \phi d\phi \nonumber \\
&=\frac{2}{\lambda _{\ell }^{2}}\int_{C/\ell}^{
\frac{\pi }{2}}  \Big\{  -\sin^{2}\phi  \sqrt{\frac 2 \pi } \frac{ \ell^{2-1/2}}{%
\sin^{2+1/2} \phi} \left[ - \cos  \psi_\ell^- + \frac{1}{8 \ell \phi} \cos
\psi_\ell^+ \right] \Big\}^{2} \sin \phi d\phi+O(\ell^{-2} \log \ell)\nonumber \\
&=\frac 4 \pi \int_{C/\ell}^{\frac{\pi }{2}}    \frac{1}{
\ell \sin \phi} \left[ - \cos \psi_\ell^- + \frac{1}{8 \ell \phi} \cos
\psi_\ell^+ \right]^{2} \sin \phi d\phi+O(\ell^{-2} \log \ell),
\end{align}
where $\psi_{\ell}^{\pm}=(\ell+1/2) \phi \pm \pi/4$. Applying Lemma \ref{sine4} to get the asymptotic behaviour of the integral \eqref{o}, we have
\begin{align*}
\frac{2}{\lambda _{\ell }^{2}}\int_{C/\ell}^{
\frac{\pi }{2}}\alpha^{2}_{2,\ell }(\phi)  \sin \phi d\phi&=\frac 4 \pi \frac{1}{\ell} \int_{C/\ell}^{\frac{\pi }{2}}     \cos^2 \psi_\ell^-  d\phi+O(\ell^{-2} \log \ell)\\
&= \frac 4 \pi \frac{1}{\ell} \int_{C/\ell }^{\frac{\pi }{2}}  \left[\frac 1 2+\frac 1 2 \cos(2 \psi _{\ell }^{-})\right]
d\phi+O(\ell^{-2} \log \ell)\\
&=\frac{2}{\pi} \frac{1}{\ell}\int_{C/\ell }^{\frac{\pi }{2}}   d\phi +\frac{2}{\pi \ell }\int_{C/\ell }^{\frac{\pi }{2}} \cos (2\psi _{\ell }^{-}) d\phi+O(\ell^{-2} \log \ell)\\
&=\frac{2}{\pi } \frac{1}{\ell} \left(\frac{\pi }{2}+\frac{C}{\ell }\right)+\frac{2}{\pi \ell }
\frac{\cos (C(2+1/\ell )+\sin (\ell \pi ))}{1+2\ell }+O(\ell^{-2} \log \ell)\\
&= \ell^{-1} +O(\ell^{-2} \log \ell).
\end{align*}
\end{proof}

\begin{proof}[Proof of Lemma \ref{Ai}]
We start by observing that the terms $A_{i,\ell}$  can be written in the form
\begin{align*}
A_{i,\ell}=\int_{C/\ell}^{\pi/2} \frac{N_{i,\ell}(\phi)}{(1-4 \alpha^2_{2,\ell}(\phi)/ \lambda_{\ell}^2)^{m/2}(1-4 \alpha^2_{1,\ell}(\phi)/ \lambda_\ell^2)^{n/2}} \sin \phi d \phi
\end{align*}
for a suitable function $N_{i,\ell}(\phi)$ and $n,m=1,2,3$. 
 By expanding in power series around the origin the ratio
 $$\frac{1}{(1-4 \alpha^2_{2,\ell}(\phi)/ \lambda_{\ell}^2)^{m/2}(1-4 \alpha^2_{1,\ell}(\phi)/ \lambda_\ell^2)^{n/2}}$$
and with computations analogous to those performed in the proof of Lemma \ref{A0}, it follows that the dominant terms of $A_{i,\ell}$ are all of the form
\begin{align} \label{//}
\int_{C/\ell}^{\pi/2} N_{i,\ell}(\phi) \sin \phi d \phi.
\end{align}
We study now the asymptotic behaviour of \eqref{//}, for $i=1,\dots,8$:
\begin{itemize}
\item To obtain the asymptotic behaviour of $A_{1,\ell}(\phi)$, we note that
\end{itemize}
\begin{align*}
\int_{C/\ell}^{\pi/2} N_{1,\ell}(\phi) \sin \phi d \phi
&=-\frac{16}{\lambda_\ell^2 (\lambda_\ell-2)} \int_{C/\ell}^{ \pi / 2}  \beta_{2,\ell}^2(\phi) \sin \phi\; d \phi.
\end{align*}
Now Lemma \ref{alphaalpha} gives the asymptotic behaviour of the terms of the integrand function and Lemma \ref{sine2} gives the asymptotic behaviour of the integrand of each term, so that we get:
\begin{align*}
\int_{C/\ell}^{\pi/2} \Big(\frac{\beta_{2,\ell}(\phi)}{\ell^3}\Big)^2   \sin \phi d \phi = O \Big( \int_{C/\ell}^{\pi/2}   \frac{1}{\ell^3 \sin^3 \phi}   \sin \phi d \phi \Big)=O(\ell^{-2}) .
\end{align*}

\begin{itemize}
\item For the term $A_{2,\ell}(\phi)$ we note that
\end{itemize}
\begin{align*}
&\int_{C/\ell}^{\pi/2} N_{2,\ell}(\phi) \sin \phi d \phi=-\frac{16}{\lambda_\ell^2 (\lambda_\ell-2)} \int_{C/\ell}^{ \pi / 2} \beta_{1,\ell}^2 (\phi)  \sin \phi\; d \phi
\end{align*}
and, applying again Lemma \ref{alphaalpha} and Lemma \ref{sine2}, we have
\begin{align*}
 \int_{C/\ell}^{\pi/2} \Big(\frac{ \beta_{1,\ell}(\phi)}{\ell^3}\Big)^2     \sin \phi d \phi = O \Big( \int_{C/\ell}^{\pi/2}   \frac{1}{\ell^3 \sin^3 \phi}   \sin \phi d \phi \Big)=O(\ell^{-2}).
\end{align*}

\begin{itemize}
\item The term $A_{3,\ell}(\phi)$ leads to
\end{itemize}
\begin{align*}
\int_{C/\ell}^{ \pi / 2} N_{3,\ell}(\phi) \sin \phi d \phi=-\frac{16}{\lambda_\ell^2 (\lambda_\ell-2)} \int_{C/\ell}^{ \pi / 2} \beta_{3,\ell}^2(\phi)   \sin \phi\; d \phi
\end{align*}
and, applying Lemma \ref{alphaalpha}  and Lemma \ref{sine2},
\begin{align*}
\int_{C/\ell}^{ \pi / 2} \Big(\frac{ \beta_{3,\ell}(\phi)}{\ell^3}\Big)^2     \sin \phi\; d \phi = O\Big( \int_{C/\ell}^{ \pi / 2} \frac{1}{\ell \sin \phi} \sin \phi\; d \phi \Big)=O(\ell^{-1}).
\end{align*}
Since it is a dominant term, we compute now the leading constant of the term $A_{3,\ell}(\phi)$. Recalling the definition of $\beta_{3,\ell}$ and Lemma \ref{hilbs}, we get
\begin{align*}
&-\frac{16}{\lambda_\ell^2 (\lambda_\ell-2)} \int_{C/\ell}^{ \pi / 2} \beta_{3,\ell}^2(\phi)   \sin \phi\; d \phi\\
&=-\frac{16}{\lambda_\ell^2 (\lambda_\ell-2)}   \int_{C/\ell}^{ \pi / 2} \left[-\sin^3 \phi P'''_\ell(\cos \phi)+3 \sin \phi \cos \phi P''_\ell(\cos \phi)+\sin \phi P'_\ell(\cos \phi)\right]^2  \sin \phi\; d \phi\\
&=-\frac{16}{\ell^6}   \int_{C/\ell}^{ \pi / 2} \left[-\sin^3 \phi \sqrt{\frac{2}{\pi}} \frac{\ell^{3-1/2}}{\sin^{3+1/2} \phi }  \cos \psi^+_\ell \right]^2 \sin \phi\; d \phi+O(\ell^{-2} \log \ell),
\end{align*}
where we have also applied Lemma \ref{sine2}, Lemma \ref{sine5} and Lemma \ref{sine4} to identify the leading term. Now computing explicitly the integral, we have
\begin{align*}
-\frac{16}{\lambda_\ell^2 (\lambda_\ell-2)} \int_{C/\ell}^{ \pi / 2} \beta_{3,\ell}^2(\phi)   \sin \phi\; d \phi&=-16 {\frac{2}{\pi}}\frac{1}{\ell}  \int_{C/\ell}^{ \pi / 2}   \cos^2 \psi^+_\ell \; d \phi+O(\ell^{-2} \log \ell)\\
&=-16 {\frac{2}{\pi}} \frac{1}{\ell}  \int_{C/\ell}^{ \pi / 2}   \cos^2 [(\ell+1/2) \phi+\pi/4] \; d \phi+O(\ell^{-2} \log \ell)\\
&=-16 {\frac{2}{\pi}} \frac{1}{\ell}   \frac{\pi}{4}+O(\ell^{-2} \log \ell)\\
&=-8  \ell^{-1}   +O(\ell^{-2} \log \ell).
\end{align*}

\begin{itemize}
\item The term $A_{4,\ell}(\phi)$ leads to
\end{itemize}
\begin{align*}
\int_{C/\ell}^{ \pi / 2} N_{4,\ell}(\phi) \sin \phi d \phi=-\frac{16}{\lambda_\ell^2 (\lambda_\ell-2)} \int_{C/\ell}^{ \pi / 2} \beta_{2,\ell}(\phi) \beta_{3,\ell}(\phi)   \sin \phi\; d \phi
\end{align*}
again, by  Lemma  \ref{alphaalpha} and Lemma \ref{sine2}, we have
\begin{align*}
\int_{C/\ell}^{\pi/2}  \frac{\beta_{2,\ell}(\phi)}{\ell^3} \frac{\beta_{3,\ell}(\phi)}{\ell^3}   \sin \phi d \phi
= O \Big(  \int_{C/\ell}^{\pi/2}   \frac{1}{\ell^2 \sin^2 \phi}   \sin \phi d \phi  \Big)=O(\ell^{-2} \log \ell).
\end{align*}

\begin{itemize}
\item For $A_{5,\ell}(\phi)$ we have
\end{itemize}

\begin{align*}
\int_{C/\ell}^{ \pi / 2}  N_{5,\ell}(\phi) \sin \phi\; d \phi=\frac{8}{\lambda_\ell (\lambda_\ell-2)} \int_{C/\ell}^{ \pi / 2} \gamma_{1,\ell}(\phi)   \sin \phi\; d \phi+ \frac{8 \cdot 4}{\lambda_\ell^3 (\lambda_\ell-2)} \int_{C/\ell}^{ \pi / 2} \alpha_{2,\ell}(\phi) \beta_{2,\ell}^2(\phi)  \sin \phi\; d \phi
\end{align*}
and then, by Lemma \ref{alphaalpha},  and Lemma \ref{sine5},
\begin{align*}
& \int_{C/\ell}^{\pi/2}  \frac{\gamma_{1,\ell}(\phi)}{\ell^{4}}   \sin \phi d \phi = O \left(  \int_{C/\ell}^{\pi/2}   \frac{1}{\ell^{2+1/2} \sin^{2+1/2} \phi}   \sin \phi d \phi \right)=O(\ell^{-2}),\\
& \int_{C/\ell}^{\pi/2}  \Big( \frac{\beta_{2,\ell}(\phi)}{\ell^3}\Big)^{2}
     \frac{\alpha_{2,\ell}(\phi)}{\ell^2}   \sin \phi d \phi = O \Big(  \int_{C/\ell}^{\pi/2}   \frac{1}{\ell^{3+1/2} \sin^{3+1/2} \phi}   \sin \phi d \phi \Big)=O(\ell^{-2}).
\end{align*}

\begin{itemize}
\item The term $A_{6,\ell}(\phi)$ leads to
\end{itemize}
\begin{align*}
\int_{C/\ell}^{ \pi / 2}  N_{6,\ell}(\phi)  \sin \phi\; d \phi =\frac{8}{\lambda_\ell (\lambda_\ell-2)} \int_{C/\ell}^{ \pi / 2} \gamma_{2,\ell}(\phi)   \sin \phi\; d \phi + \frac{8 \cdot 4}{\lambda_\ell^3 (\lambda_\ell-2)} \int_{C/\ell}^{ \pi / 2} \alpha_{1,\ell}(\phi) \beta_{1,\ell}^2(\phi)   \sin \phi\; d \phi
\end{align*}
and applying again Lemma \ref{alphaalpha} and Lemma \ref{sine5}, we have
\begin{align*}
&\int_{C/\ell}^{\pi/2} \frac{\gamma_{2,\ell}(\phi)}{\ell^{4}}   \sin \phi d \phi = O \Big(  \int_{C/\ell}^{\pi/2}   \frac{1}{\ell^{1+1/2} \sin^{1+1/2} \phi}   \sin \phi d \phi \Big)=O(\ell^{-1-1/2}),\\
&\int_{C/\ell}^{\pi/2}  \Big(\frac{\beta_{1,\ell}(\phi)}{\ell^3}\Big)^2  \frac{\alpha_{1,\ell}(\phi)}{\ell^2}  \sin \phi d \phi = O \Big(  \int_{C/\ell}^{\pi/2}   \frac{1}{\ell^{4+1/2} \sin^{4+1/2} \phi}   \sin \phi d \phi \Big)=O(\ell^{-2}).
\end{align*}

\begin{itemize}
\item For $A_{7,\ell}(\phi)$, we get
\end{itemize}

\begin{align*}
\int_{C/\ell}^{ \pi / 2} N_{7,\ell}(\phi) \sin \phi\; d \phi=\frac{8}{\lambda_\ell (\lambda_\ell-2)} \int_{C/\ell}^{ \pi / 2} \gamma_{4,\ell}(\phi)  \sin \phi\; d \phi + \frac{8 \cdot 4}{\lambda_\ell^3 (\lambda_\ell-2)} \int_{C/\ell}^{ \pi / 2} \alpha_{2,\ell}(\phi) \beta_{3,\ell}^2(\phi)  \sin \phi\; d \phi
\end{align*}
and, by Lemma \ref{alphaalpha} and Lemma \ref{sine5},
\begin{align} \label{gamma4}
& \int_{C/\ell}^{\pi/2}  \frac{ \gamma_{4,\ell}(\phi)}{\ell^{4}}  \sin \phi d \phi = O \Big(  \int_{C/\ell}^{\pi/2}   \frac{1}{\ell^{1/2} \sin^{1/2} \phi}   \sin \phi d \phi \Big)=O(\ell^{-1/2}),\\
& \int_{C/\ell}^{\pi/2}   \Big(\frac{ \beta_{3,\ell}(\phi)}{\ell^3}\Big)^2    \frac{\alpha_{2,\ell}(\phi)}{\ell^2}  \sin \phi d \phi = O \Big(  \int_{C/\ell}^{\pi/2}   \frac{1}{\ell^{1+1/2} \sin^{1+1/2} \phi}   \sin \phi d \phi \Big)=O(\ell^{-1-1/2}). \nonumber
\end{align}
Since the leading term in $\gamma_{4,\ell}(\phi)$ is oscillatory, we can get a sharper bound for the term in \eqref{gamma4}, by observing that
\begin{align*}
& \hspace{6cm} \int_{C/\ell}^{\pi/2}  \frac{ \gamma_{4,\ell}(\phi)}{\ell^{4}} \sin \phi d \phi \\
&= \frac{ 1}{\ell^{4}} \int_{C/\ell}^{\pi/2} [\sin ^{4}\phi P_{\ell }^{\prime \prime \prime
\prime }(\cos \phi )-6\sin ^{2}\phi \cos \phi P_{\ell }^{\prime \prime
\prime }(\cos \phi )+(-4\sin ^{2}\phi +3\cos ^{2}\phi )P_{\ell }^{\prime
\prime }(\cos \phi )+\cos \phi P_{\ell }^{\prime }(\cos \phi )]  \sin \phi d \phi \\
 &= \frac{1}{\sqrt \ell } \int_{C/\ell}^{\pi/2}  \frac{1}{ \sin^{1/2} \phi} [\cos \psi^-_\ell  - \frac{1}{8 \ell \phi} \cos  \psi^+_\ell ] \sin \phi d \phi+O(\ell^{-1-1/2})  \\
&\le \frac{1}{\sqrt \ell } \int_{C/\ell}^{\pi/2}  \cos \psi^-_\ell    d \phi
  + \frac{1}{\sqrt \ell } \int_{C/\ell}^{\pi/2}  \frac{1}{8 \ell \phi}  d \phi+O(\ell^{-1-1/2})  \\
  &= \frac{1}{\sqrt \ell } \frac{2}{1+2 \ell} \Big[ \sin \big( \frac{\ell \pi}{2}\big)+ \sin\big( \frac \pi 4 - \frac{C}{2 \ell} -C\big) \Big]+ O(\ell^{-1-1/2})  \\
 &=O(\ell^{-1-1/2}).
 \end{align*}

 \begin{itemize}
\item Finally for $A_{8,\ell}(\phi)$ we have
\end{itemize}

\begin{align*}
\int_{C/\ell}^{ \pi / 2}  N_{8,\ell}(\phi) \sin \phi\; d \phi =\frac{8}{\lambda_\ell (\lambda_\ell-2)} \int_{C/\ell}^{ \pi / 2} \gamma_{3,\ell}(\phi)   \sin \phi\; d \phi + \frac{8 \cdot 4}{\lambda_\ell^3 (\lambda_\ell-2)} \int_{C/\ell}^{ \pi / 2} \alpha_{2,\ell}(\phi) \beta_{2,\ell}(\phi) \beta_{3,\ell}(\phi)  \sin \phi\; d \phi
\end{align*}
where, from Lemma \ref{alphaalpha} and Lemma \ref{sine5}, we have
\begin{align*}
&\int_{C/\ell}^{\pi/2}  \frac{ \gamma_{3,\ell}(\phi)}{\ell^{4}}    \sin \phi d \phi
= O \Big(  \int_{C/\ell}^{\pi/2}   \frac{1}{\ell^{1+1/2} \sin^{1+1/2} \phi}   \sin \phi d \phi \Big)=O(\ell^{-1-1/2}),\\
& \int_{C/\ell}^{\pi/2}  \frac{\beta_{2,\ell}(\phi)}{\ell^{3}}  \frac{\beta_{3,\ell}(\phi) }{\ell^{3}}  \frac{\alpha_{2,\ell}(\phi)}{\ell^2} \sin \phi d \phi = O \Big(  \int_{C/\ell}^{\pi/2}   \frac{1}{\ell^{2+1/2} \sin^{2+1/2} \phi}   \sin \phi d \phi \Big)=O(\ell^{-2}).
\end{align*}
\end{proof}

\begin{proof}[Proof of Lemma \ref{Aij}]
 \noindent We note that the second order terms are all of the form
 \begin{align*}
\int_{C/\ell}^{  \pi / 2}  \frac{a_{i,\ell}(\phi) a_{j,\ell}(\phi) }{ \sqrt{ (1 -4 \alpha_{2,\ell}^2/\lambda_\ell^2) (1 -4 \alpha_{1,\ell}^2/\lambda_\ell^2) }}\sin \phi\; d \phi
\end{align*}
with $i,j=1,2, \dots,8$. Applying Lemma \ref{sine2}, Lemma \ref{sine5} and Lemma \ref{sine4}  we identify, as in the proof of the previous lemma, that the product of two terms $a_{i,\ell}(\phi) a_{j,\ell}(\phi)$ such that at least one is non dominant produces a non dominant term, so that, if $(i,j) \ne (7,7)$, we immediately see that
\begin{align*}
\int_{C/\ell}^{  \pi / 2}  \frac{a_{i,\ell}(\phi) a_{j,\ell}(\phi) }{ \sqrt{ (1 -4 \alpha_{2,\ell}^2/\lambda_\ell^2) (1 -4 \alpha_{1,\ell}^2/\lambda_\ell^2) }}\sin \phi\; d \phi = O \Big(\int_{C/\ell}^{ \pi / 2}  \frac{1}{\ell^2 \sin^2 \phi}  \sin \phi\; d \phi \Big)=O(\ell^{-2} \log \ell).
\end{align*}
Instead, for $(i,j) = (7,7)$, we have the square of the integrand function in \eqref{gamma4}, that in view of Lemma \ref{sine2}, is immediately seen to be dominant, i.e.,
\begin{align*}
\int_{C/\ell}^{  \pi / 2}  \frac{a^2_{7,\ell}(\phi) }{ \sqrt{ (1 -4 \alpha_{2,\ell}^2/\lambda_\ell^2) (1 -4 \alpha_{1,\ell}^2/\lambda_\ell^2) }} \sin \phi\; d \phi=O\Big(\int_{C/\ell}^{  \pi / 2} \frac{1}{\ell \sin \phi} \sin \phi\; d \phi\Big)=O( \ell^{-1}).
\end{align*}
Since we need the multiplicative constant of the leading terms we first note that:
\begin{align*}
 a^2_{7,\ell}(\phi) =  \frac{8^2}{\lambda_\ell^2 (\lambda_\ell-2)^2} \left[ \gamma_{4,\ell}^2(\phi)+\frac{16}{(1 -4 \alpha_{2,\ell}^2(\phi)/\lambda_\ell^2)^2} \frac{\alpha_{2,\ell}^2(\phi) \beta_{3,\ell}^4(\phi)}{\lambda_\ell^4}+ \frac{8}{1 -4 \alpha_{2,\ell}^2(\phi)/\lambda_\ell^2} \frac{\gamma_{4,\ell}(\phi) \alpha_{2,\ell}(\phi) \beta_{3,\ell}^2(\phi)}{\lambda_\ell^2} \right],
\end{align*}
then, by isolating the leading integral terms with the aid of Lemma \ref{sine2}, Lemma \ref{sine5} and Lemma \ref{sine4}, and finally by computing the integral, we get
\begin{align*}
&\hspace{6cm} \int_{C/\ell}^{  \pi / 2} a^2_{7,\ell}(\phi)  \sin \phi\; d \phi\\
&=8^2\int_{C/\ell}^{  \pi / 2}  \frac{\gamma_{4,\ell}^2(\phi)}{\lambda_\ell^2 (\lambda_\ell-2)^2}    \sin \phi d \phi +O \Big( \int_{C/\ell}^{  \pi / 2}  \frac{1}{\ell^2 \sin^2 \phi}    \sin \phi d \phi  \Big) \\
&=8^2  \int_{C/\ell}^{  \pi / 2} \left\{ \sqrt{\frac{\pi}{2}} \frac{1}{\ell^{1/2} \sin^{1/2} \phi} \left[ \cos \psi^-_\ell -\frac{1}{8 \ell \phi} \cos \psi^+_\ell \right] +O\left( \frac{1}{\ell^{1+1/2} \sin^{1+1/2} \phi}\right) \right\}^2   \sin \phi d \phi + O \left( \int_{C/\ell}^{  \pi / 2} \frac{1}{\ell^2 \sin^2 \phi}    \sin \phi d \phi  \right)\\
&=8^2  \int_{C/\ell}^{  \pi / 2} \left\{ \sqrt{\frac{\pi}{2}} \frac{1}{\ell^{1/2} \sin^{1/2} \phi} \left[ \cos \psi^-_\ell -\frac{1}{8 \ell \phi} \cos \psi^+_\ell \right] \right\}^2   \sin \phi d \phi + O \Big( \int_{C/\ell}^{  \pi / 2} \frac{1}{\ell^2 \sin^2 \phi}    \sin \phi d \phi  \Big)\\
&=8^2  \int_{C/\ell}^{  \pi / 2} \left\{ \frac{2}{\pi} \frac{1}{\ell \sin \phi} \left[ \frac 1 2 + \frac 1 2 \cos(2 \psi^-_\ell)+ \frac{1}{64 \ell^2 \phi^2} \cos^2 \psi^+_\ell - \frac 1{4 \ell \phi} \cos \psi^-_\ell \cos \psi^+_\ell \right] \right\} \sin \phi d \phi + O \Big( \int_{C/\ell}^{  \pi / 2} \frac{1}{\ell^2 \sin^2 \phi}    \sin \phi d \phi  \Big)\\
&= {32} \ell^{-1}+O(\ell^{-2} \log \ell).
\end{align*}
\end{proof}

\section{Nonsingularity of the covariance matrix for $\phi<c/\ell $} \label{AppC}

We only need to show that, after scaling, the determinant of the
matrix $\Sigma_\ell(\psi)$, evaluated for points on the equatorial line $x=(\pi/2,\phi)$, $y=(\pi/2,0)$,  is strictly positive for $c>\psi>0$; for points outside the equator the covariance matrix is obtained by a change of basis: the corresponding matrix does not depend on $\ell$ and  can be easily shown to be full rank.

By expanding the terms of the matrix up to order $12$ around $\psi
=0$, for $\lambda_\ell=\ell(\ell+1)$,
$\lambda_{1,\ell}=(\ell-1)(\ell+1)(\ell+2)$ and
$\lambda_{2,\ell}=(\ell-1) \ell (\ell+1) (\ell+2)$, we have

\begin{align*}
\alpha _{1,\ell }(\psi ) & =\frac{\ell+1}{2\ell}-\frac{
\lambda_{\ell} ( \lambda_{\ell}-2 )   \psi
^{2}}{16\ell^{4}}+\frac{\lambda_{1,\ell} (
\lambda_{\ell}-4 ) \psi ^{4}}{2^7 3 \ell^{5}}
 -\frac{  \lambda_{1,\ell} ( 5\lambda_{\ell} ( \lambda_{\ell}-10 )
+2^3 \cdot  17 ) \psi ^{6}}{2^{11} 3^2 5 \ell^{7}} \\
& +\frac{\lambda_{1,\ell} ( 7\lambda_{\ell} ( \lambda_{\ell} ( \lambda_{\ell}-2 \cdot  3^2 )
+2^3 3 \cdot  5 ) -2^6 \cdot  31 ) \psi ^{8}}{2^{15} 3 ^2 5\cdot  7 \ell^{9}} \\
& -\frac{  \lambda_{1,\ell} ( \lambda_{\ell} ( 21\lambda_{\ell} (
\lambda_{\ell} ( \lambda_{\ell}-2^2 \cdot 7 ) +2^2 \cdot 83 ) -2^4 5 \cdot 499 ) +2^7 691 )
  \psi ^{10}}{2^{18} 3^4 5^27\ell^{11}} \\
& +\frac{\lambda_{1,\ell} ( 11\lambda_{\ell} ( \lambda_{\ell} ( 3\lambda_{\ell} (
\lambda_{\ell} ( \lambda_{\ell}-2^3 5 ) +2^2 3 \cdot 61  ) -2^6 347 ) +2^6 17 \cdot 109)
-2^9 43 \cdot 127 ) \psi ^{12}}{2^{21} 3^5 5^2 7 \cdot 11 \ell^{13}} \\
& +O ( \psi ^{13} )
\end{align*}

\begin{align*}
\alpha _{2,\ell }(\psi ) &=\frac{\ell+1}{2
\ell}-\frac{\lambda_{\ell} (3 \lambda_{\ell}-2) \psi ^2}{2^4
\ell^4}+\frac{\lambda_{\ell} (5 \lambda_{2,\ell}+2^3) \psi ^4}{2^7 3
   \ell^6} -\frac{ \lambda_{\ell}  (7 \lambda_{\ell}  (5 \lambda_{\ell}  (\lambda_{\ell}-2^2 )+2^2 11 )-2^4 17 )  \psi^6}{2^{11} 3^2 5 \ell^8}\\
   &\;\;+\frac{\lambda_{\ell} (3 \lambda_{\ell} (7 \lambda_{\ell} (\lambda_{\ell} (3 \lambda_{\ell}-2^2 5)+2^2 3 \cdot 7)-2^4 89)+2^7 31) \psi ^8}{2^{15}3^2 5 \cdot 7
   \ell^{10}}\\
   &\;\;-\frac{ \lambda_{\ell} (11 \lambda_{\ell}  (\lambda_{\ell}  (3 \cdot 7 \lambda_{\ell}  (\lambda_{\ell}
    (\lambda_{\ell}-2 \cdot 5 )+2^2 17 )-2^3 823 )+2^5 17 \cdot 31 )-2^8 691 )  \psi ^{10}}{2^{18} 3^4 5^2 7
   \ell^{12}}\\
   &\;\;+\frac{\lambda_{\ell} (13 \lambda_{\ell}  (11 \lambda_{\ell}  (\lambda_{\ell}  (3 \lambda_{\ell}  (\lambda_{\ell}
    (\lambda_{\ell}-2 \cdot 7 )+2^2 5 \cdot 7  )-2^3 17 \cdot 23  )+2^6 239 )-2^7 23 \cdot 151 )+2^{10} 43 \cdot 127) \psi ^{12}}{2^{21} 3^5 5^2 7 \cdot 11
   \ell^{14}}\\
   &\;\;+O (\psi ^{13} )
 \end{align*}

 \begin{align*}
\beta_{1,\ell }(\psi ) &=\frac{\lambda_{\ell}  (\lambda_{\ell}-2 )
\psi }{2^3 \ell^4}-\frac{ \lambda_{1,\ell}  (\lambda_{\ell}-2^2  )
\psi ^3}{2^5 3 
   \ell^5}+\frac{\lambda_{1,\ell}  (5 \lambda_{\ell}  (\lambda_{\ell}-2 \cdot 5  )+2^3 17  ) \psi ^5}{2^{10} 3 \cdot 5 \ell^7}\\
   &\;\;-\frac{ \lambda_{1,\ell}  (7 \lambda_{\ell}  (\lambda_{\ell}  (\lambda_{\ell}-2 \cdot 3^2 )+2^3 3 \cdot 5 )-2^6 31 )  \psi ^7}{2^{12} 3^2 5 \cdot 7
   \ell^9}\\
   &\;\;+\frac{\lambda_{1,\ell}  (\lambda_{\ell}  (3 \cdot 7  \lambda_{\ell}  (\lambda_{\ell}
    (\lambda_{\ell}-2^2 7 )+2^2 83 )-2^4 5 \cdot 499 )+2^7 691 ) \psi ^9}{2^{17} 3^4 5 \cdot 7 \ell^{11}}\\
   &\;\;-\frac{ \lambda_{1,\ell}
    (11 \lambda_{\ell}  (\lambda_{\ell}  (3 \lambda_{\ell}  (\lambda_{\ell}
    (\lambda_{\ell}-2^3 5  )+2^2 3 \cdot 61 )-2^6 347 )+2^6 17 \cdot 109 )-2^9 43 \cdot 127 )  \psi ^{11}}{2^{19} 3^4 5^2 7 \cdot 11
   \ell^{13}}\\
   &\;\;+\frac{\lambda_{1,\ell}  (13 \lambda_{\ell}  (11 \lambda_{\ell}  (\lambda_{\ell}  (3 \lambda_{\ell}  (\lambda_{\ell}
    (\lambda_{\ell}-2 \cdot 3^3 )+2^2 \cdot 3^3 \cdot 13 )-2^3 8179 ) ) )) \psi
   ^{13}}{2^{25} 3^5 5^2 7 \cdot 11 \cdot 13 \ell^{15}}\\
    &\;\;+\frac{\lambda_{1,\ell}  (13 \lambda_{\ell}  (11 \lambda_{\ell}  (\lambda_{\ell}  2^7 3 \cdot 1601 )-2^7 3 \cdot 5 \cdot 7 \cdot 2591 )+2^{10} 257 \cdot 3617) \psi
   ^{13}}{2^{25} 3^5 5^2 7 \cdot 11 \cdot 13 \ell^{15}}\\
      &\;\;+O (\psi ^{14} )
 \end{align*}

  \begin{align*}
\beta_{2,\ell }(\psi )&=\frac{\lambda_{\ell}  (\lambda_{\ell}+2 ) \psi }{2^3  \ell^4}-\frac{ \lambda_{\ell}  (\lambda_{\ell}  (\lambda_{\ell}+2 \cdot 3  )-2^3 )
   \psi ^3}{2^5 3  \ell^6}+\frac{\lambda_{\ell}  (\lambda_{\ell}  (5 \lambda_{\ell}  (\lambda_{\ell}+2^2 3 )-2^2 61 )+2^4 17 ) \psi
   ^5}{2^{10} 3 \cdot 5 \ell^8}\\
   &\;\;-\frac{ \lambda_{\ell} (\lambda_{\ell}  (7 \lambda_{\ell}  (\lambda_{\ell}
    (\lambda_{\ell}+2^2 5 )-2^2 41 )+2^4 233 )-2^7 31)  \psi ^7}{2^{12} 3^2 5 \cdot 7 \ell^{10}}\\
   &\;\;+\frac{\lambda_{\ell}  (\lambda_{\ell}
    (\lambda_{\ell}  (2 \cdot 7  \lambda_{\ell}  (\lambda_{\ell}
    (\lambda_{\ell}+2 \cdot 3 \cdot 5 )-2^2 103 )+2^3 6917 )-2^5 3^2 19 \cdot 31 )+2^8 691 ) \psi ^9}{2^{17} 3^4 5 \cdot 7
   \ell^{12}}\\
   &\;\;-\frac{ \lambda_{\ell}  (\lambda_{\ell}  (11 \lambda_{\ell}  (\lambda_{\ell}  (3 \lambda_{\ell}  (\lambda_{\ell}
    (\lambda_{\ell}+2 \cdot 3 \cdot 7 )-2 \cdot 7 \cdot 31 )+2^3 \cdot 7 \cdot 491 )-2^8 3 \cdot 5 \cdot 43)+2^7 5 \cdot 7^2 173 )-2^{10} 43 \cdot 127 )  \psi
   ^{11}}{2^{19} 3^4 5^2 7 \cdot 11 \ell^{14}}\\
   &\hspace{-1cm}+\frac{\lambda_{\ell}  (\lambda_{\ell}  (13 \lambda_{\ell}  (11 \lambda_{\ell}  (\lambda_{\ell}  (3 \lambda_{\ell} (\lambda_{\ell}
    (\lambda_{\ell}+2^3 7)-2^3 \cdot 7 \cdot 29 )+2^7 599)-2^4 5 \cdot 9473 )+2^7 593 \cdot 641 )-2^8 347 \cdot 20947 )+2^{11} 257 \cdot 3617 )
   \psi ^{13}}{2^{25} 3^5 5^2 7 \cdot 11 \cdot 13 \ell^{16}}\\
   &\:\:+O (\psi ^{14} )
 \end{align*}

  \begin{align*}
\beta_{3,\ell }(\psi )&=\frac{(\ell (3 \ell (\ell+2)+1)-2) \psi }{2^3 \ell^3}-\frac{\lambda_{\ell} (5 \lambda_{2,\ell}+2^3) \psi ^3}{2^5 3  \ell^6}
   +\frac{\lambda_{\ell}  (7 \lambda_{\ell}  (5 \lambda_{\ell}  (\lambda_{\ell}-2^2 )+2^2 11 )-2^4 17 ) \psi ^5}{2^{10} 3 \cdot 5 \ell^8}\\
   &\;\;-\frac{\lambda_{\ell} (3 \lambda_{\ell} (7 \lambda_{\ell} (\lambda_{\ell} (3 \lambda_{\ell}-2^2 5)+2^2 3 \cdot 7)-2^4 89)+2^7 31) \psi ^7}{2^{12} 3^2 5 \cdot 7 \ell^{10}}\\
   &\;\;+\frac{\lambda_{\ell}  (11 \lambda_{\ell}  (\lambda_{\ell}
    (3 \cdot 7 \lambda_{\ell}  (\lambda_{\ell}  (\lambda_{\ell}-5 \cdot 5 )+2^2 17)-2^3 823 )+2^5 17 \cdot 31 )-2^8 691 ) \psi
   ^9}{2^{17} 3^4 5 \cdot 7 \ell^{12}}\\
   &\;\;-\frac{ \lambda_{\ell} (13 \lambda_{\ell}  (11 \lambda_{\ell}  (\lambda_{\ell}  (3 \lambda_{\ell}  (\lambda_{\ell}  (\lambda_{\ell}-2 \cdot 7 )+2^2  5 \cdot 7 )-2^{3} 17 \cdot 23 )+2^6 239 )-2^7 23 \cdot 151)+2^{10} 43 \cdot 127 )  \psi
   ^{11}}{2^{19} 3^4 5^2 7 \cdot 11 \ell^{14}}\\
   &\hspace{-1.5cm}+\frac{\lambda_{\ell} (3 \lambda_{\ell} (13 \lambda_{\ell} (11 \lambda_{\ell} (\lambda_{\ell} (\lambda_{\ell} (5 \lambda_{\ell} (3 \lambda_{\ell}-2^3 7)+2^3 3 \cdot 7 \cdot 23)-2^8 3 \cdot 5 \cdot 11)+2^4 20887)-2^7 3 \cdot 48409)+2^8 3 \cdot 846457)-2^{11}257 \cdot 3617) \psi ^{13}}{2^{25} 3^5 5^2 7 \cdot 11 \cdot 13 \ell^{16}}\\
   &\;\;+O (\psi^{14} )
 \end{align*}

  \begin{align*}
\gamma_{1,\ell }(\psi )&=\frac{\ell (3 \ell (\ell+2)+1)-2}{2^3 \ell^3}-\frac{\lambda_{\ell} (\lambda_{2,\ell}+2^3) \psi^2}{2^5 \ell^6}
   +\frac{\lambda_{\ell}  (3 \lambda_{\ell}
    (\lambda_{\ell}  (\lambda_{\ell}-2^2 )+2^2 3 \cdot 5 )-2^4 17 ) \psi ^4}{2^{10} 3 \ell^8}\\
   &\;\;-\frac{\lambda_{\ell} (\lambda_{\ell} (\lambda_{\ell} (\lambda_{\ell}
   (3 \lambda_{\ell}-2^2 5 )+2^2 3 \cdot 5 \cdot 11)-2^4 199)+2^7 31) \psi ^6}{2^{12} 3^2 5 \ell^{10}}\\
   &\;\;+\frac{\lambda_{\ell}  (\lambda_{\ell}  (\lambda_{\ell}  (7 \lambda_{\ell}
    (\lambda_{\ell}  (\lambda_{\ell}-2 \cdot 5 )+2^2 5 \cdot 29 )-2^3 31 \cdot 163 )+2^5 5 \cdot 31^2 )-2^8 691 ) \psi ^8}{2^{17} 3^2 5 \cdot 7
   \ell^{12}}\\
   &\;\;-\frac{ \lambda_{\ell} (3 \lambda_{\ell}  (\lambda_{\ell}  (\lambda_{\ell}  (3 \lambda_{\ell}  (\lambda_{\ell}
    (\lambda_{\ell}-2 \cdot 7 )+2^2 3^2 5 \cdot 7 )-2^3 61 \cdot 131 )+2^6 7^2 157 )-2^7 47 \cdot 281 )+2^{10} 43 \cdot 127  ) \psi
   ^{10}}{2^{19} 3^4 5^2 7 \ell^{14}}\\
   &\hspace{-2cm}+\frac{\lambda_{\ell} (\lambda_{\ell} (\lambda_{\ell} (11 \lambda_{\ell} (3 \lambda_{\ell} (\lambda_{\ell} (\lambda_{\ell} (3 \lambda_{\ell}-2^3 7)+2^3 3 \cdot 7 \cdot 43)-2^8 \cdot 5^2 \cdot 29)+2^4 5^4 7 \cdot 103)-2^7 3^2 307 \cdot 1543)+2^8 5 \cdot 47 \cdot 29443)-2^{11} 257 \cdot 3617) \psi ^{12}}{2^{25} 3^5 5^2 7 \cdot 11 \ell^{16}}\\
   &\;\;+O (\psi^{13} )
 \end{align*}

  \begin{align*}
\gamma_{2,\ell }(\psi )&=\frac{\lambda_{\ell}  (\lambda_{\ell}-2 )}{2^3 \ell^4}-\frac{ \lambda_{1,\ell}  (\lambda_{\ell}-2^2 )  \psi ^2}{2^5
   \ell^5}+\frac{\lambda_{1,\ell}  (5 \lambda_{\ell}  (\lambda_{\ell}-2 \cdot 5 )+2^3 17 ) \psi ^4}{2^{10} 3 \ell^7}\\
   &\;\;-\frac{   \lambda_{1,\ell}  (7 \lambda_{\ell}  (\lambda_{\ell}  (\lambda_{\ell}-2 \cdot 3^2 )+2^3 3 \cdot 5 )-2^6 31  ) \psi ^6}{2^{12} 3^2 5
   \ell^9}\\
   &\;\;+\frac{\lambda_{1,\ell}  (\lambda_{\ell}  (3 \cdot 7 \lambda_{\ell}  (\lambda_{\ell}
    (\lambda_{\ell}-2^2 7 )+2^2 83 )-2^4 5 \cdot 499 )+2^7 691 ) \psi ^8}{2^{17} 3^2 5 \cdot 7 \ell^{11}}\\
   &\;\;-\frac{ \lambda_{1,\ell}
    (11 \lambda_{\ell}  (\lambda_{\ell}  (3 \lambda_{\ell}  (\lambda_{\ell}
    (\lambda_{\ell}-2^3 5 )+2^2 3 \cdot 61 )-2^2 3 \cdot 61 )+2^6 17 \cdot 109 )-2^9 43 \cdot 127 )  \psi ^{10}}{2^{19} 3^4 5^2 7
   \ell^{13}}\\
   &\;\;+\frac{\lambda_{1,\ell}  (13 \lambda_{\ell}  (11 \lambda_{\ell}  (\lambda_{\ell}  (3 \lambda_{\ell}  (\lambda_{\ell}
    (\lambda_{\ell}-2 \cdot 3^3)+2^2 3^3 13  )-2^3 8179 )+2^7 3 \cdot 1601 )-2^7 3 \cdot 5 \cdot 7 \cdot 2591 )+2^{10} 257 \cdot 3617 ) \psi
   ^{12}}{2^{25} 3^5 5^2 \cdot 7 \cdot 11 \ell^{15}}\\
   &\;\;+O (\psi ^{13} )
 \end{align*}

  \begin{align*}
\gamma_{3,\ell }(\psi )&=\frac{\lambda_{\ell}  (\lambda_{\ell}+2 )}{2^3 \ell^4}-\frac{ \lambda_{\ell}  (\lambda_{\ell}  (\lambda_{\ell}+2 \cdot 3 )-2^3 ) \psi
   ^2}{2^5 \ell^6}+\frac{\lambda_{\ell}  (\lambda_{\ell}  (5 \lambda_{\ell}  (\lambda_{\ell}+2^2 3 )-2^2 61 )+2^4 17 ) \psi ^4}{2^{10} 3
   \ell^8}\\
   &\;\;-\frac{ \lambda_{\ell}  (\lambda_{\ell}  (7 \lambda_{\ell}  (\lambda_{\ell}
    (\lambda_{\ell}+2^2 5 )-2^2 41 )+2^4 233 )-2^7 31 )  \psi ^6}{2^{12} 3^2 5 \ell^{10}}\\
   &\;\;+\frac{\lambda_{\ell}  (\lambda_{\ell}
    (\lambda_{\ell}  (3 \cdot 7 \lambda_{\ell}  (\lambda_{\ell}
    (\lambda_{\ell}+2 \cdot 3 \cdot 5 )-2^2 103 )+2^3 6917 )-2^5 3^2 19 \cdot 31 )+2^8 691 ) \psi ^8}{2^{17} 3^2 5 \cdot 7
   \ell^{12}}\\
   &\;\;-\frac{ \lambda_{\ell}  (\lambda_{\ell}  (11 \lambda_{\ell}  (\lambda_{\ell}  (3 \lambda_{\ell}  (\lambda_{\ell}
    (\lambda_{\ell}+2 \cdot 3 \cdot 7 )-2^2 7 \cdot 31 )+2^2 7 \cdot 31 )-2^8 3 \cdot 5 \cdot 43 )+2^7 5 \cdot 7^2 173 )-2^{10} 43 \cdot 127 ) \psi
   ^{10}}{2^{19} 3^4 5^2 7 \ell^{14}}\\
   &\hspace{-1cm}+\frac{\lambda_{\ell}  (\lambda_{\ell}  (13 \lambda_{\ell}  (11 \lambda_{\ell}  (\lambda_{\ell}  (3 \lambda_{\ell}
    (\lambda_{\ell}
    (\lambda_{\ell}+2^3 7 )-2^3 7 \cdot 29 )+2^7 599 )-2^4 5 \cdot 9473 )+2^7 593 \cdot 641 )-2^8 347 \cdot 20947 )+2^{11} 257 \cdot 3617 )
   \psi ^{12}}{2^{25} 3^5 5^2 7 \cdot 11 \ell^{16}}\\
   &\;\;+O (\psi ^{13} )
 \end{align*}

  \begin{align*}
\gamma_{4,\ell }(\psi )&=\frac{\ell (3 \ell (\ell+2)+1)-2}{2^3 \ell^3}-\frac{\lambda_{\ell} (5 \lambda_{2,\ell}+2^3) \psi ^2}{2^5 \ell^6}
+\frac{\lambda_{\ell}  (7 \lambda_{\ell}
    (5 \lambda_{\ell}  (\lambda_{\ell}-4 )+2^2 11 )-2^4 17 ) \psi ^4}{2^{10} 3 \ell^8}\\
   &\;\;-\frac{\lambda_{\ell} (3 \lambda_{\ell} (7 \lambda_{\ell} (\lambda_{\ell} (3 \lambda_{\ell}-2^2 5)+2^2 3 \cdot 7)-2^4 89)+3968) \psi ^6}{2^{12} 3^2 5 \ell^{10}}\\
   &\;\;+\frac{\lambda_{\ell}  (11 \lambda_{\ell}  (\lambda_{\ell}  (3 \cdot 7 \lambda_{\ell} (\lambda_{\ell}  (\lambda_{\ell}-2 \cdot 5  )+2^2 17 )-2^3 823 )+2^5 17 \cdot 31 )-2^8 691 ) \psi^8}{2^{17} 3^2 5 \cdot 7
   \ell^{12}}\\
   &\;\;-\frac{ \lambda_{\ell}  (13 \lambda_{\ell}  (11 \lambda_{\ell}  (\lambda_{\ell}  (3 \lambda_{\ell}  (\lambda_{\ell}
    (\lambda_{\ell}-2 \cdot 7 )+2^{2} 5 \cdot 7 )-2^3 \cdot 17 \cdot 23)+2^6 239 )-2^{7} 23 \cdot 151 )+2^{10} 43 \cdot 127 ) \psi
   ^{10}}{2^{19} 3^4 5^2 7 \ell^{14}}\\
   &\hspace{-1.5cm}+\frac{\lambda_{\ell} (3 \lambda_{\ell} (13 \lambda_{\ell} (11 \lambda_{\ell} (\lambda_{\ell} (\lambda_{\ell} (5 \lambda_{\ell} (3 \lambda_{\ell}-2^3 7)+2^3 3 \cdot 7 \cdot 23 )-2^8 3 \cdot 5 \cdot 11 )+2^4 20887)-2^7 3 \cdot 48409)+2^8 3 \cdot 846457)-2^{11} 257 \cdot 3617 ) \psi ^{12}}{2^{25} 3^5 5^2 7 \cdot 11 \ell^{16}}\\
   &\;\;+O (\psi
   ^{13} )
 \end{align*}
It should be noted that, as $\ell \rightarrow \infty$, all coefficients
converge to constants; more importantly, the constants involved in the $O$-notation
for all the $O ( \psi ^{13}) $ terms are universal. A computer-oriented computation
yields the following Taylor expansion for the determinant:
\begin{align*}
\text{det}(\Sigma _{\ell }(\psi ))&=\frac{
(\ell-4)(\ell-3)^{3}(\ell+1)^{13}(\ell+4)^{3}(\ell+5)\left( \ell^{2}+\ell-6\right) ^{5}\left(
\ell^{2}+\ell-2\right) ^{7}\psi ^{26}}{2^{57}\; 3^8\; 5^2 \; \ell^{45}}+O( \psi
^{27})\\
&=\psi^{26} (1+O(\ell^{-1}))+O(\psi^{27})>0
\end{align*}%
the inequality holding for $c$ sufficiently small, because by the above, the $O(\psi^{27})$ term is universal.

\end{appendices}

\newpage

\end{document}